\documentclass[10pt,english]{article}

\usepackage[T1]{fontenc}
\usepackage[latin9]{inputenc}
\usepackage[a4paper]{geometry}
\usepackage{color}
\usepackage{babel}
\usepackage{array}
\usepackage{textcomp}
\usepackage{mathtools}
\usepackage{multirow}
\usepackage{amsmath}
\usepackage{amsthm}
\usepackage{amssymb}
\usepackage{graphicx}
\usepackage{setspace}
\usepackage[unicode=true,pdfusetitle,
 bookmarks=true,bookmarksnumbered=false,bookmarksopen=false,
 breaklinks=false,pdfborder={0 0 1},backref=false,colorlinks=false]
 {hyperref}

\makeatletter

\providecommand{\tabularnewline}{\\}

\theoremstyle{plain}
\newtheorem{thm}{\protect\theoremname}
\theoremstyle{plain}
\newtheorem{prop}[thm]{\protect\propositionname}
\theoremstyle{plain}
\newtheorem{cor}[thm]{\protect\corollaryname}
\theoremstyle{plain}
\newtheorem{lem}[thm]{\protect\lemmaname}

\usepackage{enumitem}

\usepackage{colortbl}
\definecolor{lg}{gray}{0.8}

\usepackage[labelfont=bf]{caption}

\newcommand{\PVSnew}{\color{black}}
\newcommand{\PVSedit}{\color{black}}

\usepackage[bottom]{footmisc}

\usepackage[sort&compress]{natbib}
\usepackage[running, mathlines]{lineno}

\numberwithin{thm}{section}
\numberwithin{equation}{section}
\numberwithin{table}{section}
\numberwithin{figure}{section}

\newtheoremstyle{boldremark}
  {}
  {}
  {}
  {}
  {\bfseries}
  {.}
  { }
  {}%

\theoremstyle{boldremark}
\newtheorem{brem}{Remark} 
\numberwithin{brem}{section}

\newtheoremstyle{assumption}
  {}
  {}
  {}
  {}
  {\bfseries}
  {.}
  { }
  {}%

\theoremstyle{assumption}
\newtheorem{assumption}{Assumption} 
\numberwithin{assumption}{section}

\usepackage[title]{appendix}
\usepackage{etoolbox}
\patchcmd{\appendices}{\quad}{: }{}{}

\usepackage{mathtools} 

\@ifundefined{showcaptionsetup}{}{%
 \PassOptionsToPackage{caption=false}{subfig}}
\usepackage{subfig}
\makeatother

\providecommand{\corollaryname}{Corollary}
\providecommand{\lemmaname}{Lemma}
\providecommand{\propositionname}{Proposition}
\providecommand{\theoremname}{Theorem}
\providecommand{\remarkname}{Remark}

\begin{document}
\title{Smart leverage? Rethinking the role of Leveraged Exchange Traded Funds
in constructing portfolios to beat a benchmark}
\author{Pieter M. van Staden\thanks{National Australia Bank, Melbourne, Victoria, Australia 3000. The
research results and opinions expressed in this paper are solely those
of the authors, are not investment recommendations, and do not reflect
the views or policies of the NAB Group. \texttt{pieter.vanstaden@gmail.com}} \and Peter A. Forsyth\thanks{Cheriton School of Computer Science, University of Waterloo, Waterloo
ON, Canada, N2L 3G1, \texttt{paforsyt@uwaterloo.ca}} \and Yuying Li\thanks{Cheriton School of Computer Science, University of Waterloo, Waterloo
ON, Canada, N2L 3G1, \texttt{yuying@uwaterloo.ca}} }
\date{\today}
\maketitle
\begin{abstract}
Leveraged Exchange Traded Funds (LETFs), while extremely controversial
in the literature, remain stubbornly popular with both institutional
and retail investors in practice. While the criticisms of LETFs are
certainly valid, we argue that their potential has been underestimated
in the literature due to the use of very simple investment strategies
involving LETFs. In this paper, we systematically investigate the
potential of including a broad stock market index LETF in long-term,
dynamically-optimal investment strategies designed to maximize the
outperformance over standard investment benchmarks in the sense of
the information ratio (IR). Our results exploit the observation that
positions in a LETF deliver call-like payoffs, so that the addition
of a LETF to a portfolio can be a convenient way to add inexpensive
leverage while providing downside protection. Under stylized assumptions,
we present and analyze closed-form IR-optimal investment strategies
using either a LETF or standard/vanilla ETF (VETF) on the same equity
index, which provides the necessary intuition for the potential and
benefits of LETFs. In more realistic settings of infrequent trading,
leverage restrictions and additional constraints, we use a neural
network-based approach to determine the IR-optimal strategies,
trained on bootstrapped historical data. We
find that IR-optimal strategies with a broad stock market LETF are
not only more likely to outperform the benchmark than IR-optimal strategies
derived using the corresponding VETF, but are able to achieve partial
stochastic dominance over the benchmark and VETF-based strategies
in terms of terminal wealth, even if investment in the VETF can be
leveraged. Our results help to explain the empirical appeal of LETFs
to investors, and encourage the reconsideration in academic research
of the role of broad stock market LETFs within the context of more
sophisticated investment strategies.
\end{abstract}
\medskip{}

\medskip{}


\section{Introduction\label{sec:Introduction}}

Leveraged Exchange Traded Funds (LETFs) are exchange traded funds
(ETFs) with the stated objective of replicating some multiple $\beta$
of the daily returns of their underlying reference assets/indices
before costs, where values of $\beta$ of $+2$, $+3$, $-2$ and
$-3$ are commonly used. In contrast, standard/vanilla ETFs (VETFs)
aim simply to replicate the returns of their underlying assets/indices
before costs (i.e. $\beta=1$).

A review of the academic literature suggests that incorporating LETFs
in investment strategies are commonly regarded with at least some
suspicion, if not outright distrust. ``\textit{Just say no to leveraged
ETFs}'', the title of a recent article (\cite{BednarekPatel2022}),
provides perhaps the most succinct summary of the broadly negative
view of LETFs that permeates the literature. There are certainly good
reasons for these negative perceptions of LETFs. A common criticism
in the literature focuses on the ``compounding'' effect of LETF
returns, which arises since a LETF returning $\beta$ times the \textit{daily}
returns of the underlying index obviously does not imply that the
LETF also returns  $\beta$ times the \textit{quarterly}
returns of the underlying index. This observation, along with the
time decay and volatility decay \textcolor{black}{of LETF positions,
results in the potential wealth-destroying effects of LETF investments
which increase with the holding time horizon (\cite{Mackintosh2008},\cite{Carver2009},
\cite{Sullivan2009}, \cite{CharupatMiu2011}). The poor investment
outcomes using LETFs reported in the literature should therefore not
come as a surprise if an investor uses a basic buy-and-hold investment
strategy (}\cite{CharupatMiu2011}, \cite{BednarekPatel2022}, \cite{LeungSantoli2012}\textcolor{black}{),
or very simple }(if not outright naive) portfolio rebalancing rules
typically considered in conjunction with unrealistically frequent
rebalancing\footnote{For example, daily rebalancing. Note that in some cases, the general
investment literature actually advises investors not to hold LETFs
for longer than a single trading session. See for example Forbes Advisor,
accessed 10 March 2024. Michael Adams. \textit{Eight best leveraged
ETFs of March 2024}. https://www.forbes.com/advisor/investing/best-leveraged-etfs.} from the perspective of long-term investors (\cite{ChenMadhavan2009},
\cite{AvellanedaZhang2010}, \cite{BansalMarshall2015}, \cite{DeVaultEtAl2021})\textcolor{black}{.
}While there are a few studies observing that LETFs could have a role
within diversified portfolios, especially where the investor might
have relatively aggressive performance targets (\cite{BansalMarshall2015},
\cite{HillFoster2009}), wish to circumvent onerous leverage restrictions
or large margin rates on borrowing (\cite{DeVaultEtAl2021}),  {\PVSedit{or
wish to outperform broad market indices using very simple strategies
with frequent rebalancing (\citet{Knapp2023,TrainorEtAl2018})}},
these perspectives remain the exception to the mainstream academic
view and tend to leave questions regarding the formulation of practical
investment strategies unanswered.

However, the contrast between the general perceptions in the academic
literature and investment practice could not be more profound. LETFs
have consistently remained incredibly popular financial products since
their introduction in 2006, as emphasized by recent headlines such
as ``\textit{Investors pump record sums into leveraged ETFs}'' (Financial
Times, November 2022\footnote{Financial Times, November 14, 2022. Steven Johnson. \textit{Investors
pump record sums into leveraged ETFs}. https://www.ft.com/content/b98ab360-2506-44f2-8e08-9d434df5f15d}) and ``\textit{Retail investors snap up triple-leveraged US equity
ETFs}'' (Financial Times, May 2024\footnote{Financial Times, May 4, 2024. George Steer and Will Schmitt. \textit{Retail
investors snap up triple-leveraged US equity ETFs}. https://on.ft.com/3WsWTom}). LETFs consistently dominate the top 10 most popular ETFs listed
on US exchanges\footnote{For example, as at 9 October 2023, four out of the five most popular
ETFs as measured by the average daily trading volume over the preceding
three months were LETFs. https://etfdb.com/compare/volume. Accessed
9 October 2023}. 

Perhaps more significantly, LETFs also enjoy substantial popularity
among institutional investors. Analyzing the quarterly reports by
institutional investment managers with at least US\$100 million in
assets under management that were filed with the SEC from September
2006 to December 2016, \cite{DeVaultEtAl2021} finds that more than
20\% of the reports reference at least one LETF in the end-of-quarter
portfolio allocation. 

Leaving speculative trading aside, what could explain the appeal of
LETFs for institutional investors? Suppose an investor wants to leverage
returns in a cost-effective way which also offers some downside protection.
Since the requirement of downside protection rules out simple leverage,
such an institutional investor has broadly speaking at least two options,
namely (a) engage with a hedge fund or fund manager to use for example
managed futures strategies, or (b) follow a dynamic trading strategy
using for example LETFs as discussed in this paper. Since the expense
ratios of LETFs range typically between 80 and 150 basis points, whereas
standard leveraged positions are subject to substantial margin rates
of borrowing which can easily exceed 5\% for smaller institutional
investors even during periods of low interest rates (\cite{DeVaultEtAl2021})
while hedge funds charge hefty management fees, LETFs are certainly
cost effective. In addition, LETFs can offer great upside returns
in combination with limited liability without the need to manage short
positions, and as discussed in this paper, positions in LETFs can
be \textit{infrequently} rebalanced while still obtaining competitive
investment outcomes relative to standard investment benchmarks.

\subsection{Main Contributions}
While the academic literature offers a sophisticated and
careful treatment of optimal rebalancing for hedging purposes in the
case of LETFs (\cite{DaiEtAl2023}), or optimal replication policies
for LETF construction (\cite{GuasoniMayerhofer2023}), in this paper
we therefore aim to make progress towards closing the observed gap
between the literature and investment practice by showing that there
is a role for LETFs in \textit{infrequently rebalanced} (e.g. quarterly
rebalanced) portfolios designed for \textit{long-term} institutional
or retail investors wishing to outperform some investment benchmark.

{\PVSedit{
Since LETFs are a relatively recent invention, historical LETF returns
since 1926 are obtained by constructing a proxy LETF replicating\footnote{We make the standard assumption (see e.g. \citet{BansalMarshall2015},
\citet{LeungSircar2015}) that the LETF managers do not have challenges
in replicating the underlying index. Note that given improvements
in designing replication strategies for LETFs that remain robust even
during periods of market volatility (see for example \citet{GuasoniMayerhofer2023}),
this appears to be a reasonable assumption for ETFs (VETFs and LETFs)
written on major equity market indices as considered in this paper.} $\beta=2$ times the daily returns of a broad US equity market index,
deflating the returns by a typical LETF expense ratio and interest
rates (see Appendix \ref{sec:Appendix - Source-data}). All returns
time series are inflation-adjusted to reflect real returns.

To ensure the practical relevance of our conclusions, illustrative
investment results are based on data sets generated using (i) stochastic
differential equations calibrated to historical data since 1926 for closed-form
solutions and (ii) block bootstrap resampling of historical data since 1926 (\citet{politis1994,CogneauZakalmouline2013,Cederburg_2022})
for neural network-based numerical solutions. Note that the bootstrap
resampling of historical data captures \textit{all} empirical qualities of actual
returns, including potentially sophisticated volatility dynamics,
which may not be reflected in the stylized
settings of stochastic differential equations.
}}

In more detail, the main contributions of this paper are as follows:

\begin{enumerate}

\item We construct dynamic (multi-period) investment strategies which
maximize the information ratio (IR) of the active portfolio manager
(or simply ``investor'') relative to standard investment benchmarks
using a LETF or a VETF on the same underlying equity index as well
as bonds. The IR is chosen due to its popularity in investment practice
(\cite{HassineRoncalli2013,BolshakovChincarini2020,IsraelsenCogswell2006,BajeuxEtAl2013}),
so that the results are not just of academic interest but also of
practical relevance to institutional investors. 

\item Under stylized assumptions including parametric dynamics for
the underlying assets, we present closed-form IR-optimal dynamic investment
strategies which enable us to obtain intuition regarding the expected
behavior of IR-optimal investment strategies in more general settings.
Note that the closed-form solutions allow for jump-diffusion dynamics
of the equity index, which as noted above is crucial to consider in
the case of LETFs, so that our results contributes to the existing
literature which is almost exclusively based on diffusion dynamics
(see for example \cite{Giese2010}, \cite{Jarrow2010} \cite{LeungSantoli2012},
\cite{LeungEtAl2016}, \cite{LeungSircar2015}, \cite{Wagalath2014},
\cite{GuasoniMayerhofer2023}). However, in the context of $\mathbb{Q}$ measure option
pricing, \citet{Ahn_2015} consider jump processes  for LETFs.

\item {\PVSedit{Relaxing the stylized assumptions to allow for more general and practical
conclusions, we implement a data-driven neural network approach to obtain optimal investment strategies using
stationary block bootstrap resampled historical data (including
proxy LETF returns) since 1926. This investment setting considers multi-asset
portfolios, infrequent rebalancing, as well as multiple investment
constraints including leverage restrictions and borrowing premiums.
We analyze the dynamic IR-optimal investment strategies for different
scenarios, including: (i) investing in a VETF and bonds but with no
leverage allowed, (ii) investing in a VETF and bonds with different
levels of maximum leverage allowed and different levels of borrowing
premiums being applicable, and (iii) investing in a LETF on the same
equity market index as well as bonds with no leverage being allowed.}}

\item We find that IR-optimal investment strategies involving LETFs
are fundamentally \textit{contrarian}. This finding aligns to the
empirical asset allocation behavior observed by \cite{DeVaultEtAl2021}
in their analysis of the SEC filings by institutional fund managers,
whereby managers seem to decrease their holdings in LETFs after observing
strong recent investment performance. In terms of investment performance,
we find that IR-optimal strategies including the LETF are not only
more likely to outperform the benchmark than IR-optimal strategies
derived using the corresponding VETF, but are able to achieve partial
stochastic dominance over the investment benchmark in terms of portfolio
value\footnote{
For a definition of partial stochastic dominance, see
\citet{Atkinson1987,PvSDangForsyth2019_Distributions}
} (wealth).

\end{enumerate}

Our results therefore
encourage the reconsideration of the role of broad equity market LETFs
within more sophisticated dynamic investment strategies, and provide
a potential additional motivation regarding the enduring popularity
of LETFs observed in practice.

\subsection{Intuition}
In this section, we provide some insight into 
the potential advantages of including LETFs in optimal dynamic asset allocation.
We will give an overview here, leaving
the  technical details to
Section \ref{subsec:Intuition:-lump-sum investment scenario}.

Suppose an investor allocates their initial wealth $W\left(0\right)$
to US 30-day T-bills and either a LETF ($\beta=2$) or a VETF on a
broad US equity market index $S$ at time $t_{0}=0$, and does not
rebalance the portfolio over the holding time horizon $\Delta t>0$.
We discuss in more detail in Section \ref{sec:Closed-form-solutions} how the LETF
and VETF can be viewed as derivative contracts on underlying $S$ costing
$F_{\ell}\left(0\right)$ (LETF) and $F_{v}\left(0\right)$ (VETF)
to purchase,
with payoffs of $F_{\ell}\left(\Delta t\right)$ 
and $F_{v}\left(\Delta t\right)$,
respectively. 

Assume that the underlying index price follows a geometric Brownian motion, 
which implies that the price is always positive.  
In addition, assume a constant risk-free rate for short term bonds.
It can be easily shown that (see e.g. \citep{AvellanedaZhang2010})
\begin{eqnarray}
\frac{F_{\ell}\left(\Delta t\right)}{F_{\ell}\left(0\right)} & = & \exp\left\{ -c_{\ell}\cdot\Delta t\right\} \cdot f_{\ell}\left(\Delta t;\beta\right)
\cdot\left(\frac{S\left(\Delta t\right)}{S\left(0\right)}\right)^{\beta},\label{eq: bsLETF *only* ratio discrete rebal}
\end{eqnarray}
where 
\begin{equation}
f_{\ell}\left(\Delta t;\beta\right)=\exp\left\{ -\left[\left(\beta-1\right)r+\frac{1}{2}\left(\beta-1\right)\beta\sigma^{2}\right]\cdot\Delta t\right\} 
\label{eq: bsf_ell and Y_ell_tilde}
\end{equation}
$c_{\ell}>0$ is the expense ratio, $r$ is the risk-free interest and $\sigma$ is the volatility of return.
In other words, payoff $F_{\ell}\left(\Delta t\right)$ is a 
deterministic function of the terminal value of the underlying index.
Of course, $S(\Delta t)$ itself is stochastic.

Since small values of maturity $\Delta t$ can be undesirable
due to frequent trading, and large values of $\Delta t$ emphasize
the time- and volatility decay of simply holding the LETF $F_{\ell}$,
suppose the investor chooses a convenient maturity of $\Delta t=0.25$
years (one quarter). Figure \ref{fig: Intuition - no jumps}
illustrates the payoff diagrams
for the investor's wealth at maturity $W\left(\Delta t\right)$ under
different combinations of T-bills and an ETF, as a function of the
value at maturity $S\left(\Delta t\right)$ 
of the underlying equity index.
Similar payoff diagrams can be seen in \citet{Knapp2023} and for the unleveraged
case in \citet{Bertrand2022}.

For simplicity, we assume parametric asset dynamics for the T-bills
and index $S$ ($S$ follows geometric Brownian motion) in Figure \ref{fig: Intuition - no jumps}
calibrated to US market data over 1926 to 2023. We impose realistic ETF expense
ratios, and a borrowing premium of 3\% over the T-bill rate for short
positions. Note that the assumption of parametric asset dynamics is
for purposes of intuition only, since the investment results in Section
\ref{sec:Indicative-investment-results} do not rely on any parametric
assumptions. Leaving rigorous derivation for subsequent sections,
we make the following qualitative observations regarding the payoff
diagrams for purposes of intuition:
\begin{itemize}
\item Figures \ref{fig: Intuition - no jumps} 
illustrate that we can characterize the payoffs of LETF investments
as \textit{call-like}. This suggests that the addition of an LETF
to a portfolio can be a useful way to add inexpensive leverage while
preserving downside protection, much like a usual call option. Provided
that the investment in the LETF is itself not funded by borrowing
(i.e. the LETF position itself is not leveraged), the LETF payoff
is always non-negative due to limited liability even in the case of
significant drops in the value of the underlying equity index, in
contrast with leveraged VETF positions. 
\item For calibrated geometric Brownian motion (GBM) dynamics
for the equity index $S$, Figure \ref{fig: Intuition - no jumps}
illustrates that investing all wealth in the LETF with $\beta=2$
(Figure \ref{fig: Intuition - no jumps}(a)) dominates the 2x leveraged
investment in the VETF (200\% of wealth in VETF funded by borrowing
and amount equal to 100\% of wealth) almost everywhere, but underperforms
a 100\% investment in the VETF for negative underlying index returns
(i.e. when $S\left(\Delta t\right)/S\left(0\right)<1$). By contrast,
investing 50\% of wealth in the LETF (Figure \ref{fig: Intuition - no jumps}(a))
and the remaining 50\% in T-bills dominates the payoff of investing
100\% of wealth in the VETF almost everywhere, but underperforms a
2x leveraged investment in the VETF for positive underlying index
returns ($S\left(\Delta t\right)/S\left(0\right)>1$). 
Note that under
GBM dynamics for $S$, the terminal wealth $W\left(\Delta t\right)$
conditional on the terminal value $S\left(\Delta t\right)$ is deterministic
for both VETF and LETF investments (see Section \ref{subsec:Intuition:-lump-sum investment scenario}).

\end{itemize}

\noindent 
\begin{figure}[!tbh]
\noindent \begin{centering}
\subfloat[LETF: 100\%, T-bills: 0\% vs. VETF combinations]{\includegraphics[scale=0.7]{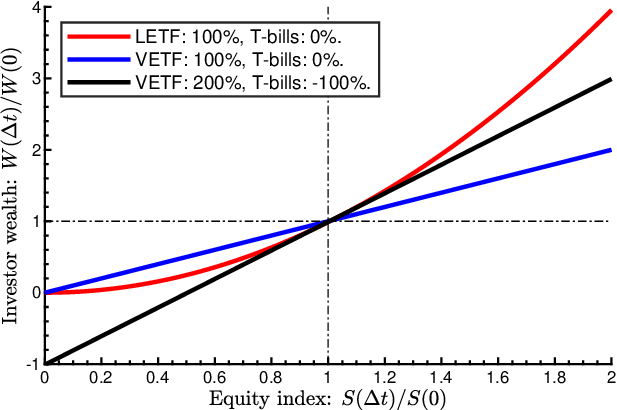}

}$\quad$$\quad$\subfloat[LETF: 50\%, T-bills: 50\% vs. VETF combinations]{\includegraphics[scale=0.7]{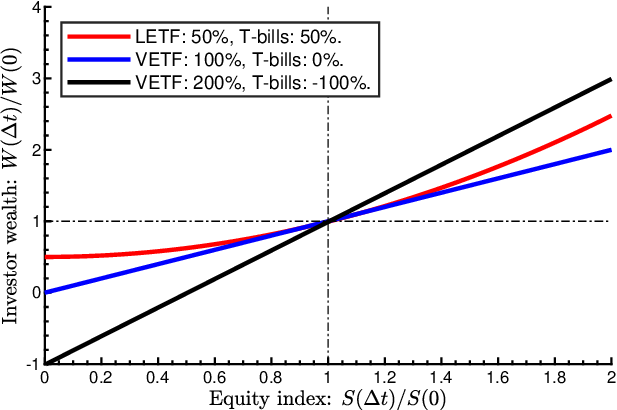}

}
\par\end{centering}
\caption{Payoffs when equity market index $S$ follows calibrated GBM dynamics:
Investor wealth gross return $W\left(\Delta t\right)/W\left(0\right)$
as a function of underlying equity index gross return $S\left(\Delta t\right)/S\left(0\right)$,
$\Delta t=0.25$ (1 quarter), for different proportions of initial
wealth $W\left(0\right)$ invested in the LETF, VETF and T-bills at
time $t_{0}=0$. Asset parameters are calibrated to US equity and
bond market data over the period 1926:01 to 2023:12 (Appendix \ref{sec:Appendix - Source-data}),
LETF and VETF expense ratios are assumed to be 0.89\% and 0.06\% respectively,
and a borrowing premium of 3\% over the T-bill rate is applicable
to short positions. See Section \ref{subsec:Intuition:-lump-sum investment scenario}
for a rigorous treatment of the illustrated relationships. \label{fig: Intuition - no jumps}}
\end{figure}

 When the underlying index is modelled by a jump process, due to limited liability,  a correction to the index price is needed so that LETF remains positive, see Section \ref{subsec:Intuition:-lump-sum investment scenario} for a more detailed discussion. This makes the payoff relation 
between LETF and the underlying index stochastic.
While qualitatively
similar observations as in the case of no jumps (Figure \ref{fig: Intuition - no jumps})
apply to the \textit{median} payoffs of the LETF investments,
allowing for
jumps in the underlying asset dynamics can affect the LETF payoff
significantly, and jumps are therefore critical to incorporate in
the investor's strategy. However, most of the existing literature
on investment strategies with LETFs only allows for pure diffusion
processes for the equity index underlying the LETF (\cite{Giese2010}
\cite{Jarrow2010}, \cite{LeungSantoli2012}, \cite{LeungEtAl2016},
\cite{LeungSircar2015}, \cite{Wagalath2014}, \cite{GuasoniMayerhofer2023}).

We emphasize that while Figure \ref{fig: Intuition - no jumps}
is for the purposes of intuition
only, it is  nevertheless based on asset dynamics calibrated to
empirical US market data. Since even long-term investments (e.g. 10
years) can be managed effectively using a \textit{dynamic} investment
strategy with for example quarterly rebalancing (i.e. at the beginning
of each quarter, the investor faces investment choices and associated
outcomes such as those in Figure \ref{fig: Intuition - no jumps}
), this suggests that the benefits
of LETFs could potentially be harnessed without being unduly affected
by the compounding effects as well as time- and volatility-decay.
Our results show that this is indeed the case, even if no parametric
form of the underlying dynamics is assumed.

\subsection{Organization}
The remainder of the paper is organized as follows. Section \ref{sec:Problem-formulation}
provides the general problem formulation, with Section \ref{sec:Closed-form-solutions}
presenting closed-form results obtained under stylized assumptions.
Section \ref{sec:Numerical-solutions} discusses a neural network-based
solution approach to obtain the optimal investment strategies numerically
under multiple investment constraints. Finally, Section \ref{sec:Indicative-investment-results}
presents indicative investment results and Section \ref{sec:Conclusion}
concludes the paper, with additional analytical and numerical results
presented in Appendices \ref{sec:Appendix Proofs of key results},
\ref{sec:Appendix - Source-data} and \ref{sec:Appendix - Additional numerical results}.

\section{General problem formulation\label{sec:Problem-formulation}}

In this section we formulate, in general terms, the dynamic portfolio
optimization problem to be solved by an active portfolio manager (simply
``investor'') over a given time horizon $\left[t_{0}=0,T\right]$,
where $T>0$ can be large (e.g. 10 years). We assume that the investment
performance of the investor is measured relative to that of a given
benchmark portfolio, as is typically the case for professional asset
managers (see for example \cite{KashyapEtAl2021,LehalleSimon2021,AlekseevSokolov2016,Zhao2007,KornLindberg2014}).
To this end, for any $t\in\left[t_{0}=0,T\right]$, let $W\left(t\right)$
and $\hat{W}\left(t\right)$ denote the portfolio value (or informally,
simply the ``wealth'') of the investor and benchmark portfolios,
respectively. The same initial wealth $w_{0}\coloneqq W\left(t_{0}\right)=\hat{W}\left(t_{0}\right)>0$
is assumed to ensure that the performance comparison remains fair.
The investor's strategy is based on investing in any of a set of $N_{a}$
candidate assets, while the benchmark is defined in terms of $\hat{N}_{a}$
potentially different underlying assets. 

In more detail, if $\hat{\boldsymbol{X}}\left(t\right)$ denotes the
state (or informally, the information) used in obtaining the benchmark
asset allocation strategy at time $t\in\left[t_{0},T\right]$, let
$\hat{p}_{j}\left(t,\hat{\boldsymbol{X}}\left(t\right)\right)$ denotes
the proportion of the benchmark wealth $\hat{W}\left(t\right)$ invested
in asset $j\in\left\{ 1,..,\hat{N}_{a}\right\} $. The vector $\hat{\boldsymbol{p}}\left(t,\hat{\boldsymbol{X}}\left(t\right)\right)=\left(\hat{p}_{j}\left(t,\hat{\boldsymbol{X}}\left(t\right)\right):j=1,..,\hat{N}_{a}\right)\in\mathbb{R}^{\hat{N}_{a}}$
then denotes the asset allocation or investment strategy of the benchmark
at time $t\in\left[t_{0},T\right]$. 

Similarly, if $\boldsymbol{X}\left(t\right)$ denotes the state or
information incorporated by the investor in making their asset allocation
decision, let $p_{i}\left(t,\boldsymbol{X}\left(t\right)\right)$
denote the proportion of the investor's wealth $W\left(t\right)$
invested in asset $i\in\left\{ 1,..,N_{a}\right\} $, with $\boldsymbol{p}\left(t,\boldsymbol{X}\left(t\right)\right)=\left(p_{i}\left(t,\boldsymbol{X}\left(t\right)\right):i=1,..,N_{a}\right)\in\mathbb{R}^{N_{a}}$
denoting the investor's asset allocation or investment strategy at
time $t\in\left[t_{0},T\right]$.

The set of portfolio rebalancing events is denoted by $\mathcal{T}\subseteq\left[t_{0},T\right]$,
where we consider $\mathcal{T}=\left[t_{0},T\right]$ in the case
of continuous rebalancing (Section \ref{sec:Closed-form-solutions}),
or a discrete subset $\mathcal{T}\subset\left[t_{0},T\right]$ in
the case of discrete rebalancing (Section \ref{sec:Numerical-solutions}).
Given the set $\mathcal{T}$, the investor and benchmark investment
strategies are respectively defined as
\begin{equation}
\mathcal{P}=\left\{ \boldsymbol{p}\left(t,\boldsymbol{X}\left(t\right)\right),~t\in\mathcal{T}\right\} ,\qquad\textrm{and }\qquad\hat{\mathcal{P}}=\left\{ \hat{\boldsymbol{p}}\left(t,\hat{\boldsymbol{X}}\left(t\right)\right),~t\in\mathcal{T}\right\} .\label{eq:Investor and benchmark investment strategies}
\end{equation}

It is typical for the investor to be subject to investment constraints,
which are encoded by the set $\mathcal{A}$ of admissible controls.
In the simplest case, admissible investor strategies $\mathcal{P}\in\mathcal{A}$
are such that $\mathcal{P}=\left\{ \boldsymbol{p}\left(t,\boldsymbol{X}\left(t\right)\right)\in\mathcal{Z}:t\in\mathcal{T}\right\} $,
with $\mathcal{Z}$ denoting the admissible control space. More complex
constraints require a more careful formulation of $\mathcal{A}$ and
$\mathcal{Z}$, see for example Section \ref{sec:Numerical-solutions}.

Since the investor aims to construct $\mathcal{P}$ to \textit{outperform}
the benchmark strategy $\hat{\mathcal{P}}$, Assumption \ref{assu: General benchmark assumptions}
below outlines some general assumptions regarding the investment benchmark
$\hat{\mathcal{P}}$. Note that Assumption \ref{assu: General benchmark assumptions}
aligns with investment practice and is important for assessing the
relevance of LETFs when constructing portfolios for outperforming
a benchmark - see further discussion in Remark \ref{rem: Clarifying general assumptions benchmarks}
below. 

\begin{assumption}\label{assu: General benchmark assumptions}(General
assumptions regarding the benchmark strategy $\hat{\mathcal{P}}$)
We make the following general assumptions regarding the benchmark
strategies considered in this paper:
\begin{enumerate}
\item The investor can observe the asset allocation $\hat{\boldsymbol{p}}\left(t,\hat{\boldsymbol{X}}\left(t\right)\right)$
of the benchmark strategy at each $t\in\mathcal{T}$.
\item The sets of investable assets available to the investor and benchmark,
respectively, do not necessarily correspond. In particular, the benchmark
strategy may invest in assets which the investor is unwilling or unable
to invest in, or the investor might consider investing in a much larger
universe of investable assets than those included in the benchmark.\qed
\end{enumerate}
\end{assumption}

Remark \ref{rem: Clarifying general assumptions benchmarks} highlights
some key observations regarding Assumption \ref{assu: General benchmark assumptions}.

\begin{brem}

\label{rem: Clarifying general assumptions benchmarks}(Clarification
of benchmark assumptions) With regards to Assumption \ref{rem: Clarifying general assumptions benchmarks},
we observe the following:
\begin{enumerate}
\item Observable benchmarks play a key role in performance reporting for
many institutional investors, since active portfolio managers often
explicitly pursue the outperformance of a \textit{predetermined} investment
benchmark (see for example \cite{KashyapEtAl2021,LehalleSimon2021,AlekseevSokolov2016,Zhao2007,KornLindberg2014}).
As a result, the benchmark is clearly defined and transparent in the
sense of the underlying asset allocation, which often incorporates
broad market indices and bonds. In the case of pension funds, the
benchmark (or ``reference'') portfolios are usually constructed
using traded assets in fixed proportions. Examples include the Canadian
Pension Plan (CPP) with a base reference portfolio of 15\% Canadian
government bonds and 85\% global equity (\cite{cpp_site}), or the
Norwegian government pension plan (``Government Pension Fund Global'',
or GPFG) using a benchmark portfolio of 70\% equities and 30\% bonds
(\cite{Norway_plan}).
\item Active portfolio managers often consider not only different but indeed
larger/broader sets of assets than the benchmark. For example, pension
funds might include private equity whereas the benchmark might be
based on publicly traded assets only (see for example \cite{cpp_site}).
In the assessment of the effect of replacing VETFs with LETFs discussed
in Section \ref{sec:Closed-form-solutions}, we consider scenarios
where the benchmark strategy is defined in terms of a broad stock
market index, but the investor might not be able to invest directly
in the index itself, and invests instead in an ETF (VETF or LETF)
replicating the index returns. Since the ETFs only replicate (a multiple
of) the index returns \textit{before} costs, the existence of a non-zero
ETF expense ratios implies that investing in ETFs is not exactly the
same as investing in the underlying index, i.e. an ETF and its underlying
index can be viewed as different assets. Assumption \ref{rem: Clarifying general assumptions benchmarks}(ii)
is therefore relevant to an assessment of the role of a VETF or LETF
within portfolios designed to beat a broad equity index-based investment
benchmark. \qed
\end{enumerate}
\end{brem}

Let $E_{\mathcal{P}}^{t_{0},w_{0}}\left[\cdot\right]$ and $Var_{\mathcal{P}}^{t_{0},w_{0}}\left[\cdot\right]$
denote the expectation and variance, respectively, given initial wealth
$w_{0}=W\left(t_{0}\right)=\hat{W}\left(t_{0}\right)$ at time $t_{0}=0$
and using admissible investor strategy $\mathcal{P}\in\mathcal{A}$
over $\left[t_{0},T\right]$. As a result of Assumption \ref{rem: Clarifying general assumptions benchmarks},
the benchmark strategy $\hat{\mathcal{P}}$ remains implicit and fixed
for notational simplicity. 

As discussed in the Introduction, for an investment objective measuring
outperformance, we wish to maximize the information ratio (IR) of
the investor relative to the benchmark. In the context of dynamic
trading with strategies of the form (\ref{eq:Investor and benchmark investment strategies}),
the information ratio (IR) is defined as (\cite{GoetzmannEtAl2002,BajeuxEtAl2013})

\begin{eqnarray}
\mathcal{IR}_{\mathcal{P}}^{t_{0},w_{0}} & = & \frac{E_{\mathcal{P}}^{t_{0},w_{0}}\left[W\left(T\right)-\hat{W}\left(T\right)\right]}{Stdev_{\mathcal{P}}^{t_{0},w_{0}}\left[W\left(T\right)-\hat{W}\left(T\right)\right]}.\label{eq: IR goetzmann}
\end{eqnarray}
Maximizing the IR (\ref{eq: IR goetzmann}) can be achieved by solving
a mean-variance (MV) optimization problem (\cite{BajeuxEtAl2013})
,

\begin{equation}
\sup_{\mathcal{P}\in\mathcal{A}}\left\{ E_{\mathcal{P}}^{t_{0},w_{0}}\left[W\left(T\right)-\hat{W}\left(T\right)\right]-\rho\cdot Var_{\mathcal{P}}^{t_{0},w_{0}}\left[W\left(T\right)-\hat{W}\left(T\right)\right]\right\} ,
\quad\rho>0,
    \label{eq: IR MV formulation}
\end{equation}
where $\rho$ denotes a scalarization parameter.

Using the embedding technique of \cite{LiNg2000,ZhouLi2000}, we solve
(\ref{eq: IR MV formulation}) by formulating the equivalent problem
(\cite{PvSForsythLi2022_stochbm})

\begin{eqnarray}
\left(IR\left(\gamma\right)\right): &  & \inf_{\mathcal{P}\in\mathcal{A}}E_{\mathcal{P}}^{t_{0},w_{0}}\left[\left(W\left(T\right)-\left[\hat{W}\left(T\right)+\gamma\right]\right)^{2}\right],\qquad\gamma>0,\label{eq: IR objective USED}
\end{eqnarray}
where $\gamma>0$ denotes the embedding parameter. As discussed in
\cite{PvSForsythLi2022_stochbm}, the parameter $\gamma$ can be viewed
as the investor's (implicit) target for benchmark outperformance formulated
in terms of terminal wealth.

\begin{brem}\label{time_cons_remark}(Time consistency.)
Note that the control for problem (\ref{eq: IR MV formulation}) is formally the
pre-commitment policy, i.e. not time consistent.  However, the pre-commitment policy
solution of Problem (\ref{eq: IR MV formulation}) is identical to the strategy for an induced time
consistent policy \citep{Strub_2019,Forsyth2019CVaR}, and 
hence it is implementable.\footnote{An implementable strategy 
has the property that the investor has no incentive to deviate from the strategy
computed at time zero at later times \citep{Forsyth2019CVaR}.} 
The induced time consistent strategy in this  case
is the target based Problem (\ref{eq: IR objective USED}), with a fixed value of $\gamma, \forall t > 0$.
The relationship between pre-commitment
and implementable target-based schemes in the mean-variance context is discussed in  
\citet{Vigna_efficiency2014,Vigna2016TC,Vigna2017TailOptimality,MenoncinVigna2013}.
We consider the policy followed by the investor for $t>0$ to be the implementable
solution of Problem (\ref{eq: IR objective USED}) with a fixed value of $\gamma$.  This is identical to the
solution of Problem (\ref{eq: IR MV formulation}) as seen at $t=0$.
\end{brem}

\section{Closed-form solutions\label{sec:Closed-form-solutions}}

To obtain the valuable intuition regarding the characteristics of
IR-optimal investment strategies incorporating a LETF or VETF on a
broad equity market index, we present closed-form solutions to the
IR problem (\ref{eq: IR objective USED}) under stylized assumptions.
Remark \ref{rem: Relaxing closed-form assumptions} emphasizes that
these assumptions are required for the derivation of closed-form solutions
in this section only.

\begin{brem}

\label{rem: Relaxing closed-form assumptions}(Relaxing closed-form
assumptions) The closed-form results presented in this section require
stylized assumptions (Assumption \ref{assu: Stylized-assumptions-for closed-form}
and Assumption \ref{assu: Extra stylized-assumptions-for continuous rebalancing}
below), but we will use numerical techniques (Section \ref{sec:Numerical-solutions})
and present indicative investment results (Section \ref{sec:Indicative-investment-results})
where these assumptions are not required. The investment problem
is solved numerically in a setting where the following is allowed:
(i) no restrictions on the number of underlying assets, (ii) no parametric
assumptions are required for the dynamics of the underlying assets,
(iii) discrete portfolio rebalancing is used, (iv) leverage is restricted
and in some scenarios not allowed at all, (v) nonzero borrowing premiums
over the risk-free rate are applicable when funding leveraged positions,
(vi) no trading in the event of insolvency can occur, and (vii) more
general benchmark strategies are allowed, though for illustrative
purposes we will use constant proportion strategies due to their popularity
in practical applications.\qed

\end{brem}

The first set of general assumptions for the derivation of the closed-form
solution in this section is outlined in Assumption \ref{assu: Stylized-assumptions-for closed-form}.

\begin{assumption}

\label{assu: Stylized-assumptions-for closed-form}(Stylized assumptions
for closed-form solutions) For the purposes of obtaining closed-form
solutions in this section, we assume the following: 
\begin{enumerate}
\item The benchmark investment strategy (asset allocation) is a deterministic
function of time defined in terms of the 30-day T-bills (``risk-free''
asset) denoted by $B$ and a broad equity market index (``risky''
asset) denoted by $S$. Note that any known deterministic benchmark
strategy clearly satisfies Assumption \ref{assu: General benchmark assumptions},
and includes as a special case the constant proportion strategies
which are popular benchmarks used in practice (see for example \cite{cpp_site},\cite{Norway_plan}). 
\item We consider two investors, each optimizing their respective portfolios
relative to the same benchmark. Both investors are assumed to be unable
or unwilling to invest directly the underlying broad equity market
index itself (i.e. replicate the index with individual stocks), and
instead invests in ETFs referencing the index. The first investor,
informally referred to as the ``VETF investor'', allocates wealth
to two underlying assets, namely 30-day T-bills $B$ and a VETF $F_{v}$
with expense ratio $c_{v}>0$, where the VETF simply replicates the
instantaneous returns of the index $S$ before costs. The second investor,
informally referred to as the ``LETF investor'', allocates wealth
to two underlying assets, namely 30-day T-bills $B$ and a LETF $F_{\ell}$
with expense ratio $c_{\ell}>0$, where the LETF returns $\beta>1$
times the instantaneous returns of the index $S$ before costs.
\item Parametric dynamics for all underlying assets are assumed, including
jump-diffusion dynamics for the broad equity market index $S$ - see
(\ref{eq: B dynamics})-(\ref{eq: S dynamics}), (\ref{eq: VETF dynamics})
and (\ref{eq: LETF dynamics}) below. \qed
\end{enumerate}
\end{assumption}

Table \ref{tab: Closed-form solns - Candidate assets and benchmark}
provides an example of an investment scenario consistent with Assumption
\ref{assu: Stylized-assumptions-for closed-form}(i)-(ii), which will
be used for the illustration of the closed-form solutions of this
section. 

\noindent 
\begin{table}[!tbh]
\caption{Closed-form solutions - Candidate assets and benchmark: Example of
the investment scenario outlined in Assumption \ref{assu: Stylized-assumptions-for closed-form}(i)-(ii),
which will be used when illustrating the closed-form solutions in
this section. The constant proportion benchmark has been chosen to
align with typical benchmarks used by pension funds, while the indicative
expense (or cost) ratios are chosen from the range of expense ratios
of VETFs and LETFs on broad equity market indices typically observed
in practice. \label{tab: Closed-form solns - Candidate assets and benchmark}}

\noindent \centering{}{\footnotesize{}}%
\begin{tabular}{|c|>{\centering}p{1cm}|>{\raggedright}p{7.5cm}|>{\centering}m{1.5cm}||>{\centering}p{1.8cm}|>{\centering}p{1.8cm}|}
\hline 
\multicolumn{3}{|c|}{{\footnotesize{}Underlying assets}} & \multirow{2}{1.5cm}{{\footnotesize{}Benchmark}} & \multicolumn{2}{>{\centering}p{3.6cm}|}{{\footnotesize{}Investor candidate assets}}\tabularnewline
\cline{1-3} \cline{2-3} \cline{3-3} \cline{5-6} \cline{6-6} 
{\footnotesize{}Label} & {\footnotesize{}Value} & {\footnotesize{}Asset description} &  & {\footnotesize{}Using VETF} & {\footnotesize{}Using LETF}\tabularnewline
\hline 
\hline 
{\footnotesize{}T30} & {\footnotesize{}$B\left(t\right)$} & {\footnotesize{}30-day Treasury bill} & {\footnotesize{}30\%} & {\footnotesize{}$\checkmark$} & {\footnotesize{}$\checkmark$}\tabularnewline
\hline 
{\footnotesize{}Market} & {\footnotesize{}$S\left(t\right)$} & {\footnotesize{}Market portfolio (broad equity market index)} & {\footnotesize{}70\%} & {\footnotesize{}-} & {\footnotesize{}-}\tabularnewline
\hline 
{\footnotesize{}VETF} & {\footnotesize{}$F_{v}\left(t\right)$} & {\footnotesize{}Vanilla or standard/unleveraged ETF (VETF) replicating
the returns of the market portfolio $S\left(t\right)$, with expense
ratio $c_{v}=0.06\%$} & {\footnotesize{}0\%} & {\footnotesize{}$\checkmark$} & {\footnotesize{}-}\tabularnewline
\hline 
{\footnotesize{}LETF} & {\footnotesize{}$F_{\ell}\left(t\right)$} & {\footnotesize{}Leveraged ETF (LETF) with daily returns replicating
$\beta=2$ times the daily returns of the market portfolio $S\left(t\right)$,
with expense ratio $c_{\ell}=0.89\%$} & {\footnotesize{}0\%} & {\footnotesize{}-} & {\footnotesize{}$\checkmark$}\tabularnewline
\hline 
\end{tabular}{\footnotesize\par}
\end{table}

We assume that the underlying index $S$ can follow any of the commonly-encountered
jump diffusion processes in finance (see for example \cite{KouOriginal,MertonJumps1976}),
resulting in the following assumed dynamics for $B$ and $S$, respectively,
\begin{eqnarray}
\frac{dB\left(t\right)}{B\left(t\right)} & = & r\cdot dt,\label{eq: B dynamics}\\
\frac{dS\left(t\right)}{S\left(t^{-}\right)} & = & \left(\mu-\lambda\kappa_{1}^{s}\right)dt+\sigma\cdot dZ+d\left(\sum_{i=1}^{\pi\left(t\right)}\left(\xi_{i}^{s}-1\right)\right).\label{eq: S dynamics}
\end{eqnarray}
In (\ref{eq: B dynamics})-(\ref{eq: S dynamics}), $r$ denotes the
continuously compounded risk-free rate, $\pi\left(t\right)$ denotes
a Poisson process with intensity $\lambda\geq0$, while $\mu$ and
$\sigma$ denote the drift and volatility coefficients, respectively,
under the objective (real-world) probability measure. $\xi_{i}^{s}$
are i.i.d. random variables with the same distribution as $\xi^{s}$,
which represents the jump multiplier associated with the $S$-dynamics,
and we define 
\begin{equation}
\kappa_{1}^{s}=\mathbb{E}\left[\xi^{s}-1\right],\qquad\kappa_{2}^{s}=\mathbb{E}\left[\left(\xi^{s}-1\right)^{2}\right],\label{eq: Kappa and Kappa2 for S dynamics}
\end{equation}
which can be obtained using a given probability density function (pdf)
of $\xi^{s}$, denoted by $G\left(\xi^{s}\right)$. Finally, for any
functional $\psi\left(t\right),t\in\left[t_{0},T\right]$, we use
$\psi\left(t^{-}\right)$ and $\psi\left(t^{+}\right)$ as shorthand
notation for the one-sided limits $\psi\left(t^{-}\right)=\lim_{\epsilon\downarrow0}\psi\left(t-\epsilon\right)$
and $\psi\left(t^{+}\right)=\lim_{\epsilon\downarrow0}\psi\left(t+\epsilon\right)$,
respectively. Note that we can recover the assumption of geometric
Brownian motion (GBM) dynamics for $S$ by simply setting the intensity
$\lambda\equiv0$ in (\ref{eq: S dynamics}).

Since the VETF $F_{v}$ with expense ratio $c_{v}$ is a vanilla/standard
ETF simply replicating the returns of $S$ before costs, it has dynamics
given by
\begin{eqnarray}
\frac{dF_{v}\left(t\right)}{F_{v}\left(t^{-}\right)} & = & \frac{dS\left(t\right)}{S\left(t^{-}\right)}-c_{v}\cdot dt\nonumber \\
 & = & \left(\mu-\lambda\kappa_{1}^{s}-c_{v}\right)\cdot dt+\sigma\cdot dZ+d\left(\sum_{i=1}^{\pi\left(t\right)}\left(\xi_{i}^{s}-1\right)\right).\label{eq: VETF dynamics}
\end{eqnarray}

In contrast, the LETF $F_{\ell}$ with expense ratio $c_{\ell}>0$
aims at replicating $\beta>1$ times the instantaneous returns of
the underlying broad stock market index $S$ before costs, and therefore
has dynamics \textit{approximately} given by 
\begin{eqnarray}
\frac{dF_{\ell}\left(t\right)}{F_{\ell}\left(t^{-}\right)} & \simeq & \beta\frac{dS\left(t\right)}{S\left(t^{-}\right)}-\left[\left(\beta-1\right)r+c_{\ell}\right]dt.\label{eq: Approx LETF dynamics ito S dynamics}
\end{eqnarray}
It should be emphasized that (\ref{eq: Approx LETF dynamics ito S dynamics})
is only an approximation. Since the investor
in an LETF has limited liability,  exact equality in (\ref{eq: Approx LETF dynamics ito S dynamics})
only holds in the case where there are no jumps in the dynamics of
$S$. In the case of pure GBM dynamics, $F_{\ell}$ can never become
negative, hence limited liability is irrelevant. As a result, (\ref{eq: Approx LETF dynamics ito S dynamics})
is indeed used to model LETF dynamics in the literature where the
analysis is limited to GBM dynamics for $S$ (see for example \cite{AvellanedaZhang2010},
\cite{Jarrow2010}, \cite{GuasoniMayerhofer2023}), with the notable exception
of \citet{Ahn_2015} in the context of $\mathbb{Q}$ measure dynamics.

However, in the case of jump-diffusion dynamics for $S$, (\ref{eq: Approx LETF dynamics ito S dynamics})
is not quite correct since the LETF investor is protected by limited liability.
In more detail, if there is a jump with multiplier $\xi^{s}$ in the
underlying index (\ref{eq: S dynamics}) at a specific time $t$,
then the approximation (\ref{eq: Approx LETF dynamics ito S dynamics})
suggests that the value of the LETF would jump to $F_{\ell}\left(t\right)=F_{\ell}\left(t^{-}\right)\cdot\left[1+\beta\left(\xi^{s}-1\right)\right]$.

Therefore, in the case of a large downward jump in the underlying
index, in particular where the jump multiplier satisfies $\xi^{s}<\left(\beta-1\right)/\beta$,
the approximation (\ref{eq: Approx LETF dynamics ito S dynamics})
implies that $F_{\ell}\left(t\right)<0$, which cannot occur due to
limited liability of the LETF. Instead, if it is indeed the case that
$\xi^{s}\leq\left(\beta-1\right)/\beta$, the value of the LETF simply
drops to zero, i.e. $F_{\ell}\left(t\right)\equiv0$. 

We can therefore model the limited liability of the LETF $F_{\ell}$
by observing that $F_{\ell}$ therefore experiences jumps which are
related to, but not necessarily exactly the same as the jumps experienced
by the underlying index $S$. To this end, we define a jump multiplier
$\xi^{\ell}$ for the $F_{\ell}$-dynamics in terms of the jump multiplier
$\xi^{s}$ in the $S$-dynamics as \begin{eqnarray} \xi^{\ell} & = & \begin{dcases} \xi^{s} & \textrm{if }\xi^{s}>\left(\beta-1\right)/\beta,\\ \frac{\left(\beta-1\right)}{\beta} & \textrm{if }\xi^{s}\leq\left(\beta-1\right)/\beta. \end{dcases}\label{eq: LETF jump multiplier} \end{eqnarray}The
second case in (\ref{eq: LETF jump multiplier}) enforces the limited
liability of the LETF investor in the case of large downward jumps,
i.e. $F_{\ell}\left(t\right)\equiv0$ if $\xi^{s}\leq\left(\beta-1\right)/\beta$.
For subsequent use, we also define the following quantities involving
the LETF jump multiplier $\xi^{\ell}$, 
\begin{equation}
\kappa_{1}^{\ell}=\mathbb{E}\left[\xi^{\ell}-1\right],\qquad\kappa_{2}^{\ell}=\mathbb{E}\left[\left(\xi^{\ell}-1\right)^{2}\right],\qquad\kappa_{\chi}^{\ell,s}=\mathbb{E}\left[\left(\xi^{\ell}-1\right)\left(\xi^{s}-1\right)\right].\label{eq: Kappas for LETF dynamics}
\end{equation}
Given (\ref{eq: S dynamics}), (\ref{eq: Approx LETF dynamics ito S dynamics})
and (\ref{eq: LETF jump multiplier}), the LETF dynamics correctly
incorporating jumps is therefore given by 
\begin{eqnarray}
\frac{dF_{\ell}\left(t\right)}{F_{\ell}\left(t^{-}\right)} & = & \left[\beta\left(\mu-\lambda\kappa_{1}^{s}\right)-\left(\beta-1\right)r-c_{\ell}\right]\cdot dt+\beta\sigma\cdot dZ+\beta\cdot d\left(\sum_{i=1}^{\pi\left(t\right)}\left(\xi_{i}^{\ell}-1\right)\right),\label{eq: LETF dynamics}
\end{eqnarray}
where $\xi_{i}^{\ell}$ are i.i.d. random variables with the same
distribution as $\xi^{\ell}$, which represents the jump multiplier
associated with the $F_{\ell}$-dynamics (\ref{eq: LETF jump multiplier}).
We highlight the following observations regarding the dynamics (\ref{eq: S dynamics}),
(\ref{eq: VETF dynamics}) and (\ref{eq: LETF dynamics}).
\begin{enumerate}
\item The jumps in the underlying index $S$, the VETF $F_{v}$ and LETF
$F_{\ell}$ occur at the same time, so the Poisson process $\pi\left(t\right)$
and intensity $\lambda$ are the same in (\ref{eq: S dynamics}),
(\ref{eq: VETF dynamics}) and (\ref{eq: LETF dynamics}). More formally,
the processes (\ref{eq: S dynamics}), (\ref{eq: VETF dynamics})
and (\ref{eq: LETF dynamics}) have the same Poisson random measures,
but the compensated Poisson random measure is different for the $F_{\ell}$-dynamics
since the LETF jump sizes are slightly different due to (\ref{eq: LETF jump multiplier}). 
\item The dynamics (\ref{eq: VETF dynamics}) and (\ref{eq: LETF dynamics})
implicitly assume that the ETFs have negligible tracking errors, while
having non-negligible expense ratios. While ETF expense ratios can
indeed be material, especially for LETFs, the assumption that tracking
errors are negligible are often employed in the literature (see for
example \cite{BansalMarshall2015}, \cite{LeungSircar2015}). Given
the recent developments in designing replication strategies for LETFs
that remain robust even during periods of market volatility (see for
example \cite{GuasoniMayerhofer2023}), this appears to be a reasonable
assumption especially in the case of the most popular VETFs and LETFs
written on the major stock market indices.
\item As noted in the Introduction, we limit the analysis to the case of
LETFs where $\beta>1$ (i.e. ``bullish'' LETFs). However, the dynamics
(\ref{eq: LETF dynamics}) could also be applicable to inverse or
``bearish'' ETFs where $\beta<1$ (see for example \cite{Jarrow2010}),
but adjustments are usually required to incorporate the time-dependent
borrowing cost involved in short-selling the particular components
of the underlying replication basket each time $t$ (see for example
\cite{AvellanedaZhang2010}).
\end{enumerate}

\begin{brem}
(Relationship to \citet{Ahn_2015})
In \citet{Ahn_2015}, the authors model the limited liability of the LETF by taking the
point of view of the manager of the LETF.  The manager must purchase insurance
to handle the cases where the manager's position becomes negative.  In this work, we simply take
the point of view of the holder the LETF (not the manager), who has no exposure
to the possible negative value of the manager's position.  The cost of this
insurance (to the manager) is assumed to be passed
on to the LETF investor as part of the fee $c_{\ell}$ charged by the
manager, which is easily observable.
\end{brem}

\subsection{Intuition: lump sum investment scenario\label{subsec:Intuition:-lump-sum investment scenario}}

As a simple and intuitive illustration of the potential and risks
of using LETFs vs. VETFs, we consider a simple version of the general
formulation of the problem as outlined in Section \ref{sec:Problem-formulation}.
Specifically, we consider a lump sum investment scenario as discussed
in the Introduction (see Figures \ref{fig: Intuition - no jumps}
and \ref{fig: Intuition - with jumps}), where the initial wealth
$w_{0}=W\left(t_{0}\right)=\hat{W}\left(t_{0}\right)>0$ is invested
at time $t_{0}=0$ with no intermediate intervention/rebalancing until
the terminal time $T=\Delta t=0.25$ years (i.e. one quarter). In
the notation of Section \ref{sec:Problem-formulation}, we therefore
have a trivial set of rebalancing events $\mathcal{T}=\left[t_{0}\right]$. 

First, we consider the implications of the underlying asset dynamics
without referencing the investment strategy (i.e. wealth allocation
to assets). The dynamics (\ref{eq: B dynamics})-(\ref{eq: S dynamics})
imply that 
\begin{eqnarray}
\frac{B\left(\Delta t\right)}{B\left(0\right)} & = & \exp\left\{ r\cdot\Delta t\right\} ,\label{eq: B ratio discrete rebal}\\
\frac{S\left(\Delta t\right)}{S\left(0\right)} & = & \exp\left\{ \left(\mu-\lambda\kappa_{1}^{s}-\frac{1}{2}\sigma^{2}\right)\cdot\Delta t+\sigma\cdot Z\left(\Delta t\right)+\sum_{i=1}^{\pi\left(\Delta t\right)}\log\xi_{i}^{s}\right\} ,\label{eq: S ratio discrete rebal}
\end{eqnarray}
where the values $B\left(0\right)$ and $S\left(0\right)$ are observable
at time $t_{0}=0$. In the case of the VETF, we simply have 
\begin{align}
\frac{F_{v}\left(\Delta t\right)}{F_{v}\left(0\right)}= & \exp\left\{ -c_{v}\cdot\Delta t\right\} \cdot\left(\frac{S\left(\Delta t\right)}{S\left(0\right)}\right).\label{eq: VETF ratio discrete rebal}
\end{align}
In the case of the LETF, we have
\begin{eqnarray}
\frac{F_{\ell}\left(\Delta t\right)}{F_{\ell}\left(0\right)} & = & \exp\left\{ -c_{\ell}\cdot\Delta t\right\} \cdot f_{\ell}\left(\Delta t;\beta\right)\cdot\widetilde{Y}_{\ell}\left(\Delta t;\beta\right)\cdot\left(\frac{S\left(\Delta t\right)}{S\left(0\right)}\right)^{\beta},\label{eq: LETF *only* ratio discrete rebal}
\end{eqnarray}
where 
\begin{equation}
f_{\ell}\left(\Delta t;\beta\right)=\exp\left\{ -\left[\left(\beta-1\right)r+\frac{1}{2}\left(\beta-1\right)\beta\sigma^{2}\right]\cdot\Delta t\right\} ,\quad\textrm{and }\quad\widetilde{Y}_{\ell}\left(\Delta t;\beta\right)=\prod_{i=1}^{\pi\left(\Delta t\right)}\left[\frac{1+\beta\left(\xi_{i}^{\ell}-1\right)}{\left(\xi_{i}^{s}\right)^{\beta}}\right].\label{eq: f_ell and Y_ell_tilde}
\end{equation}
Expression (\ref{eq: LETF *only* ratio discrete rebal}) 
for the case where there are no
jumps, is given in \citet{AvellanedaZhang2010}.

For purposes of intuition, consider the \textit{special case} where
the underlying index $S$ experiences \textit{zero growth/decline}
over the time horizon $\Delta t$. If $S\left(\Delta t\right)=S\left(0\right)$,
it is clear that the LETF will perform worse than the VETF, i.e. 
assuming $\beta > 1$, $F_{\ell}\left(\Delta t\right)<F_{v}\left(\Delta t\right)$,
due to the following:
\begin{itemize}
\item Decay due to volatility: The term $f_{\ell}\left(\Delta t;\beta\right)<1$,
which only affects the LETF, is dominated by the diffusive volatility
$\sigma$ in the underlying $S$-dynamics. All else being equal, the
larger volatility of $S$, the worse the performance of the LETF relative
to the VETF, with the limiting case $\lim_{\sigma\rightarrow\infty}f_{\ell}\left(\Delta t;\beta\right)=0$.
\item Time decay: Even if there is no change in the value of the underlying
index, $S\left(\Delta t\right)=S\left(0\right)$, the value of the
LETF tends to zero if held for a long time, since $\lim_{\Delta t\rightarrow\infty}f_{\ell}\left(\Delta t;\beta\right)=0$. 
\item Costs and interest rates: Expense ratios for LETFs are typically substantially
higher than for VETFs, $0<c_{v}\ll c_{\ell}$. In addition, all else
being equal, increasing interest rates $r>0$ also decreases $f_{\ell}\left(\Delta t;\beta\right)$.
However, while these effects further decrease the value of the VETF
relative to the LETF, they are expected to be comparatively small
compared to the other effects.
\item Decay due to jumps: It can be shown that the term $\widetilde{Y}_{\ell}\left(\Delta t;\beta\right)$,
which only affects the LETF, satisfies $\widetilde{Y}_{\ell}\left(\Delta t;\beta\right)\leq1$,
where the maximum value ($\widetilde{Y}_{\ell}\left(\Delta t;\beta\right)=1$)
is achieved either when there are no jumps in the value of the underlying
over $\left[0,\Delta t\right]$ (i.e. $\pi\left(\Delta t\right)=0$),
or if there are jumps but they all satisfy $\xi_{i}^{s}=1$ which
has probability of almost surely zero. In other words, when $S\left(\Delta t\right)=S\left(0\right)$,
the mere presence of jumps in the underlying $S$ decreases the value
of the LETF relative to the VETF via the term $\widetilde{Y}_{\ell}\left(\Delta t;\beta\right)$.
This is illustrated in Figure \ref{fig: Intuition - with jumps} at
the point where $S\left(\Delta t\right)/S\left(0\right)=1$.
\end{itemize}
The preceding observations summarize what are effectively the standard
objections to LETFs that can be found in the literature, the only
addition being the rigorous treatment of jumps in the LETF dynamics
and the associated jump decay component.

Next, we discuss lump sum investment strategies, i.e. wealth allocation
to assets at time $t_{0}=0$ with no subsequent rebalancing prior
to maturity $T=\Delta t$. For simplicity, we consider a constant
proportion benchmark strategy $\hat{\mathcal{P}}=\hat{\boldsymbol{p}}\left(t_{0}\right)\coloneqq\left(1-\hat{p}_{s},\hat{p}_{s}\right)$,
where $\hat{p}_{s}$ denotes the proportion of benchmark wealth $\hat{W}\left(t_{0}\right)=w_{0}$
invested in the broad equity market index $S$ at time $t_{0}$, with
the remaining proportion $\left(1-\hat{p}_{s}\right)$ invested in
30-day T-bills. To emphasize that the benchmark wealth at the end
of the investment time horizon depends on $\hat{p}_{s}$, we use the
notation $\hat{W}\left(\Delta t;\hat{p}_{s}\right)$, and observe
that
\begin{eqnarray}
\frac{\hat{W}\left(\Delta t;\hat{p}_{s}\right)}{w_{0}} & = & \left(1-\hat{p}_{s}\right)\cdot\exp\left\{ r\cdot\Delta t\right\} +\hat{p}_{s}\cdot\frac{S\left(\Delta t\right)}{S\left(0\right)}.\label{eq: W_hat ratio discrete rebal}
\end{eqnarray}

{\color{red}
\noindent 
\begin{figure}[!tbh]
\noindent \begin{centering}
\subfloat[LETF: 100\%, T-bills: 0\% vs. VETF combinations]{\includegraphics[scale=0.7]{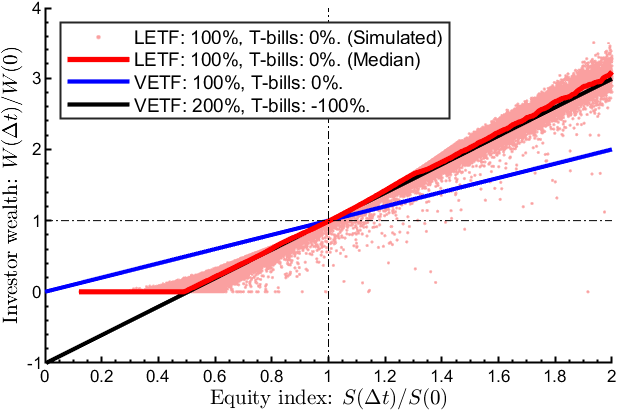}

}$\quad$$\quad$\subfloat[LETF: 50\%, T-bills: 50\% vs. VETF combinations]{\includegraphics[scale=0.7]{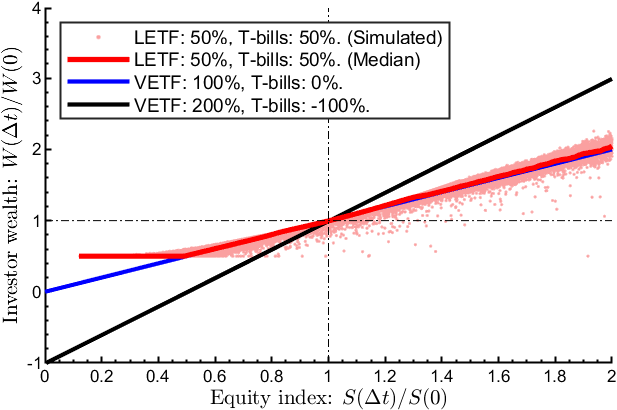}

}
\par\end{centering}
\caption{Payoffs when equity market index $S$ follows calibrated jump-diffusion
dynamics (\cite{KouOriginal} model): Investor wealth gross return
$W\left(\Delta t\right)/W\left(0\right)$ as a function of underlying
equity index gross return $S\left(\Delta t\right)/S\left(0\right)$,
$\Delta t=0.25$ (1 quarter), for different proportions of initial
wealth $W\left(0\right)$ invested in the LETF, VETF and T-bills at
time $t_{0}=0$. 
Asset parameters are calibrated to US equity and
bond market data over the period 1926:01 to 2023:12 (Appendix \ref{sec:Appendix - Source-data}),
LETF and VETF expense ratios are assumed to be 0.89\% and 0.06\% respectively,
and a borrowing premium of 3\% over the T-bill rate is applicable
to short positions.
\label{fig: Intuition - with jumps}}
\end{figure}
}

The investor, being unable to invest directly in $S$, can combine
an ETF investment with T-bills. We will assume that the investor does
not short-sell the LETF or VETF\footnote{As can be seen from the relationship between the objective functions
(\ref{eq: IR MV formulation}) and (\ref{eq: IR objective USED}),
optimizing the IR essentially places us within a variant of the Mean-Variance
(MV) framework with a constant risk aversion parameter. In MV optimization
(see for example \cite{PvSDangForsyth2018_TCMV,BensoussanEtAl2014})
with a constant risk aversion parameter, it is never optimal to short-sell
the risky asset. The subsequent results of Section \ref{sec:Closed-form-solutions}
and the results of \cite{NiLiForsyth2023_LFNN} suggest that this
observation also holds our investment scenario, i.e. it is never expected
to be IR-optimal to short-sell high return/high volatility assets
given the existence of low return/low volatility assets.}, but might short-sell the T-bills (i.e. borrow funds) to leverage
their investment in the ETF, in which case a constant borrowing premium
$b\geq0$ is added to the T-bill returns. In more detail, if $p$
denotes the fraction of wealth $W\left(t_{0}\right)=w_{0}$ that the
LETF or VETF investors invest in their respective ETFs, an investment
fraction $p>1$ in the ETF is funded by borrowing the amount $\left(1-p\right)\cdot w_{0}$
at an interest rate of $\left(r+b\right)$, so the T-bill dynamics
applicable to the investors can be modified as
\begin{equation}
\frac{\overline{B}\left(\Delta t\right)}{\overline{B}\left(0\right)}=\exp\left\{ \overline{r}\left(p\right)\cdot\Delta t\right\} ,\qquad\textrm{where \qquad}\overline{r}\left(p\right)=r+b\cdot\mathbb{I}_{\left[p>1\right]},\label{eq: T-bill dynamics for investors discrete}
\end{equation}
with $\mathbb{I}_{\left[A\right]}$ denoting the indicator of the
event $A$.

The VETF investor (see Assumption \ref{assu: Stylized-assumptions-for closed-form}(ii))
specifies an investment strategy $\mathcal{P}_{v}=\boldsymbol{p}_{v}\left(t_{0}\right)=\left(1-p_{v},p_{v}\right)$,
where $p_{v}$ denotes the fraction of wealth $W\left(t_{0}\right)=w_{0}$
invested in the VETF $F_{v}$ at time $t_{0}=0$, and the remaining
fraction of wealth $\left(1-p_{v}\right)$ invested in 30-day T-bills.
The VETF investor's wealth at the end of the investment time horizon,$W_{v}\left(\Delta t;p_{v}\right)$
therefore satisfies
\begin{eqnarray}
\frac{W_{v}\left(\Delta t;p_{v}\right)}{w_{0}} & = & \left(1-p_{v}\right)\cdot\exp\left\{ \overline{r}\left(p_{v}\right)\cdot\Delta t\right\} +p_{v}\cdot\exp\left\{ -c_{v}\cdot\Delta t\right\} \cdot\left(\frac{S\left(\Delta t\right)}{S\left(0\right)}\right).\label{eq: VETF W ratio discrete rebal}
\end{eqnarray}

Similarly, the LETF investor specifies investment strategy $\mathcal{P}_{\ell}=\boldsymbol{p}_{\ell}\left(t_{0}\right)=\left(1-p_{\ell},p_{\ell}\right)$,
where $p_{\ell}$ is the fraction of wealth $W\left(t_{0}\right)=w_{0}$
invested in the LETF $F_{\ell}$ at time $t_{0}=0$, and the remaining
fraction of wealth $\left(1-p_{\ell}\right)$ invested in 30-day T-bills.
Using (\ref{eq: LETF *only* ratio discrete rebal}), the LETF investor's
wealth at the end of the investment time horizon,$W_{\ell}\left(\Delta t;p_{\ell}\right)$
therefore satisfies
\begin{equation}
\frac{W_{\ell}\left(\Delta t;p_{\ell}\right)}{w_{0}}=\left(1-p_{\ell}\right)\cdot\exp\left\{ \overline{r}\left(p_{\ell}\right)\cdot\Delta t\right\} +p_{\ell}\cdot\exp\left\{ -c_{\ell}\cdot\Delta t\right\} \cdot f_{\ell}\left(\Delta t;\beta\right)\cdot\widetilde{Y}_{\ell}\left(\Delta t;\beta\right)\cdot\left(\frac{S\left(\Delta t\right)}{S\left(0\right)}\right)^{\beta}.\label{eq: LETF W ratio discrete rebal}
\end{equation}

Varying $p_{v}$ and $p_{\ell}$ in (\ref{eq: VETF W ratio discrete rebal})
and (\ref{eq: LETF W ratio discrete rebal}) therefore trace out different
payoffs for $W_{v}\left(\Delta t;p_{v}\right)$ and $W_{\ell}\left(\Delta t;p_{\ell}\right)$
as a function of the (random) underlying index outcome $S\left(\Delta t\right)$,
with specific choices of $p_{v}$ and $p_{\ell}$ illustrated in Figures
\ref{fig: Intuition - no jumps} and \ref{fig: Intuition - with jumps}.
Note however that the wealth of the VETF investor $W_{v}\left(\Delta t;p_{v}\right)$
is linear in $S\left(\Delta t\right)$, and conditional on $S\left(\Delta t\right)$
the outcome $W_{v}\left(\Delta t;p_{v}\right)$ is deterministic.
However, this is not the case for the wealth $W_{\ell}\left(\Delta t;p_{\ell}\right)$
of the LETF investor, which has a power call-like payoff due to the
$\left[S\left(\Delta t\right)\right]^{\beta}$ term of (\ref{eq: LETF W ratio discrete rebal})
in conjunction with limited liability. Note that even if we condition
on the value of $S\left(\Delta t\right)$, the wealth outcome $W_{\ell}\left(\Delta t;p_{\ell}\right)$
is \textit{not} deterministic due to the presence of the jump term
$\widetilde{Y}_{\ell}\left(\Delta t;\beta\right)$in (\ref{eq: LETF W ratio discrete rebal}).
However, if no jumps are present then $W_{\ell}\left(\Delta t;p_{\ell}\right)$
conditional on $S\left(\Delta t\right)$ is linear in $\left[S\left(\Delta t\right)\right]^{\beta}$,
compare Figures \ref{fig: Intuition - no jumps} and \ref{fig: Intuition - with jumps}.

Suppose the LETF and VETF investors want to choose values $p_{v}^{\ast}$
and $p_{\ell}^{\ast}$, respectively, to maximize the IR (\ref{eq: IR goetzmann})
subject to an implicit target $\gamma>0$ in (\ref{eq: IR objective USED}).
In this setting of parametric asset dynamics, we can simulate $N_{d}$
paths of the underlying equity index using (\ref{eq: S ratio discrete rebal})
use each path's information together. 

{\PVSnew{First, we discretize possible values of the fractions $p_{v}$ and $p_{\ell}$ using a fine grid, so that using each discretized value of $p_{v}$ and $p_{\ell}$, we can obtain the corresponding values of $W_{v}^{\left(j\right)}\left(\Delta t;p_{v}\right)$
and $W_{\ell}^{\left(j\right)}\left(\Delta t;p_{\ell}\right)$ respectively,
along each path $j=1,...,N_{d}$. Next, using a discretization
of the objective (\ref{eq: IR objective USED}) in this setting, we can find the approximate IR-optimal values $p_{v}^{\ast}$ and $p_{\ell}^{\ast}$ by exhaustive search over the grid by solving
\begin{eqnarray}
p_{k}^{\ast} & = & \arg\min_{p_{k}}\left\{ \frac{1}{N_{d}}\sum_{j=1}^{N_{d}}\left(W_{k}^{\left(j\right)}\left(\Delta t;p_{k}\right)-\left[\hat{W}^{\left(j\right)}\left(\Delta t;\hat{p}_{s}\right)+\gamma\right]\right)^{2}\right\} ,\qquad k\in\left\{ v,\ell\right\} .\label{eq: IR optimization numerical discrete lump sum}
\end{eqnarray} 
}}
Figure \ref{fig: Single period discrete rebalancing} illustrates
the results of this procedure for two different values of the implicit target, $\gamma=20$
and $\gamma=50$, where we observe the following:
\begin{itemize}
\item In the case of $\gamma=20$ in Figure \ref{fig: Single period discrete rebalancing}(a), the VETF investor simply invests all wealth in the VETF ($p_{v}^{\ast}=100\%$),
whereas the LETF investor invests slightly less than half of total
wealth in the LETF ($p_{\ell}^{\ast}=48.3\%$). Note that the IR-optimal strategies
in this case satisfy $p_{v}^{\ast}/p_{\ell}^{\ast}=2.070$.

\item With a significantly more aggressive benchmark outperformance target
of $\gamma=50$ in Figure \ref{fig: Single period discrete rebalancing}(b),
the LETF investor now invests $p_{\ell}^{\ast}=70.1\%$ in the LETF,
i.e. there is no need to leverage the LETF investment itself, whereas
the VETF investor borrows 20\% of wealth to invest $p_{v}^{\ast}=120\%$
in the VETF. This leverage is costly for the VETF investor due to
the lack of downside protection and borrowing premiums, which can
be seen in both the upside and extreme downside outcomes of Figure
\ref{fig: Single period discrete rebalancing}(b). In this case, we
have $p_{v}^{\ast}/p_{\ell}^{\ast}=1.712$.
\end{itemize}
Comparing IR-optimal investment strategies implemented using a LETF
or VETF with the same benchmark outperformance target $\gamma>0$,
the observation that $p_{v}^{\ast}/p_{\ell}^{\ast}\approx\beta=2$
holds in the cases illustrated in Figure \ref{fig: Single period discrete rebalancing}
is not an accident. This relationship is more rigorously discussed
in the subsequent results of this section, but for now we note that
exact equality $p_{v}^{\ast}/p_{\ell}^{\ast}\equiv\beta$, where $\beta$
is the returns multiplier of the LETF, only holds for the IR-optimal
investment strategies in the case of continuous rebalancing ($\Delta t\downarrow0$),
zero expense ratios ($c_{v}=c_{\ell}=0$), zero borrowing premium
over the risk-free rate $r$ ($b=0$), and when no leverage restrictions
are applicable. While this is a very restrictive set of assumptions,
$p_{v}^{\ast}/p_{\ell}^{\ast}\approx\beta$ is nevertheless a useful
rule-of-thumb to keep in mind when comparing IR-optimal investment
strategies in more general cases, a simple example being Figure \ref{fig: Single period discrete rebalancing}.
However, it should be emphasized that while the \textit{strategies}
might satisfy $p_{v}^{\ast}/p_{\ell}^{\ast}\approx\beta$, this does
\textit{not} mean that the ultimate investment \textit{outcomes} for
the LETF and VETF investors (such as investor wealth, benchmark outperformance
etc.) have straightforward relationships except under very restrictive
conditions, since the outcomes are affected by limited liability and
the path-dependent role of jumps (see Figure \ref{fig: Single period discrete rebalancing}).

\noindent 
\begin{figure}[!tbh]
\noindent \begin{centering}
\subfloat[$\gamma=20$: $p_{\ell}^{\ast}=48.3\%$, $p_{v}^{\ast}=100\%$]{\includegraphics[scale=0.75]{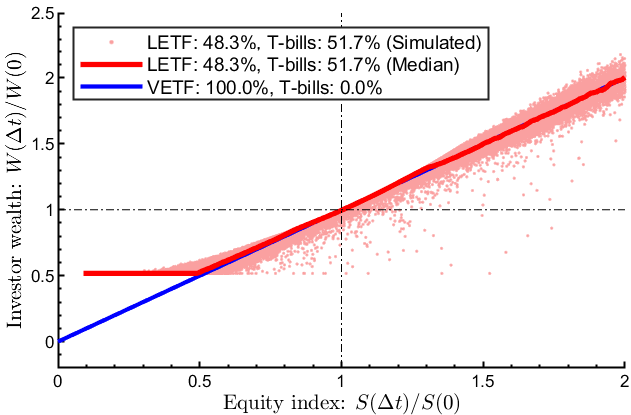}

}$\quad$$\quad$\subfloat[$\gamma=50$: $p_{\ell}^{\ast}=70.1\%$, $p_{v}^{\ast}=120\%$]{\includegraphics[scale=0.75]{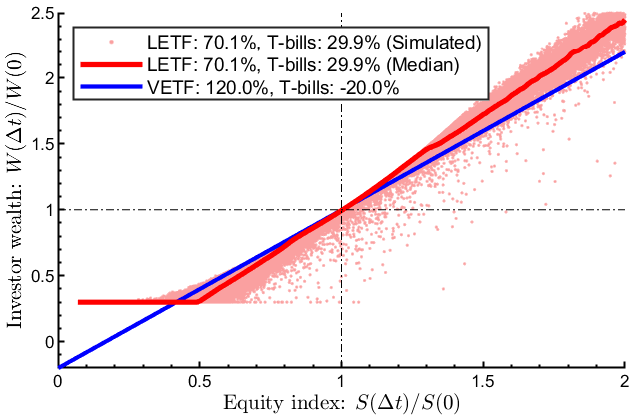}

}
\par\end{centering}
\caption{Payoffs when equity market index $S$ follows calibrated jump-diffusion
dynamics (\cite{KouOriginal} model): Investor wealth gross return
$W\left(\Delta t\right)/W\left(0\right)$ as a function of underlying
equity index gross return $S\left(\Delta t\right)/S\left(0\right)$,
$\Delta t=0.25$ (1 quarter), for different proportions of initial
wealth $W\left(0\right)$ invested in the LETF, VETF and T-bills at
time $t_{0}=0$. Asset parameters are calibrated to US equity and
bond market data over the period 1926:01 to 2023:12 (Appendix \ref{sec:Appendix - Source-data}),
LETF and VETF expense ratios are assumed to be 0.89\% and 0.06\% respectively,
and a borrowing premium of 3\% over the T-bill rate is applicable
to short positions.\label{fig: Single period discrete rebalancing}}
\end{figure}

\subsection{Dynamically-optimal strategies under continuous rebalancing\label{subsec: Optimal strategies CONTINUOUS rebal}}

Section \ref{subsec:Intuition:-lump-sum investment scenario} treated
the lump-sum investment scenario with no subsequent intervention over
$\left(t_{0},T\right]$, i.e. the set of rebalancing times being simply
$\mathcal{T}=\left[t_{0}\right]$. We now consider the other extreme,
namely that of continuous rebalancing, where the set of rebalancing
times is given by $\mathcal{T}=\left[t_{0},T\right]$.

Derivation of closed-form optimal strategies necessarily requires
stylized assumptions, in this case outlined in Assumption \ref{assu: Extra stylized-assumptions-for continuous rebalancing}.
As per Remark \ref{rem: Relaxing closed-form assumptions}, we emphasize
that these assumptions are not required for the subsequent results
discussed in Section \ref{sec:Indicative-investment-results}.

\begin{assumption}

\label{assu: Extra stylized-assumptions-for continuous rebalancing}(Stylized
assumptions - continuous rebalancing) In the case of continuous rebalancing,
for the purposes of obtaining closed-form solutions in this section,
we assume the following: 
\begin{enumerate}
\item Assumption \ref{assu: Stylized-assumptions-for closed-form} holds,
including parametric dynamics (\ref{eq: B dynamics})-(\ref{eq: S dynamics}),
(\ref{eq: VETF dynamics}) and (\ref{eq: LETF dynamics}) for the
underlying assets.
\item The investor injects cash into the portfolio at a constant rate of
$q\geq0$ per year. To ensure the wealth processes remain comparable,
the identical rate of cash injection is assumed for the benchmark
portfolio.
\item We assume continuous portfolio rebalancing ($\mathcal{T}=\left[t_{0},T\right]$),
no investment constraints (i.e. no leverage limits or short-selling
constraints), zero borrowing premium so that both borrowing and lending
occurs at the risk-free rate $r$, and trading is allowed to continue
in the event of insolvency. Note that these assumptions are standard
in the derivation of closed-form solutions of multi-period portfolio
optimization problems (see for example \cite{ZhouLi2000}). This implies
that the investment in the LETF can itself be leveraged, which is
plausible since even retail investors can borrow and invest in LETFs.
However, the degree to which leverage is required by either the LETF
or VETF investors depends on the aggressiveness of the outperformance
target $\gamma$, as shown by the subsequent results. \qed
\end{enumerate}
\end{assumption}

Since Assumption \ref{assu: Stylized-assumptions-for closed-form}
remains applicable (see Assumption \ref{assu: Extra stylized-assumptions-for continuous rebalancing}(i)),
the deterministic benchmark strategy allocates wealth to two assets,
namely the T-bills $B$ and the broad equity market index $S$. For
the special case of continuous rebalancing, let $t\rightarrow\hat{\varrho}_{s}\left(t\right)$
be a deterministic function of time denoting the proportional allocation
to $S$ at time $t\in\mathcal{T}=\left[t_{0},T\right]$, with $\left(1-\hat{\varrho}_{s}\left(t\right)\right)$
denoting the corresponding allocation to T-bills. The benchmark strategy
in this section is therefore given by 
\begin{eqnarray}
\hat{\mathcal{P}} & = & \left\{ \hat{\boldsymbol{p}}\left(t,\hat{W}\left(t\right)\right)=\left(1-\hat{\varrho}_{s}\left(t\right),\;\hat{\varrho}_{s}\left(t\right)\right):t\in\left[t_{0},T\right]\right\} .\label{eq: Benchmark strategy -  analytical solutions}
\end{eqnarray}

By Assumption \ref{assu: Stylized-assumptions-for closed-form}(ii),
in the case of the VETF investor, let $\varrho_{v}\left(t,\boldsymbol{X}_{v}\left(t\right)\right)$
denote the proportional allocation of wealth $W_{v}\left(t\right)$
at time $t$ to the VETF $F_{v}$ in the case of continuous rebalancing,
with $\boldsymbol{X}_{v}\left(t\right)=\left(W_{v}\left(t\right),\hat{W}\left(t\right),\hat{\varrho}_{s}\left(t\right)\right)$.
The VETF investor strategy is therefore of the form 
\begin{eqnarray}
\mathcal{P}_{v} & = & \left\{ \boldsymbol{p}_{v}\left(t,\boldsymbol{X}_{v}\left(t\right)\right)=\left(1-\varrho_{v}\left(t,\boldsymbol{X}_{v}\left(t\right)\right),\;\varrho_{v}\left(t,\boldsymbol{X}_{v}\left(t\right)\right)\right):t\in\left[t_{0},T\right]\right\} .\label{eq: VETF Investor strategy - analytical solutions}
\end{eqnarray}

In the case of the LETF investor, let $\varrho_{\ell}\left(t,\boldsymbol{X}_{\ell}\left(t\right)\right)$
denote the proportional allocation of wealth $W_{\ell}\left(t\right)$
at time $t$ to the LETF $F_{\ell}$ in the case of continuous rebalancing,
with $\boldsymbol{X}_{\ell}\left(t\right)=\left(W_{\ell}\left(t\right),\hat{W}\left(t\right),\hat{\varrho}_{s}\left(t\right)\right)$,
to obtain the LETF investor strategy as
\begin{eqnarray}
\mathcal{P}_{\ell} & = & \left\{ \boldsymbol{p}_{\ell}\left(t,\boldsymbol{X}_{\ell}\left(t\right)\right)=\left(1-\varrho_{\ell}\left(t,\boldsymbol{X}_{\ell}\left(t\right)\right),\;\varrho_{\ell}\left(t,\boldsymbol{X}_{\ell}\left(t\right)\right)\right):t\in\left[t_{0},T\right]\right\} .\label{eq: LETF investor strategy - analytical solutions}
\end{eqnarray}

Given investment strategies of the form (\ref{eq: Benchmark strategy -  analytical solutions})
and (\ref{eq: LETF investor strategy - analytical solutions}), as
well as dynamics (\ref{eq: B dynamics}), (\ref{eq: S dynamics}),
(\ref{eq: VETF dynamics}) and (\ref{eq: LETF dynamics}), we therefore
have the following wealth dynamics in the case of continuous rebalancing:
\begin{eqnarray}
d\hat{W}\left(t\right) & = & \left\{ \hat{W}\left(t^{-}\right)\cdot\left[r+\hat{\varrho}_{s}\left(t\right)\left(\mu-\lambda\kappa_{1}^{s}-r\right)\right]+q\right\} \cdot dt\nonumber \\
 &  & +\hat{W}\left(t^{-}\right)\hat{\varrho}_{s}\left(t\right)\sigma\cdot dZ\left(t\right)+\hat{W}\left(t^{-}\right)\hat{\varrho}_{s}\left(t\right)\cdot d\left(\sum_{i=1}^{\pi\left(t\right)}\left(\xi_{i}^{s}-1\right)\right),\label{eq: SDE W_hat}
\end{eqnarray}
\begin{eqnarray}
dW_{v}\left(t\right) & = & \left\{ W_{v}\left(t^{-}\right)\cdot\left[r+\varrho_{v}\left(t,\boldsymbol{X}_{v}\left(t^{-}\right)\right)\left\{ \left(\mu-\lambda\kappa_{1}^{s}-r\right)-c_{v}\right\} \right]+q\right\} \cdot dt\nonumber \\
 &  & +W_{v}\left(t^{-}\right)\varrho_{v}\left(t,\boldsymbol{X}_{v}\left(t^{-}\right)\right)\sigma\cdot dZ\left(t\right)+W_{v}\left(t^{-}\right)\varrho_{v}\left(t,\boldsymbol{X}_{v}\left(t^{-}\right)\right)\cdot d\left(\sum_{i=1}^{\pi\left(t\right)}\left(\xi_{i}^{s}-1\right)\right),\label{eq: SDE W VETF}
\end{eqnarray}
\begin{eqnarray}
dW_{\ell}\left(t\right) & = & \left\{ W_{\ell}\left(t^{-}\right)\cdot\left[r+\varrho_{\ell}\left(t,\boldsymbol{X}_{\ell}\left(t^{-}\right)\right)\left\{ \beta\left(\mu-\lambda\kappa_{1}^{s}-r\right)-c_{\ell}\right\} \right]+q\right\} \cdot dt\nonumber \\
 &  & +W_{\ell}\left(t^{-}\right)\varrho_{\ell}\left(t,\boldsymbol{X}_{\ell}\left(t^{-}\right)\right)\beta\sigma\cdot dZ\left(t\right)+W_{\ell}\left(t^{-}\right)\varrho_{\ell}\left(t,\boldsymbol{X}_{\ell}\left(t^{-}\right)\right)\beta\cdot d\left(\sum_{i=1}^{\pi\left(t\right)}\left(\xi_{i}^{\ell}-1\right)\right),\label{eq: SDE W LETF}
\end{eqnarray}
for $t\in\left(t_{0},T\right]$, with initial wealth $W_{v}\left(t_{0}\right)=W_{\ell}\left(t_{0}\right)=\hat{W}\left(t_{0}\right)=w_{0}$.
As a reminder, $q\geq0$ denotes the constant rate per year at which
cash is contributed to each portfolio (see Assumption \ref{assu: Extra stylized-assumptions-for continuous rebalancing}(ii)),
and $\beta>1$ in (\ref{eq: SDE W LETF}) denotes the multiplier of
the LETF.

Due to Assumption \ref{assu: Extra stylized-assumptions-for continuous rebalancing}(iii),
the set of admissible investor strategies is given by $\varrho_{k}\left(t,\boldsymbol{X}_{k}\left(t\right)\right)\in\mathcal{A}_{0}$
for $k\in\left\{ v,\ell\right\} $, where
\begin{eqnarray}
\mathcal{A}_{0} & = & \left\{ \left.\varrho_{k}\left(t,w,\hat{w},\hat{\varrho}_{s}\left(t\right)\right)\right|\varrho_{k}:\left[t_{0},T\right]\times\mathbb{R}^{3}\rightarrow\mathbb{R}\right\} ,\qquad k\in\left\{ v,\ell\right\} .\label{eq: A admissible controls - NO constraints}
\end{eqnarray}
The IR optimization problem (\ref{eq: IR objective USED}) in this
setting can therefore be written as
\begin{eqnarray}
\left(IR\left(\gamma\right)\right): &  & \inf_{\varrho_{k}\in\mathcal{A}_{0}}E_{\varrho_{k}}^{t_{0},w_{0}}\left[\left(W_{k}\left(T\right)-\left[\hat{W}\left(T\right)+\gamma\right]\right)^{2}\right],\qquad\gamma>0,\quad\textrm{for }k\in\left\{ v,\ell\right\} ,\label{eq: IR problem for analytical}
\end{eqnarray}
with wealth dynamics (\ref{eq: SDE W_hat}),(\ref{eq: SDE W VETF})
and (\ref{eq: SDE W LETF}) respectively. 

The following theorem describes the HJB partial integro-differential
equation (PIDE) satisfied by the value function of (\ref{eq: IR problem for analytical})
for the LETF investor.
\begin{thm}
\label{thm: Verification theorem}(IR optimization for the LETF investor:
Verification theorem) Let $\gamma>0$, and assume a given benchmark
strategy $t\rightarrow\hat{\varrho}_{s}\left(t\right)$ that is deterministic
and integrable. Suppose that for all $\left(t,w,\hat{w},\hat{\varrho}_{s}\right)\in\left[t_{0},T\right]\times\mathbb{R}^{3}$,
there exist functions $V\left(t,w,\hat{w},\hat{\varrho}_{s}\right):\left[t_{0},T\right]\times\mathbb{R}^{3}\rightarrow\mathbb{R}$
and $\varrho_{\ell}^{\ast}\left(t,w,\hat{w},\hat{\varrho}_{s};\gamma\right):\left[t_{0},T\right]\times\mathbb{R}^{3}\rightarrow\mathbb{R}$
such that: (i) $V$ and $\varrho_{\ell}^{\ast}$ are sufficiently
smooth and solve the HJB PIDE (\ref{eq: IR pide})-(\ref{eq: IR pide terminal condition}),
and (ii) the pointwise supremum in (\ref{eq: IR pide}) is attained
by the function $\varrho_{\ell}^{\ast}\left(t,w,\hat{w},\hat{\varrho}_{s};\gamma\right)$.
\begin{eqnarray}
\frac{\partial V}{\partial t}+\inf_{\varrho_{\ell}\in\mathbb{R}}\left\{ \begin{array}{c}
\\
\\
\end{array}\!\!\!\mathcal{H}\left(\varrho_{\ell};t,w,\hat{w},\hat{\varrho}_{s}\right)\begin{array}{c}
\\
\\
\end{array}\!\!\!\right\}  & = & 0,\label{eq: IR pide}\\
V\left(T,w,\hat{w},\hat{\varrho}_{s}\right) & = & \left(w-\hat{w}-\gamma\right)^{2},\label{eq: IR pide terminal condition}
\end{eqnarray}
where 
\begin{eqnarray}
\mathcal{H}\left(\varrho_{\ell};t,w,\hat{w},\hat{\varrho}_{s}\right) & = & \left(w\cdot\left[r+\left\{ \beta\left(\mu-\lambda\kappa_{1}^{s}-r\right)-c_{\ell}\right\} \cdot\varrho_{\ell}\right]+q\right)\cdot\frac{\partial V}{\partial w}\nonumber \\
 &  & +\left(\hat{w}\cdot\left[r+\left(\mu-\lambda\kappa_{1}^{s}-r\right)\cdot\hat{\varrho}_{s}\right]+q\right)\cdot\frac{\partial V}{\partial\hat{w}}\nonumber \\
 &  & +\frac{1}{2}\left(\varrho_{\ell}\cdot w\beta\sigma\right)^{2}\cdot\frac{\partial^{2}V}{\partial w^{2}}+\frac{1}{2}\left(\hat{\varrho}_{s}\hat{w}\sigma\right)^{2}\cdot\frac{\partial^{2}V}{\partial\hat{w}^{2}}+\left(\varrho_{\ell}\cdot w\beta\sigma\right)\left(\hat{\varrho}_{s}\hat{w}\sigma\right)\cdot\frac{\partial^{2}V}{\partial w\partial\hat{w}}\nonumber \\
 &  & -\lambda\cdot V+\lambda\cdot\int_{0}^{\infty}V\left(w+\varrho_{\ell}\cdot w\beta\left(\xi^{\ell}-1\right),\hat{w}+\hat{\varrho}_{s}\hat{w}\left(\xi^{s}-1\right),t\right)G\left(\xi^{s}\right)d\xi^{s}.\label{eq: Ham_IR}
\end{eqnarray}

Then given Assumption \ref{assu: Extra stylized-assumptions-for continuous rebalancing}
and wealth dynamics (\ref{eq: SDE W_hat}) and (\ref{eq: SDE W LETF}),
$V$ is the value function and $\varrho_{\ell}^{\ast}$ is the optimal
control (i.e. optimal proportion of the investor's wealth to be invested
in the LETF with $\beta>1$) for the $IR\left(\gamma\right)$ problem
(\ref{eq: IR problem for analytical}) for the LETF investor. 
\end{thm}

\begin{proof}
See Appendix \ref{subsec: Appendix proof of Verification thm}. Note
that since $\xi^{\ell}$ is a function of $\xi^{s}$ (see (\ref{eq: LETF jump multiplier})),
the integral in (\ref{eq: Ham_IR}) is only written with respect to
values of $\xi^{s}$ with associated PDF $G\left(\xi^{s}\right)$.
\end{proof}
Solving the HJB PIDE (\ref{eq: IR pide})-(\ref{eq: IR pide terminal condition}),
we obtain the IR-optimal investment strategy for the LETF investor
as per Proposition \ref{prop: Closed-form Optimal control for LETF}. 
\begin{prop}
\label{prop: Closed-form Optimal control for LETF}(IR-optimal investment
strategy using the LETF) Let $\gamma>0$ be fixed. Suppose that Assumption
\ref{assu: Extra stylized-assumptions-for continuous rebalancing}
and wealth dynamics (\ref{eq: SDE W_hat}) and (\ref{eq: SDE W LETF})
apply. Let $W_{\ell}^{\ast}\left(t\right)$ denote the LETF investor's
wealth process (\ref{eq: SDE W LETF}) under the optimal strategy
$\varrho_{\ell}^{\ast}$, and let $\boldsymbol{X}_{\ell}^{\ast}\left(t\right)=\left(W_{\ell}^{\ast}\left(t\right),\hat{W}\left(t\right),\hat{\varrho}_{s}\left(t\right)\right)$.
Then the IR-optimal fraction of the investor's wealth invested in
the LETF, $\varrho_{\ell}^{\ast}$, satisfies
\begin{eqnarray}
 &  & \varrho_{\ell}^{\ast}\left(t,\boldsymbol{X}_{\ell}^{\ast}\left(t^{-}\right)\right)\cdot W_{\ell}^{\ast}\left(t^{-}\right)\nonumber \\
 & = & \left(\frac{\beta\left[\mu+\lambda\left(\kappa_{1}^{\ell}-\kappa_{1}^{s}\right)-r\right]-c_{\ell}}{\beta^{2}\left(\sigma^{2}+\lambda\kappa_{2}^{\ell}\right)}\right)\cdot\left[h_{\ell}\left(t\right)+\gamma e^{-r\left(T-t\right)}-\left(W_{\ell}^{\ast}\left(t^{-}\right)-g_{\ell}\left(t\right)\cdot\hat{W}\left(t^{-}\right)\right)\right]\nonumber \\
 &  & +\frac{1}{\beta}g_{\ell}\left(t\right)\left(\frac{\sigma^{2}+\lambda\kappa_{\chi}^{\ell,s}}{\sigma^{2}+\lambda\kappa_{2}^{\ell}}\right)\cdot\hat{\varrho}_{s}\left(t\right)\hat{W}\left(t^{-}\right),\label{eq: LETF optimal strategy}
\end{eqnarray}
where $g_{\ell}$ and $h_{\ell}$ are the following deterministic
functions,
\begin{eqnarray}
g_{\ell}\left(t\right) & = & \exp\left\{ K_{\beta}^{\ell,s}\cdot\int_{t}^{T}\hat{\varrho}_{s}\left(u\right)du\right\} ,\label{eq: function g for LETF}\\
h_{\ell}\left(t\right) & = & -\frac{q}{r}\left(1-e^{-r\left(T-t\right)}\right)+qe^{-r\left(T-t\right)}\cdot\int_{t}^{T}\exp\left\{ r\left(T-y\right)+K_{\beta}^{\ell,s}\cdot\int_{y}^{T}\hat{\varrho}_{s}\left(u\right)du\right\} dy,\label{eq: function h for LETF}
\end{eqnarray}
with constant $K_{\beta}^{\ell,s}$ given by 
\begin{eqnarray}
K_{\beta}^{\ell,s} & = & \mu-r-\frac{\left(\beta\left[\mu+\lambda\left(\kappa_{1}^{\ell}-\kappa_{1}^{s}\right)-r\right]-c_{\ell}\right)\left(\sigma^{2}+\lambda\kappa_{\chi}^{\ell,s}\right)}{\beta\left(\sigma^{2}+\lambda\kappa_{2}^{\ell}\right)}.\label{eq: Constant K beta for LETF}
\end{eqnarray}
\end{prop}

\begin{proof}
See Appendix \ref{subsec: Appendix proof of Optimal control for LETF}. 
\end{proof}
Note that values of $\kappa_{1}^{\ell}$, $\kappa_{2}^{\ell}$ and
$\kappa_{\chi}^{\ell,s}$ (see (\ref{eq: Kappas for LETF dynamics})),
which depend on the multiplier $\beta$ through the LETF jumps (\ref{eq: LETF jump multiplier}),
are required by the optimal strategy (\ref{eq: LETF optimal strategy}).
Expressions for these quantities can be derived in terms of the calibrated
parameters of the underlying asset dynamics (\ref{eq: B dynamics})-(\ref{eq: S dynamics})
without difficulty. As an illustration, Lemma \ref{lem: Deriving LETF kappas for Kou model}
in Appendix \ref{subsec: Appendix expressions for kappas} presents
expressions for $\kappa_{1}^{\ell}$, $\kappa_{2}^{\ell}$ and $\kappa_{\chi}^{\ell,s}$
in the case of the double-exponential Kou model (\cite{KouOriginal})
used for illustrating the results of this section, but we note that
this can also be done similarly for other jump diffusion models (e.g.
\cite{MertonJumps1976}).

The IR-optimal investment strategy for the VETF investor is given
in Corollary \ref{cor: Closed-form Optimal control for VETF}.
\begin{cor}
\label{cor: Closed-form Optimal control for VETF}(IR-optimal investment
strategy using the VETF) Let $\gamma>0$ be fixed. Suppose that Assumption
\ref{assu: Extra stylized-assumptions-for continuous rebalancing}
and wealth dynamics (\ref{eq: SDE W_hat}) and (\ref{eq: SDE W VETF})
apply. Let $W_{v}^{\ast}\left(t\right)$ denote the VETF investor's
wealth process (\ref{eq: SDE W VETF}) under the optimal strategy
$\varrho_{v}^{\ast}$, and let $\boldsymbol{X}_{v}^{\ast}\left(t\right)=\left(W_{v}^{\ast}\left(t\right),\hat{W}\left(t\right),\hat{\varrho}_{s}\left(t\right)\right)$.
Then the IR-optimal fraction of the investor's wealth invested in
the VETF, $\varrho_{v}^{\ast}$, satisfies
\begin{eqnarray}
\varrho_{v}^{\ast}\left(t,\boldsymbol{X}_{v}^{\ast}\left(t^{-}\right)\right)\cdot W_{v}^{\ast}\left(t^{-}\right) & = & \left(\frac{\mu-r-c_{v}}{\sigma^{2}+\lambda\kappa_{2}^{s}}\right)\cdot\left[h_{v}\left(t\right)+\gamma e^{-r\left(T-t\right)}-\left(W_{v}^{\ast}\left(t^{-}\right)-g_{v}\left(t\right)\cdot\hat{W}\left(t^{-}\right)\right)\right]\nonumber \\
 &  & +g_{v}\left(t\right)\cdot\hat{\varrho}_{s}\left(t\right)\hat{W}\left(t^{-}\right),\label{eq: VETF optimal strategy}
\end{eqnarray}
where $g_{v}$ and $h_{v}$ are the following deterministic functions,
\begin{eqnarray}
g_{v}\left(t\right) & = & \exp\left\{ c_{v}\cdot\int_{t}^{T}\hat{\varrho}_{s}\left(u\right)du\right\} ,\label{eq: function g for VETF}\\
h_{v}\left(t\right) & = & -\frac{q}{r}\left(1-e^{-r\left(T-t\right)}\right)+qe^{-r\left(T-t\right)}\cdot\int_{t}^{T}\exp\left\{ r\left(T-y\right)+c_{v}\cdot\int_{y}^{T}\hat{\varrho}_{s}\left(u\right)du\right\} dy.\label{eq: function h for VETF}
\end{eqnarray}
\end{cor}

\begin{proof}
See Appendix \ref{subsec: Appendix proof of Optimal control for VETF}.
\end{proof}
The following remark relates the results of Proposition \ref{prop: Closed-form Optimal control for LETF}
and Corollary \ref{cor: Closed-form Optimal control for VETF} to
the available results in the literature.

\begin{brem}

\label{rem: Relationship to results in literature} (Relationship
of Proposition \ref{prop: Closed-form Optimal control for LETF} and
Corollary \ref{cor: Closed-form Optimal control for VETF} to results
in the literature). In the special case of a VETF with zero expense
ratio $c_{v}=0$, the results of Corollary \ref{cor: Closed-form Optimal control for VETF}
imply that $g_{v}\left(t\right)=1$ and $h_{v}\left(t\right)=0$ for
all $t\in\left[t_{0}=0,T\right]$, so that (\ref{eq: VETF optimal strategy})
simplifies considerably to 
\begin{eqnarray}
\varrho_{v}^{\ast}\left(t,\boldsymbol{X}_{v}^{\ast}\left(t^{-}\right)\right)\cdot W_{v}^{\ast}\left(t^{-}\right) & = & \left(\frac{\mu-r}{\sigma^{2}+\lambda\kappa_{2}^{s}}\right)\cdot\left[\gamma e^{-r\left(T-t\right)}-\left(W_{v}^{\ast}\left(t^{-}\right)-\hat{W}\left(t^{-}\right)\right)\right]+\hat{\varrho}_{s}\left(t\right)\hat{W}\left(t^{-}\right),\label{eq: VETF zero cost IR-optimal strategy}
\end{eqnarray}
which corresponds to the IR-optimal investment strategy where direct
investment in the underlying equity market index $S$ is possible.
This special case (\ref{eq: VETF zero cost IR-optimal strategy})
can be found in \cite{PvSForsythLi2022_stochbm}, where the results
of \cite{GoetzmannEtAl2002,GoetzmannEtAl2007} are extended to the
case of jumps in the risky asset processes. Corollary \ref{cor: Closed-form Optimal control for VETF}
therefore extends this to the case of investing in the equity index
\textit{indirectly} via a VETF with a non-negligible expense ratio,
whereas Proposition \ref{prop: Closed-form Optimal control for LETF}
extends these results further to the case a LETF with multiplier $\beta>1$
and expense ratio $c_{\ell}$ referencing an equity index $S$ with
jump-diffusion dynamics.\qed

\end{brem}

To gain some intuition regarding the behavior of the IR-optimal investment
strategies in Proposition \ref{prop: Closed-form Optimal control for LETF}
and Corollary \ref{cor: Closed-form Optimal control for VETF}, we
compare the illustrative investment results from implementing these
strategies over a 10-year time horizon. We use a benchmark and  assets
as in Table \ref{tab: Closed-form solns - Candidate assets and benchmark},
illustrative investment parameters as in Table \ref{tab: Closed-form solns - investment params},
and a Kou model (\cite{KouOriginal}) is assumed for the jump diffusion
dynamics with calibrated parameters as in Appendix \ref{sec:Appendix - Source-data}
(Table \ref{tab:Calibrated asset parameters for analytical solutions}).
Since LETFs are a relatively recent invention, we follow the example
of \cite{BansalMarshall2015} in constructing a proxy LETF replicating
$\beta=2$ times the daily returns of a broad stock market index,
in this case using the CRSP VWD index, which is a capitalization-weighted
index consisting of all domestic stocks trading on major US exchanges,
with historical data available since January 1926. As in for example
\cite{BansalMarshall2015} and \cite{LeungSircar2015}, we assume
that the managers of the LETF do not have challenges in replicating
the underlying index, which is reasonable given the possibility of
designing replication strategies for LETFs that remain robust even
during periods of market volatility (see for example \cite{GuasoniMayerhofer2023}).
For more information on the source data and calibrated, inflation-adjusted
parameters, please refer to Appendix \ref{sec:Appendix - Source-data}.

\noindent 
\begin{table}[!tbh]
\caption{Closed-form solutions - Investment parameters for illustrating the
results. Note that the calibrated parameters for the jump-diffusion
process are given in Appendix \ref{sec:Appendix - Source-data} (Table
\ref{tab:Calibrated asset parameters for analytical solutions}),
while the underlying assets, benchmark and ETF expense ratios are
given in Table \ref{tab: Closed-form solns - Candidate assets and benchmark}.
\label{tab: Closed-form solns - investment params}}

\centering{}%
\begin{tabular}{|>{\centering}p{1.7cm}|>{\centering}p{1.7cm}|>{\centering}p{1.7cm}|>{\centering}p{1.7cm}|>{\centering}p{1.7cm}|}
\hline 
{\footnotesize{}Parameter} & {\footnotesize{}$T$} & {\footnotesize{}$w_{0}$} & {\footnotesize{}$q$} & {\footnotesize{}$\gamma$}\tabularnewline
\hline 
{\footnotesize{}Value} & {\footnotesize{}10 years} & {\footnotesize{}\$ 100} & {\footnotesize{}\$ 5 per year} & {\footnotesize{}125}\tabularnewline
\hline 
\end{tabular}
\end{table}

Figure \ref{fig: Matlab_heatmaps - Analytical solutions} illustrates
the IR-optimal proportion of wealth invested in the ETF as a function
of  time $t$ and the wealth difference $W_{k}^{\ast}\left(t\right)-\hat{W}\left(t\right)$,
$k\in\left\{ v,\ell\right\} $. In the case of the LETF investor,
Figure \ref{fig: Matlab_heatmaps - Analytical solutions}(a) illustrates
$\varrho_{\ell}^{\ast}$, whereas for the VETF investor, Figure \ref{fig: Matlab_heatmaps - Analytical solutions}(b)
illustrates $\varrho_{v}^{\ast}$, where both strategies use the same
target $\gamma=125$.

\noindent 
\begin{figure}[!tbh]
\noindent \begin{centering}
\subfloat[IR-optimal investment in LETF ($\beta=2$), $\varrho_{\ell}^{\ast}$]{\includegraphics[scale=0.8]{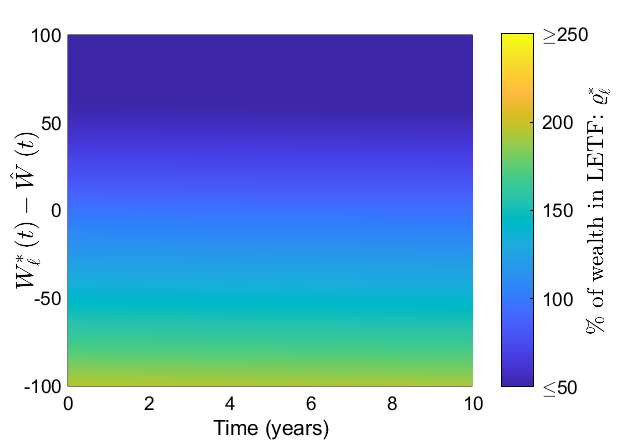}

}$\quad$\subfloat[IR-optimal investment in VETF, $\varrho_{v}^{\ast}$]{\includegraphics[scale=0.8]{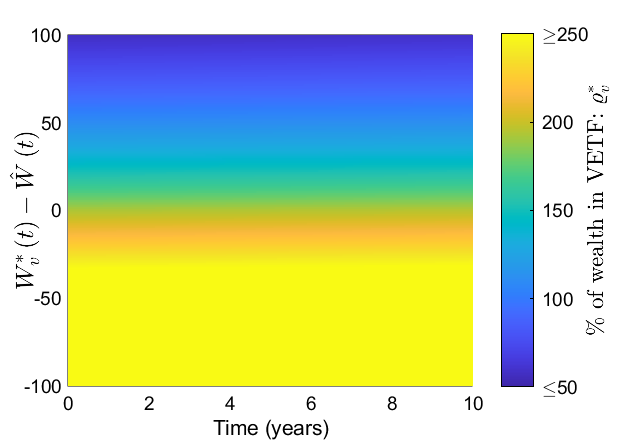}

}
\par\end{centering}
\caption{Closed-form IR-optimal investment strategies using the LETF ($\varrho_{\ell}^{\ast}$
as per (\ref{eq: LETF optimal strategy})) or the VETF ($\varrho_{v}^{\ast}$
as per (\ref{eq: VETF optimal strategy})) as a function of time $t$
and the wealth difference $W_{k}^{\ast}\left(t\right)-\hat{W}\left(t\right)$,
$k\in\left\{ v,\ell\right\} $, given the same implicit benchmark
outperformance target $\gamma$. The underlying assets, benchmark,
investment parameters and calibrated process parameters are as in
Table \ref{tab: Closed-form solns - Candidate assets and benchmark},
Table \ref{tab: Closed-form solns - investment params} and Table
\ref{tab:Calibrated asset parameters for analytical solutions}, respectively.
Note that the same color scale is used in both figures for comparison
purposes. \label{fig: Matlab_heatmaps - Analytical solutions}}
\end{figure}

Using the same strategies as illustrated in Figure \ref{fig: Matlab_heatmaps - Analytical solutions},
Figure \ref{fig: Matlab_heatmap - Ratio of optimal strategies (analytical)}
shows the \textit{ratio} of IR-optimal proportions of wealth invested
in the VETF relative to the investment in the LETF, $\varrho_{v}^{\ast}/\varrho_{\ell}^{\ast}$,
given an otherwise identical wealth difference $W_{k}^{\ast}\left(t\right)-\hat{W}\left(t\right),k\in\left\{ v,\ell\right\} $
at time $t$. We emphasize that both Figure \ref{fig: Matlab_heatmaps - Analytical solutions}
and Figure \ref{fig: Matlab_heatmap - Ratio of optimal strategies (analytical)}
treat the IR-optimal strategies from Proposition \ref{prop: Closed-form Optimal control for LETF}
and Corollary \ref{cor: Closed-form Optimal control for VETF} simply
as functions of time and the wealth difference relative to the benchmark.

\noindent 
\begin{figure}[!tbh]
\noindent \begin{centering}
\includegraphics[scale=0.8]{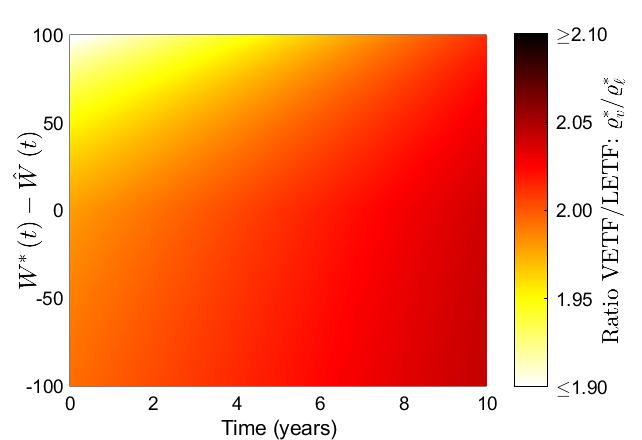}
\par\end{centering}
\caption{Ratio of IR-optimal proportions of wealth in the VETF vs. the LETF,
$\varrho_{v}^{\ast}/\varrho_{\ell}^{\ast}$, given identical wealth
differences relative to the benchmark $W^{\ast}\left(t\right)-\hat{W}\left(t\right)\equiv W_{v}^{\ast}\left(t\right)-\hat{W}\left(t\right)=W_{\ell}^{\ast}\left(t\right)-\hat{W}\left(t\right)$
at each time $t$ and same target $\gamma$. The underlying assets,
benchmark, investment parameters and calibrated process parameters
are as in Table \ref{tab: Closed-form solns - Candidate assets and benchmark},
Table \ref{tab: Closed-form solns - investment params} and Table
\ref{tab:Calibrated asset parameters for analytical solutions}, respectively.
\label{fig: Matlab_heatmap - Ratio of optimal strategies (analytical)}}
\end{figure}

With regards to Figure \ref{fig: Matlab_heatmaps - Analytical solutions}
and Figure \ref{fig: Matlab_heatmap - Ratio of optimal strategies (analytical)},
we make the following observations regarding the IR-optimal investment
strategies of the LETF investor vs. the VETF investor:
\begin{enumerate}
\item The IR-optimal investment strategy using a LETF (Figure \ref{fig: Matlab_heatmaps - Analytical solutions}(a))
is, like the strategy using a VETF (Figure \ref{fig: Matlab_heatmaps - Analytical solutions}(b)),
fundamentally \textit{contrarian}. Specifically, in the case of the
LETF, we observe that the IR-optimal proportion of wealth $\varrho_{\ell}^{\ast}$
in the LETF decreases as the wealth difference $W_{\ell}^{\ast}\left(t\right)-\hat{W}\left(t\right)$
increases, which happens after a period of strong LETF return performance.
In their analysis of reports to the SEC by institutional fund managers,
\cite{DeVaultEtAl2021} show that institutional investors indeed empirically
tend to decrease their holdings in LETFs following periods of strong
investment performance. While \cite{DeVaultEtAl2021} concludes that
this behavior might be explained as being a result of compensation-based
incentives, our results show that strategies based on maximizing the
IR (a widely-used investment metric) relative to a standard investment
benchmark could also be related to this empirical investment behavior.
\item Figure \ref{fig: Matlab_heatmap - Ratio of optimal strategies (analytical)}
shows that the IR-optimal strategies under stylized assumptions (Assumption
\ref{assu: Stylized-assumptions-for closed-form}) and identical implicit
benchmark outperformance target $\gamma$ satisfy $\varrho_{v}^{\ast}/\varrho_{\ell}^{\ast}\approx\beta=2$.
Informally, the LETF and VETF investors therefore take on nearly identical
``risk'' exposure to the movements of the underlying index, which
provides valuable intuition when interpreting the results Section
\ref{sec:Indicative-investment-results} where Assumption \ref{assu: Extra stylized-assumptions-for continuous rebalancing}
is relaxed.
\end{enumerate}
While Figure \ref{fig: Matlab_heatmaps - Analytical solutions} and
Figure \ref{fig: Matlab_heatmap - Ratio of optimal strategies (analytical)}
illustrate the IR-optimal strategy as a function of time and the wealth
difference relative to the benchmark, implementing this strategy in
a Monte Carlo simulation provides an additional perspective. Figure
\ref{fig: AnSolns_pctiles_risky_asset} shows the median and 95th
percentiles of the IR-optimal proportion of wealth invested in the
LETF and VETF, given the same benchmark outperformance target $\gamma$.
Note that under the stylized assumptions (Assumption \ref{assu: Extra stylized-assumptions-for continuous rebalancing}),
the LETF and VETF positions can be leveraged without restriction,
with leverage constraints only subsequently introduced in Section
\ref{sec:Numerical-solutions}. We observe that the corresponding
ETF exposure percentiles tend to approximately satisfy $\varrho_{v}^{\ast}/\varrho_{\ell}^{\ast}\approx\beta=2$,
with decreasing exposure over time due to the contrarian nature of
both strategies. 

\noindent 
\begin{figure}[!tbh]
\noindent \begin{centering}
\subfloat[$\varrho_{\ell}^{\ast}$ and $\varrho_{v}^{\ast}$: 95th percentiles
over time]{\includegraphics[scale=0.75]{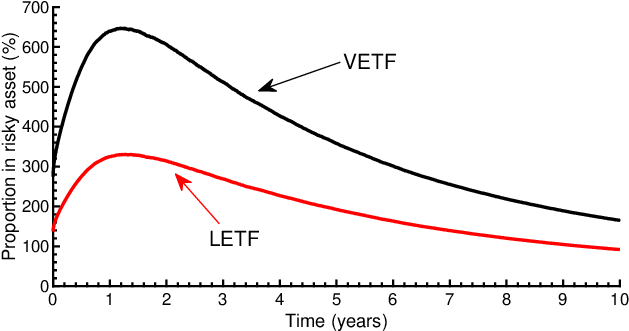}

}$\quad$$\quad$\subfloat[$\varrho_{\ell}^{\ast}$ and $\varrho_{v}^{\ast}$: 50th percentiles
over time]{\includegraphics[scale=0.75]{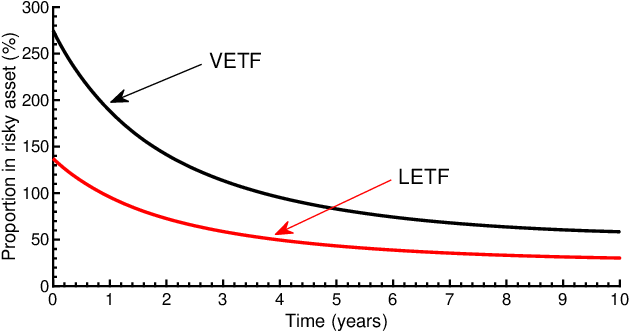}

}
\par\end{centering}
\caption{Closed-form IR-optimal investment strategies: 95th and 50th percentiles
over time of the IR-optimal proportion of wealth invested in the LETF
($\varrho_{\ell}^{\ast}$ as per (\ref{eq: LETF optimal strategy}))
or the VETF ($\varrho_{v}^{\ast}$ as per (\ref{eq: VETF optimal strategy}))
obtained using Monte Carlo simulation of the underlying dynamics (\ref{eq: SDE W_hat})-(\ref{eq: SDE W LETF}),
and same target $\gamma$. The underlying assets, benchmark, investment
parameters and calibrated process parameters are as in Table \ref{tab: Closed-form solns - Candidate assets and benchmark},
Table \ref{tab: Closed-form solns - investment params} and Table
\ref{tab:Calibrated asset parameters for analytical solutions}, respectively.
\label{fig: AnSolns_pctiles_risky_asset}}
\end{figure}

Figure \ref{fig: AnSolns_CDFs} compares the simulated CDFs of IR-optimal
terminal wealth $W_{k}^{\ast}\left(T\right),k\in\left\{ v,\ell\right\} $
and CDFs of the terminal wealth ratio relative to the benchmark $W_{k}^{\ast}\left(T\right)/\hat{W}\left(T\right),k\in\left\{ v,\ell\right\} $
for the LETF and VETF investors, respectively.

\noindent 
\begin{figure}[!tbh]
\noindent \begin{centering}
\subfloat[CDFs of $\hat{W}\left(T\right),W_{k}^{\ast}\left(T\right)$, $k\in\left\{ v,\ell\right\} $]{\includegraphics[scale=0.75]{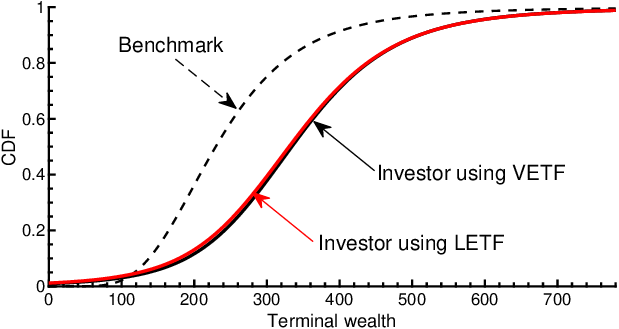}

}$\quad$$\quad$\subfloat[CDFs of ratio $W_{k}^{\ast}\left(T\right)/\hat{W}\left(T\right)$,
$k\in\left\{ v,\ell\right\} $ ]{\includegraphics[scale=0.75]{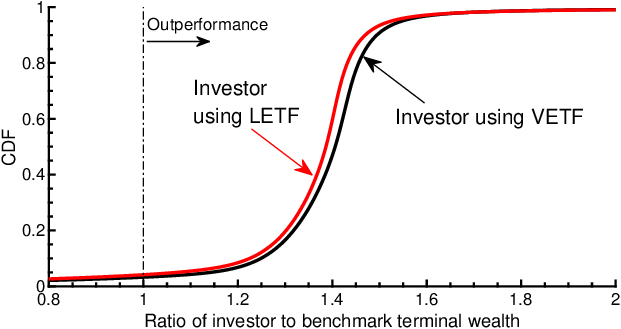}

}
\par\end{centering}
\caption{Closed-form IR-optimal investment strategies: CDFs of the IR-optimal
terminal wealth for the same target $\gamma$ obtained using Monte
Carlo simulation of the underlying dynamics and investing according
to optimal strategies (\ref{eq: LETF optimal strategy}) and (\ref{eq: VETF optimal strategy}).
The underlying assets, benchmark, investment parameters and calibrated
process parameters are as in Table \ref{tab: Closed-form solns - Candidate assets and benchmark},
Table \ref{tab: Closed-form solns - investment params} and Table
\ref{tab:Calibrated asset parameters for analytical solutions}, respectively.
\label{fig: AnSolns_CDFs}}
\end{figure}

Figure \ref{fig: AnSolns_CDFs} shows that due to implicitly similar
exposure levels to movements in the underlying equity index (since
$\varrho_{v}^{\ast}\simeq\beta\cdot\varrho_{\ell}^{\ast}$) given
identical targets $\gamma$ for the LETF and VETF investors, implementing
the strategies illustrated in Figure \ref{fig: Matlab_heatmaps - Analytical solutions}
result in nearly identical terminal wealth and outperformance outcomes.
In fact, under the stylized assumptions of this section, Proposition
\ref{prop: Special case ZERO costs} shows that in the special case
of (i) zero expense ratios and (ii) no jumps in the $S$-dynamics,
we have $\varrho_{v}^{\ast}\equiv\beta\cdot\varrho_{\ell}^{\ast}$,
the IR-optimal investor should be entirely indifferent as to whether
the LETF-based strategy (\ref{eq: LETF optimal strategy}) or VETF-based
strategy (\ref{eq: VETF optimal strategy}) is used for investment
purposes.
\begin{prop}
\label{prop: Special case ZERO costs}(Special case: Zero expense
ratios, no jumps) Let $\gamma>0$ be fixed, and let Assumption \ref{assu: Extra stylized-assumptions-for continuous rebalancing}
and wealth dynamics (\ref{eq: SDE W_hat})-(\ref{eq: SDE W LETF})
apply. If (i) both LETF ($\beta>1$) and the VETF have zero expense
ratios, i.e. $c_{v}=c_{\ell}=0$, and (ii) there are no jumps in the
underlying $S$-dynamics (i.e. $\lambda=0$ in (\ref{eq: S dynamics})),
then following results hold:

\begin{enumerate}

\item The IR-optimal proportion of wealth invested in the VETF (\ref{eq: VETF optimal strategy})
is equal to $\beta$ times the IR-optimal proportion of wealth invested
in the LETF (\ref{eq: LETF optimal strategy}),
\begin{eqnarray}
\varrho_{v}^{\ast}\left(t,\boldsymbol{X}_{v}^{\ast}\left(t\right)\right) & = & \beta\cdot\varrho_{\ell}^{\ast}\left(t,\boldsymbol{X}_{\ell}^{\ast}\left(t\right)\right),\qquad\forall t\in\left[t_{0},T\right],\label{eq: prop zero cost - LETF control a beta multiple}
\end{eqnarray}
where $\boldsymbol{X}_{v}^{\ast}\left(t\right)=\left(W_{v}^{\ast}\left(t\right),\hat{W}\left(t\right),\hat{\varrho}_{s}\left(t\right)\right)$
and $\boldsymbol{X}_{\ell}^{\ast}\left(t\right)=\left(W_{\ell}^{\ast}\left(t\right),\hat{W}\left(t\right),\hat{\varrho}_{s}\left(t\right)\right)$.

\item An IR-optimal investor with given target $\gamma>0$ would
be indifferent as to whether the optimal strategy is executed with
a LETF (\ref{eq: LETF optimal strategy}) or VETF (\ref{eq: VETF optimal strategy}),
since wealth outcomes are identical, 
\begin{eqnarray}
W_{\ell}^{\ast}\left(t\right) & = & W_{v}^{\ast}\left(t\right),\qquad\forall t\in\left[t_{0},T\right],\label{eq: prop zero cost - W equal}
\end{eqnarray}
and the same Information Ratio is obtained, 

\begin{equation}
\frac{E_{\varrho_{\ell}^{\ast}}^{t_{0},w_{0}}\left[W_{\ell}^{\ast}\left(T\right)-\hat{W}\left(T\right)\right]}{Stdev_{\varrho_{\ell}^{\ast}}^{t_{0},w_{0}}\left[W_{\ell}^{\ast}\left(T\right)-\hat{W}\left(T\right)\right]}=\frac{E_{\varrho_{v}^{\ast}}^{t_{0},w_{0}}\left[W_{v}^{\ast}\left(T\right)-\hat{W}\left(T\right)\right]}{Stdev_{\varrho_{v}^{\ast}}^{t_{0},w_{0}}\left[W_{v}^{\ast}\left(T\right)-\hat{W}\left(T\right)\right]}=\left(\exp\left\{ \left(\frac{\mu-r}{\sigma}\right)^{2}\cdot T\right\} -1\right)^{1/2}.\label{eq: Prop zero cost - IR optimal equal}
\end{equation}

\end{enumerate}
\end{prop}

\begin{proof}
See Appendix \ref{subsec: Appendix proof of optimal strategies ZERO COSTS}.
\end{proof}
We summarize the closed-form solutions results presented in this section
as follows. In the case of continuous rebalancing (i.e. rebalancing
times $\mathcal{T}=\left[t_{0},T\right]$) with no investment constraints,
the IR-optimal investor should be largely indifferent whether a VETF
or LETF is used on the same underlying equity index to execute the
IR-optimal investment strategy involving the ETF and T-bills for a
given outperformance target $\gamma$. Since the IR-optimal investment
strategies satisfy the approximate relationship $\varrho_{v}^{\ast}\simeq\beta\cdot\varrho_{\ell}^{\ast}$
(Figure \ref{fig: Matlab_heatmap - Ratio of optimal strategies (analytical)}),
when continuously rebalancing the LETF and VETF investors effectively
maintain similar implicit risk exposures at each time instant to movements
of the underlying index, resulting in broadly similar investment outcomes
(Figure \ref{fig: AnSolns_CDFs}), with differences between outcomes
entirely driven by different ETF expense ratios and the presence of
jumps (see Proposition \ref{prop: Special case ZERO costs}). 

At the other extreme, namely the lump-sum investment scenario with
no subsequent intervention (i.e. $\mathcal{T}=\left[t_{0}\right]$)
and one-quarter time horizon ($\Delta t=0.25$), we still observe
the approximate relationship $p_{v}^{\ast}\approx\beta\cdot p_{\ell}^{\ast}$
between corresponding IR-optimal strategies (Figure \ref{fig: Single period discrete rebalancing}).
However, in this case the power call-like payoff of the LETF can clearly
be observed, whereby the LETF investor benefits from an upside due
to comparatively inexpensive leverage while simultaneously enjoying
downside protection due to limited liability (Figures \ref{fig: Intuition - no jumps},
\ref{fig: Intuition - with jumps} and \ref{fig: Single period discrete rebalancing}). 

Far from ignoring the standard criticisms of LETFs in the literature
(see Section \ref{subsec:Intuition:-lump-sum investment scenario}),
the closed form results of this section - illustrated using parametric
dynamics calibrated to empirical market data - suggest that we should
not be entirely surprised that LETFs might have substantial appeal
to investors \textit{despite} these shortcomings. However, we emphasize
that none of the trading strategies illustrated - not even the lump-sum
investment scenario results - advocate for simplistic strategies like
buy-and-hold positions in the LETF over indefinite time horizons,
so a certain degree of sophistication on the part of the investor
is implicitly assumed. In addition, while the closed-form results
provide valuable intuition, they were derived under stylized assumptions,
which we relax in the subsequent sections to model the potential of
LETFs under more realistic conditions. 

\section{Numerical solutions\label{sec:Numerical-solutions}}

{\PVSedit{To allow for more general and practical conclusions than available in
the stylized setting of Section \ref{sec:Closed-form-solutions},
we use a data-driven neural network approach based on stationary block
bootstrap resampling of historical data (including proxy LETF returns)
since 1926. This ensures that the investment strategies and performance
analysis will incorporate all empirical aspects of actual returns,
including potentially sophisticated volatility dynamics, which may
not be reflected in the closed-form solutions of Section \ref{sec:Closed-form-solutions}. 

We start by formulating a more realistic investment setting with the
following characteristics: (i) Restrictions on the maximum leverage
and limitations on short-selling. (ii) An optional borrowing premium
applicable when short-selling an asset. (iii) Prohibition of trading
in insolvency. (iv) Infrequent (discrete) rebalancing of the portfolio. This is followed by a brief overview of a neural network-based numerical
solution approach to solve the IR problem (\ref{eq: IR objective USED})
in this setting. Indicative investment results obtained by implementing
these techniques on empirical market data are discussed in Section
\ref{sec:Indicative-investment-results}. }}

\subsection{Investment constraints and discrete rebalancing\label{subsec:Investment-constraints-and discrete rebal}}

Recall from Section \ref{sec:Problem-formulation} that the investor's
strategy is based on investing in a set of $N_{a}$ candidate assets
indexed by $i\in\left\{ 1,..,N_{a}\right\} $, while the benchmark
is defined in terms of $\hat{N}_{a}$ potentially different underlying
assets indexed by $j\in\left\{ 1,..,\hat{N}_{a}\right\} $. With investor
and benchmark strategies of the form (\ref{eq:Investor and benchmark investment strategies}),
and the benchmark strategy satisfying only the general assumptions
outlined in Assumption \ref{assu: General benchmark assumptions},
we assume that both the investor and benchmark portfolios are rebalanced
at each of $N_{rb}$ discrete rebalancing times during the investment
time horizon $\left[t_{0}=0,T\right]$. As a result, $\mathcal{\mathcal{T}}$
is now of the form

\begin{eqnarray}
\mathcal{\mathcal{T}} & = & \left\{ \left.t_{n}=n\Delta t\right|n=0,...,N_{rb}-1\right\} ,\qquad\Delta t=T/N_{rb}.\label{eq: Set of rebal times}
\end{eqnarray}
The assumption of equally-spaced rebalancing times in (\ref{eq: Set of rebal times})
is only for convenience, and can be relaxed without difficulty. At
each rebalancing time $t_{n}\in\mathcal{T}$, we assume a pre-specified
cash contribution $q\left(t_{n}\right)$ is made to the investor portfolio,
with the contribution also being added to the benchmark portfolio
to ensure the comparison in performance remains appropriate.

There is no need to specify any parametric dynamics for the underlying
assets in the numerical solution approach, which only requires the
availability of empirical market data for deriving the optimal strategy
(see Subsection \ref{subsec: NN solution approach} below). Specifically,
we assume that at each time $t_{n+1}\in\mathcal{\mathcal{T}}\cup T$
we can observe $R_{i}\left(t_{n}\right)$ and $\hat{R}_{j}\left(t_{n}\right)$,
the returns on investor asset $i\in\left\{ 1,..,N_{a}\right\} $ and
benchmark asset $j\in\left\{ 1,..,\hat{N}_{a}\right\} $, respectively,
over the time interval $\left[t_{n},t_{n+1}\right]$. Note that these
returns might be inflation-adjusted and might include a borrowing
premium applicable to assets that have been shorted (see e.g. Assumption
\ref{assu: Leverage assumptions} below). As a result, for the purposes
of numerical solutions, the investor and benchmark wealth dynamics
are respectively of the form

\begin{eqnarray}
W\left(t_{n+1}^{-}\right) & = & \left[W\left(t_{n}^{-}\right)+q\left(t_{n}\right)\right]\cdot\sum_{i=1}^{N_{a}}p_{i}\left(t_{n},\boldsymbol{X}\left(t_{n}^{-}\right)\right)\cdot\left[1+R_{i}\left(t_{n}\right)\right],\label{eq: W dynamics discrete - INVESTOR}\\
\hat{W}\left(t_{n+1}^{-}\right) & = & \left[\hat{W}\left(t_{n}^{-}\right)+q\left(t_{n}\right)\right]\cdot\sum_{j=1}^{\hat{N}_{a}}\hat{p}_{j}\left(t_{n},\hat{\boldsymbol{X}}\left(t_{n}^{-}\right)\right)\cdot\left[1+\hat{R}_{j}\left(t_{n}\right)\right],\label{eq: W dynamics discrete - BENCHMARK}
\end{eqnarray}
where $W\left(t_{0}^{-}\right)=\hat{W}\left(t_{0}^{-}\right)\coloneqq w_{0}>0$
and $n=0,...,N_{rb}-1$. 

Since active funds often have restrictions on leverage and short-selling
(see for example \cite{ForsythVetzalWestmacott2019,NiLiForsyth2023_LFNN}),
these constraints are included in the formulation. 

The investor's candidate asset $i=1$, assumed to be 30-day T-bills
in the indicative investment results of Section \ref{sec:Indicative-investment-results},
plays a special role in leveraged portfolios. The investor is assumed
to be able to short-sell this asset with a potential borrowing premium
payable, i.e. the investor can borrow funds at an approximation of
the prevailing short-term interest rate plus a borrowing premium to
fund leveraged investments in the other assets. {\PVSnew{In addition, in the
case of insolvency, defined as occurring when the investor wealth
is negative, $W\left(t_{n}\right)<0$, we will assume that the negative wealth (i.e. the outstanding debt) is placed
in asset $i=1$, where it grows at the rate of return of this asset
with an addition of a possible borrowing 
premium. Note that this effectively implies that trading ceases when $W<0$, either until maturity $T$ or until such a time where the cash injections pay off the debt resulting in $W>0$, in which case trading can resume.}}
A maximum leverage ratio at a portfolio level of $p_{max}$ is also
assumed, where typical values are in the range $p_{max}\in\left[1.0,1.5\right]$.
Assumption \ref{assu: Leverage assumptions} outlines the details
more formally. 

\begin{assumption}

\label{assu: Leverage assumptions}(Investor strategy: Leverage restrictions,
borrowing premium and no trading in insolvency) The following assumptions
and restrictions apply to the investor strategy, where the investor
considers investment in $N_{a}\geq2$ candidate assets. As discussed,
the investor's set of candidate assets may not correspond to the assets
included in the benchmark strategy.
\begin{enumerate}
\item \textbf{Shortable and long-only assets}: Only investor candidate asset
$i=1$ will (potentially) be shorted, with the remaining investor
candidate assets $i\in\left\{ 2,..,N_{a}\right\} $ being long only.
In other words, at any rebalancing time $t_{n}\in\mathcal{T}$, we
have 
\begin{eqnarray}
\left(\textrm{Shortable asset }i=1\right): &  & p_{1}\left(t_{n},\boldsymbol{X}\left(t_{n}^{-}\right)\right)\in\mathbb{R},\qquad,t_{n}\in\mathcal{T},\label{eq: Ass - shortable asset}\\
\left(\textrm{Long-only assets }i\in\left\{ 2,..,N_{a}\right\} \right): &  & p_{i}\left(t_{n},\boldsymbol{X}\left(t_{n}^{-}\right)\right)\geq0,\qquad i\in\left\{ 2,..,N_{a}\right\} ,t_{n}\in\mathcal{T},\label{eq: Assu - long only assets}\\
\left(\textrm{All wealth invested}\right): &  & \sum_{i=1}^{N_{a}}p_{i}\left(t_{n},\boldsymbol{X}\left(t_{n}^{-}\right)\right)=1,\qquad t_{n}\in\mathcal{T}.\label{eq: Assu - all wealth invested}
\end{eqnarray}
\item \textbf{Borrowing premium} $b\geq0$: If investor candidate asset
$i=1$ is shorted at time $t_{n}$ (i.e. if $p_{1}\left(t_{n},\boldsymbol{X}\left(t_{n}^{-}\right)\right)<0$),
then a constant borrowing premium $b\geq0$ is added to the returns
on asset $i=1$ over the time interval $\left[t_{n},t_{n+1}\right]$
to be paid by the investor. In other words, for asset $i=1$, the
return $R_{1}\left(t_{n}\right)$ incorporated in (\ref{eq: W dynamics discrete - INVESTOR})
is of the form 
\begin{eqnarray}
\left(\textrm{Borrowing premium}\right): &  & R_{1}\left(t_{n}\right)=\begin{cases}
\overline{R}_{1}\left(t_{n}\right), & \textrm{if }p_{1}\left(t_{n},\boldsymbol{X}\left(t_{n}^{-}\right)\right)\geq0\\
\overline{R}_{1}\left(t_{n}\right)+b, & \textrm{if }p_{1}\left(t_{n},\boldsymbol{X}\left(t_{n}^{-}\right)\right)<0,
\end{cases}\label{eq: Assu - Return on Asset 1 with borrowing premium}
\end{eqnarray}
where $\overline{R}_{1}\left(t_{n}\right)$ is the (possibly inflation-adjusted)
return on underlying asset $i=1$ over $\left[t_{n},t_{n+1}\right]$
without any added premiums. For long-only assets, we simply have $R_{i}\left(t_{n}\right)=\overline{R}_{i}\left(t_{n}\right)$,
$i\in\left\{ 2,..,N_{a}\right\} ,t_{n}\in\mathcal{T}$. Note that
in the case of the benchmark strategy, no borrowing premium is applicable
to any asset due to Assumption \ref{assu: Benchmark leverage assumptions}
below.
\item \textbf{Maximum leverage ratio} $p_{max}$: The total allocated proportion
of wealth to the long-only assets $i\in\left\{ 2,..,N_{a}\right\} $
cannot exceed the maximum leverage ratio $p_{max}$ , 
\begin{eqnarray}
\left(\textrm{Maximum leverage ratio}\right): &  & \sum_{i=2}^{N_{a}}p_{i}\left(t_{n},\boldsymbol{X}\left(t_{n}^{-}\right)\right)\leq p_{max},\qquad t_{n}\in\mathcal{T}.\label{eq: Assu - Max leverage ratio}
\end{eqnarray}
\item \textbf{No trading in insolvency}: If the investor wealth is negative,
i.e. if $W\left(t_{n}\right)<0$ at any $t_{n}\in\mathcal{T}$, then
all long asset positions (\ref{eq: Assu - long only assets}) are
liquidated and the total debt (the amount $W\left(t_{n}\right)<0$)
is allocated to the shortable asset (\ref{eq: Ass - shortable asset}).
In such a scenario, no further trading occurs for the remainder of
the investment time horizon ($t_{m}\in\mathcal{T},t_{m}>t_{n}$), unless
cash injections pay off the debt, and the portfolio wealth becomes positive.
Total debt accumulates at a rate (\ref{eq: Assu - Return on Asset 1 with borrowing premium})
which possibly includes a borrowing premium. 
More formally, 
\begin{eqnarray}
\left(\textrm{No trading in insolvency}\right): &  & \textrm{If }W\left(t_{n}^{-}\right)<0\quad\Rightarrow\quad\boldsymbol{p}\left(t_{n},\boldsymbol{X}\left(t_{n}^{-}\right)\right)=\boldsymbol{e}_{1},\quad t_{n}\in\mathcal{T},\label{eq: Assu - No trading in insolvency}
\end{eqnarray}
where $\boldsymbol{e}_{1}=\left(1,0,...,0\right)\in\mathbb{R}^{N_{a}}$
is the standard basis vector $\mathbb{R}^{N_{a}}$ with 1 in the first
position (corresponding to $i=1$, the shortable asset as per (\ref{eq: Ass - shortable asset}))
and all other entries are zero. \qed
\end{enumerate}
\end{assumption}

Note that (\ref{eq: Assu - No trading in insolvency}) also implies
$\boldsymbol{p}\left(t_{m},\boldsymbol{X}\left(t_{m}^{-}\right)\right)=\boldsymbol{e}_{1}$
for all $t_{m}>t_{n}$, so that no further trading does indeed occur
in the case of insolvency as required by Assumption \ref{assu: Leverage assumptions}(iv). 

Recalling that $\mathcal{A}$ denotes the set of admissible controls
and $\mathcal{Z}$ denoting the admissible control space, Assumption
\ref{assu: Leverage assumptions} implies that we have the following
form for $\mathcal{Z}$ and $\mathcal{A}$, respectively: 

\begin{eqnarray}
\mathcal{Z} & = & \left\{ \boldsymbol{z}\in\mathbb{R}^{N_{a}}\left|\begin{array}{c}
z_{1}\in\mathbb{R},\\
z_{i}\geq0,\forall i\in\left\{ 2,..,N_{a}\right\} ,\\
\sum_{i=1}^{N_{a}}z_{i}=1,\\
\sum_{i=2}^{N_{a}}z_{i}\leq p_{max}.
\end{array}\right.\right\} ,\label{eq: Set Z with multiple constraints}
\end{eqnarray}
and 
\begin{eqnarray}
\mathcal{A} & = & \left\{ \mathcal{P}=\left\{ \boldsymbol{p}\left(t_{n},\boldsymbol{X}\left(t_{n}\right)\right),t_{n}\in\mathcal{T}\left|\begin{array}{c}
\boldsymbol{p}\left(t_{n},\boldsymbol{X}\left(t_{n}\right)\right)\in\mathcal{Z},\textrm{ if }W\left(t_{n}^{-}\right)\geq0,\\
\boldsymbol{p}\left(t_{n},\boldsymbol{X}\left(t_{n}\right)\right)=\boldsymbol{e}_{1},\textrm{ if }W\left(t_{n}^{-}\right)<0.
\end{array}\right.\right\} \right\} .\label{eq: Set A with multiple constraints}
\end{eqnarray}

Note that with slight abuse of notation in (\ref{eq: Set A with multiple constraints}),
$\mathcal{Z}$ is the admissible control space in the case of solvency
only.

Finally, in order ensure that the benchmarks align with typical investment
benchmarks used in practice (see Remark \ref{rem: Clarifying general assumptions benchmarks})
and to avoid pathological examples, Assumption \ref{assu: Benchmark leverage assumptions}
below specifies that no short-selling is allowed in the case of the
benchmark strategy.

\begin{assumption}

\label{assu: Benchmark leverage assumptions} (Benchmark: leverage
restrictions) The benchmark strategy does not engage in the short-selling
of any asset.
\begin{eqnarray}
\left(\textrm{Long-only benchmark}\right) &  & \hat{p}_{j}\left(t_{n},\hat{\boldsymbol{X}}\left(t_{n}^{-}\right)\right)\geq0,\qquad\forall j=1,...,\hat{N}_{a}.\label{eq: Assu - benchmark long only}
\end{eqnarray}
As a result, the benchmark strategy has an implicit maximum leverage
ratio of $p_{max}=1$, with no borrowing premium being applicable,
while benchmark insolvency is ruled out in the sense that $\hat{W}\left(t_{n}\right)\geq0$
for all $t_{n}\in\mathcal{T}$ given dynamics (\ref{eq: W dynamics discrete - BENCHMARK}).
\qed

\end{assumption}

\subsection{Neural network solution approach\label{subsec: NN solution approach}}

The objective function in the case of the numerical solutions remains
of the form (\ref{eq: IR objective USED}), 
\begin{eqnarray}
\left(IR\left(\gamma\right)\right): &  & \inf_{\mathcal{P}\in\mathcal{A}}E_{\mathcal{P}}^{t_{0},w_{0}}\left[\left(W\left(T\right)-\left[\hat{W}\left(T\right)+\gamma\right]\right)^{2}\right],\qquad\gamma>0,\label{eq: IR objective for numerical}
\end{eqnarray}
with the main differences from the treatment in Section \ref{sec:Closed-form-solutions}
being the following: (i) The set of admissible controls $\mathcal{A}$
is now given by (\ref{eq: Set A with multiple constraints}). (ii)
Rebalancing occurs at a strict discrete subset of times $t_{n}\in\mathcal{T}\subset\left[t_{0}=0,T\right]$.
(iii) As discussed below, we no longer need the assumption of parametric
models for the underlying assets, but can use market data directly.

To solve (\ref{eq: IR objective for numerical}) numerically to obtain
the optimal investment strategy $\mathcal{P}^{\ast}\in\mathcal{A}$,
we follow the neural network-based solution approach of \cite{NiLiForsyth2023_LFNN},
where a ``leverage-feasible neural network'' (LFNN) is constructed
to approximate the investment strategy directly as a feedback control
$\left(t_{n},\boldsymbol{X}\left(t_{n}\right)\right)\rightarrow\mathcal{P}\left(t_{n},\boldsymbol{X}\left(t_{n}\right)\right)\coloneqq\boldsymbol{p}\left(t_{n},\boldsymbol{X}\left(t_{n}\right)\right),\forall t_{n}\in\mathcal{T}$
in the case of admissible sets of the form (\ref{eq: Set Z with multiple constraints})-(\ref{eq: Set A with multiple constraints}).
This approach forms part of a class of methods (see, for example,
\cite{HanWeinan2016,BuehlerGononEtAl2018,ReppenSoner2023,
ReppenSonerEtAl2023,PvsForsythLi2023_NN,PvSForsythLi2022_stochbm,Jari_2024},
)
that does not require dynamic programming to solve problems such as
(\ref{eq: IR objective USED}), thereby avoiding the typical challenges
such as evaluating high-dimensional conditional expectations and error
amplifications over time-stepping.

Since more detailed information, including a convergence analysis,
can be found in \cite{NiLiForsyth2023_LFNN}, we give only a very
short overview of the application of the LFNN approach in our setting.
In this approach, the control \textit{function} $\left(t_{n},\boldsymbol{X}\left(t_{n}\right)\right)\rightarrow\boldsymbol{p}\left(t_{n},\boldsymbol{X}\left(t_{n}\right)\right)$
is approximated by a single neural network (NN) with at least 3 features
(inputs), namely $\left(t_{n},\boldsymbol{X}\left(t_{n}\right)\right)=\left(t_{n},W\left(t_{n}\right),\hat{W}\left(t_{n}\right)\right)$.
Note that additional features such as trading signals can be incorporated
in the NN inputs in settings where this is considered valuable. Let
$\boldsymbol{F}(t,\boldsymbol{X}(t);\boldsymbol{\theta})\equiv\boldsymbol{F}(\cdot,\boldsymbol{\theta})$
denote the NN, where $\boldsymbol{\theta}\in\mathbb{R}^{\eta_{\theta}}$
is the NN parameters, i.e. the NN weights and biases. Since the time
$t_{n}$ is used is an input into the NN, a single parameter vector
$\boldsymbol{\theta}$ (equivalently, a single NN) is applicable to
all rebalancing times, identifying this as a ``global-in-time''
approach in the taxonomy of \cite{HuLauriere2023}. One of the key
contributions of \cite{NiLiForsyth2023_LFNN} is to construct the
NN $\boldsymbol{F}(\cdot,\boldsymbol{\theta})$ with an output layer
that guarantees, for all inputs $\left(t,\boldsymbol{X}\left(t\right)\right)=\left(t,W\left(t\right),\hat{W}\left(t\right)\right)$,
that
\begin{eqnarray}
\left(t,\boldsymbol{X}\left(t\right)\right) & \rightarrow & \boldsymbol{F}(t,\boldsymbol{X}(t);\boldsymbol{\theta})\left|\begin{array}{c}
\boldsymbol{F}(t,\boldsymbol{X}(t);\boldsymbol{\theta})\in\mathcal{Z},\textrm{ if }W\left(t\right)\geq0,\\
\boldsymbol{F}(t,\boldsymbol{X}(t);\boldsymbol{\theta})=\boldsymbol{e}_{1},\textrm{ if }W\left(t\right)<0.
\end{array}\right.\label{eq: NN automatically satisfies constraints}
\end{eqnarray}
As a result of (\ref{eq: Set A with multiple constraints}), by using
the approximation 
\begin{eqnarray}
\boldsymbol{p}(t,\boldsymbol{X}(t)) & \simeq & \boldsymbol{F}(t,\boldsymbol{X}(t);\boldsymbol{\theta})\equiv\boldsymbol{F}(\cdot,\boldsymbol{\theta}),\label{eq: NN control approx intuitive}
\end{eqnarray}
we can therefore approximate the investor strategies as $\mathcal{P}=\left\{ \boldsymbol{F}(t_{n},\boldsymbol{X}(t_{n});\boldsymbol{\theta}),~t_{n}\in\mathcal{T}\right\} $
while being assured that $\mathcal{P}\in\mathcal{A}$ where $\mathcal{A}$
is as per (\ref{eq: Set A with multiple constraints}), without the
need to impose constraints on the optimization problem itself. As
a result, (\ref{eq: IR objective for numerical}) can be solved as
an \textit{unconstrained} optimization problem over $\boldsymbol{\theta}\in\mathbb{R}^{\eta_{\theta}}$,
\begin{equation}
\inf_{\boldsymbol{\theta}\in\mathbb{R}^{\eta_{\theta}}}E_{\boldsymbol{F}\left(\cdot;\boldsymbol{\theta}\right)}^{t_{0},w_{0}}\left[\left(W\left(T;\boldsymbol{\theta}\right)-\left[\hat{W}\left(T\right)+\gamma\right]\right)^{2}\right],\label{eq: Unconstrained optimization}
\end{equation}
where the approximation (\ref{eq: NN control approx intuitive}) is
used to obtain the asset allocation in the investor wealth dynamics
(\ref{eq: W dynamics discrete - INVESTOR}) which depend on $\boldsymbol{\theta}$. 

Parametric models for the underlying asset dynamics are no longer
required. Instead, we use a finite set of samples from the set $Y=\left\{ Y^{\left(j\right)}:j=1,...,N_{d}\right\} $,
where each element $Y^{\left(j\right)}$ denotes a time series of
\textit{joint} asset return observations $R_{i},i\in\left\{ 1,..,N_{a}\right\} $,
possibly adjusted for inflation and the application of a borrowing
premium, observed at each $t_{n}\in\mathcal{T}$. $Y$ represents
the training data of the NN, so any $\boldsymbol{\theta}\in\mathbb{R}^{\eta_{\theta}}$
and returns path $Y^{\left(j\right)}\in Y$, the wealth dynamics (\ref{eq: W dynamics discrete - INVESTOR})
with approximation (\ref{eq: NN control approx intuitive}) generates
a terminal wealth outcome $W^{\left(j\right)}\left(T;\boldsymbol{\theta}\right)$.
The expectation in (\ref{eq: Unconstrained optimization}) is then
approximated simply by 

\begin{equation}
\min_{\boldsymbol{\theta}\in\mathbb{R}^{\eta_{\theta}}}\left\{ \frac{1}{N_{d}}\sum_{j=1}^{N_{d}}\left(W^{\left(j\right)}\left(T;\boldsymbol{\theta}\right)-\left[\hat{W}^{\left(j\right)}\left(T\right)+\gamma\right]\right)^{2}\right\} ,\label{eq: Approx of unconstrained optimization}
\end{equation}
where the optimal parameter vector $\boldsymbol{\theta}^{\ast}$ is
obtained using stochastic gradient descent. The resulting IR-optimal
strategy for (\ref{eq: IR objective for numerical}) consistent with
the constraints as outlined in Assumption \ref{assu: Leverage assumptions}
is therefore given by $\boldsymbol{p}^{\ast}(\cdot,\boldsymbol{X}(\cdot))\simeq\boldsymbol{F}(\cdot,\boldsymbol{\theta}^{\ast})$. 

While the details underlying the construction of the data set $Y$
are clearly of practical significance, we note that the approach of
\cite{NiLiForsyth2023_LFNN} remains agnostic as to the how $Y$ is
constructed. It can be obtained using for example GAN-generated data
sets (see e.g. \cite{YoonTimeGAN2019,PvsForsythLi2023_NN}), or using
Monte Carlo simulations if the underlying dynamics are specified for
ground truth analysis purposes (see e.g. \cite{PvSForsythLi2022_stochbm}),
or a version of bootstrap resampling of empirical market data, as
we now discuss.

In practical applications, the use of empirical market data might
be preferred for the construction of $Y$. However, since only a single
historical path of asset returns is available, some form of data augmentation
is typically used to obtain sufficiently rich training and testing
data. For illustrative purposes, in Section \ref{sec:Indicative-investment-results}
we use stationary block bootstrap resampling (\cite{politis1994})
to construct $Y$. This technique, designed for weakly stationary
time series with serial dependence, is both popular in academic settings
(\cite{Cederburg_2022}) and practical applications (\cite{CogneauZakalmouline2013,dichtl2016,Scott_2017,Scott_2022,Simonian_2022}).
Note that bootstrap resampling methods have been proposed for non-stationary
time series (\cite{Politis_2003_a}, \cite{Politis_1999_b}), but
this is not used in the illustrative investment results of Section
\ref{sec:Indicative-investment-results}. 

\section{Indicative investment results\label{sec:Indicative-investment-results}}

{\PVSedit{This section demonstrates the potential of including LETFs within long-term, diversified, dynamic, IR-optimal investment strategies subject to the investment constraints outlined in Section \ref{sec:Numerical-solutions}. To ensure that the results are realistic and practical, we focus entirely on using bootstrapped historical returns (including the proxy LETF returns series dating back to 1926, see Appendix \ref{sec:Appendix - Source-data}), and we apply the numerical approach discussed in Subsection \ref{subsec: NN solution approach} to obtain the IR-optimal strategies based on the combination of T-bills and T-bonds with a LETF or a VETF on a broad equity market index.}}

\subsection{Investment scenarios\label{subsec:Investment-scenarios}}

The key investment parameters used for illustrative purposes throughout
this section are outlined in Table \ref{tab: Ind Inv Results - investment params}.
Note in particular that we use a relatively long investment time horizon
(10 years) coupled with relatively infrequent (quarterly) rebalancing,
and that the same implicit outperformance target $\gamma$ is used
in each of the scenarios to facilitate a fair comparison. This value
of $\gamma$ is chosen for general illustrative purposes only, and
the conclusions remain qualitatively similar for different choices
of $\gamma$.

\noindent 
\begin{table}[!tbh]
\caption{Key investment parameters for the illustrative results of Section
\ref{sec:Indicative-investment-results}. \label{tab: Ind Inv Results - investment params}}

\centering{}%
\begin{tabular}{|>{\centering}p{2cm}|>{\centering}p{2cm}|>{\centering}p{3.5cm}|>{\centering}p{2cm}|>{\centering}p{2cm}|>{\centering}p{2cm}|}
\hline 
{\footnotesize{}Parameter} & {\footnotesize{}$T$} & {\footnotesize{}\# rebalancing events}{\footnotesize\par}

{\footnotesize{}($N_{rb}$)} & {\footnotesize{}Initial wealth ($w_{0}$)} & {\footnotesize{}Contributions}{\footnotesize\par}

{\footnotesize{}($q_{n}$)} & {\footnotesize{}Target}{\footnotesize\par}

{\footnotesize{}($\gamma$)}\tabularnewline
\hline 
{\footnotesize{}Value} & {\footnotesize{}10 years} & {\footnotesize{}40}{\footnotesize\par}

{\footnotesize{}(quarterly rebalancing)} & {\footnotesize{}\$ 100} & {\footnotesize{}\$ 5 per year}{\footnotesize\par}

{\footnotesize{}(\$1.25/quarter)} & {\footnotesize{}125}\tabularnewline
\hline 
\end{tabular}
\end{table}

Table \ref{tab: Ind Inv Results - Candidate assets and benchmark}
provides an overview of the benchmark and the investor's candidate
assets. A 70/30 benchmark strategy is again used, since it aligns
to the definition of popular investment benchmarks used in practice
(see Remark \ref{rem: Clarifying general assumptions benchmarks}).
Note that the benchmark is defined in terms of the broad equity market
index (``Market'') with 70\% of the wealth allocation, with the
remaining 30\% split between 30-day T-bills and 10-year T-bonds. As
in Section \ref{sec:Closed-form-solutions}, we assume that the investor
cannot invest directly in the broad equity market index (``Market''),
but can gain exposure to this index via a VETF (expense ratio $c_{v}=0.06\%$)
or a LETF with multiplier $\beta=2$ (expense ratio $c_{\ell}=0.89\%$). 

\noindent 
\begin{table}[!tbh]
\caption{Candidate assets and benchmark for the illustrative results of Section
\ref{sec:Indicative-investment-results}. A mark \textquotedblleft$\checkmark$\textquotedblright{}
indicates that an asset is available for inclusion. Note that the
investor cannot invest directly in the market portfolio (\textquotedblleft Market\textquotedblright ),
but only indirectly via either the VETF or LETF, whereas the benchmark
is defined directly in terms of \textquotedblleft Market\textquotedblright{}
in alignment with popular investment benchmarks. \label{tab: Ind Inv Results - Candidate assets and benchmark}}

\begin{tabular}{|>{\centering}p{1.2cm}|>{\raggedright}p{7.5cm}|>{\centering}p{1.5cm}||>{\centering}p{2cm}|>{\centering}p{2cm}|}
\hline 
\multicolumn{2}{|>{\centering}p{7.5cm}|}{{\footnotesize{}Underlying assets}} & \multirow{2}{1.5cm}{{\footnotesize{}Benchmark}} & \multicolumn{2}{>{\centering}p{4cm}|}{{\footnotesize{}Investor candidate assets}}\tabularnewline
\cline{1-2} \cline{2-2} \cline{4-5} \cline{5-5} 
{\footnotesize{}Label} & {\footnotesize{}Asset description} &  & {\footnotesize{}Using VETF} & {\footnotesize{}Using LETF}\tabularnewline
\hline 
{\footnotesize{}T30} & {\footnotesize{}30-day Treasury bill} & {\footnotesize{}15\%} & {\footnotesize{}$\checkmark$} & {\footnotesize{}$\checkmark$}\tabularnewline
\hline 
{\footnotesize{}B10} & {\footnotesize{}10-year Treasury bond} & {\footnotesize{}15\%} & {\footnotesize{}$\checkmark$} & {\footnotesize{}$\checkmark$}\tabularnewline
\hline 
{\footnotesize{}Market} & {\footnotesize{}Market portfolio (broad equity market index)} & {\footnotesize{}70\%} & {\footnotesize{}-} & {\footnotesize{}-}\tabularnewline
\hline 
{\footnotesize{}VETF} & {\footnotesize{}Vanilla (unleveraged) ETF replicating the returns
of the market portfolio, with expense ratio $c_{v}=0.06\%$} & {\footnotesize{}-} & {\footnotesize{}$\checkmark$} & {\footnotesize{}-}\tabularnewline
\hline 
{\footnotesize{}LETF} & {\footnotesize{}Leveraged ETF with daily returns replicating $\beta=2$
times the daily returns of the market portfolio, with expense ratio
$c_{\ell}=0.89\%$} & {\footnotesize{}-} & {\footnotesize{}-} & {\footnotesize{}$\checkmark$}\tabularnewline
\hline 
\end{tabular}
\end{table}

Table \ref{tab: Ind Inv Results - Leverage and borrowing premium scenarios}
provides more detail on the leverage and borrowing premium scenarios
considered, where we highlight the following:
\begin{itemize}
\item Investor portfolios formed with a LETF are never leveraged ($p_{max}=1.0$),
whereas portfolios formed with a VETF can use leverage up to a portfolio
maximum of $p_{max}\in\left\{ 1.0,1.2,1.5,2.0\right\} $ via the short-selling
of 30-day T-bills (i.e. borrowing funds to invest in the VETF) with
a borrowing premium $b\in\left\{ 0,0.03\right\} $ potentially being
applicable. This is done in order to compare the performance of an
IR-optimal portfolio with a LETF and no portfolio-level leverage with
that of an IR-optimal portfolio formed with a (potentially) leveraged
VETF under various leverage assumptions.
\item In terms of the selection of values for $p_{max}\in\left\{ 1.0,1.2,1.5,2.0\right\} $
in the case of the VETF investor, note that Regulation T of the US
Federal Reserve board requires at least 50\% of the initial price
of a stock position to be available on deposit, while brokerage firms
are free to establish more stringent requirements. For the VETF investor,
for illustrative purposes we will therefore mostly focus on the cases
of $p_{max}=1.0$ (no leverage) or $p_{max}=1.5$, and for comparison
purposes provide the additional examples using $p_{max}=1.2$ and
$p_{max}=2.0$ in Appendix \ref{sec:Appendix - Additional numerical results}. 
\item In terms of the selection of borrowing premiums $b\in\left\{ 0,0.03\right\} $
for the VETF investor, we first note that all returns are inflation-adjusted
(see Appendix \ref{sec:Appendix - Source-data}), and so these quantities
should be interpreted net of inflation. The case of zero borrowing
premium ($b=0$) is provided for comparison purposes only, while the
value of $b=3\%$ is obtained from the examples in \cite{NiLiForsyth2023_LFNN},
where it is based on a consideration of the average real return for
T-bills and the average inflation-adjusted corporate bond yields for
Moody's Aaa and Baa-rated bond issues.
\end{itemize}
\noindent 
\begin{table}[!tbh]
\caption{Maximum leverage and borrowing premium scenarios for the indicative
investment results of Section \ref{sec:Indicative-investment-results}.
\label{tab: Ind Inv Results - Leverage and borrowing premium scenarios}}

\begin{tabular}{|>{\raggedright}p{7cm}|>{\centering}p{2cm}||>{\centering}p{3cm}|>{\centering}p{3cm}|}
\hline 
\multirow{2}{7cm}{{\footnotesize{}Component of leverage scenario}} & \multirow{2}{2cm}{{\footnotesize{}Benchmark}} & \multicolumn{2}{>{\centering}p{6cm}|}{{\footnotesize{}Investor candidate assets}}\tabularnewline
\cline{3-4} \cline{4-4} 
 &  & {\footnotesize{}Using VETF} & {\footnotesize{}Using LETF}\tabularnewline
\hline 
{\footnotesize{}Maximum portfolio-level leverage ratio $p_{max}$} & {\footnotesize{}No leverage allowed ($p_{max}=1.0$)} & {\footnotesize{}No leverage allowed ($p_{max}=1.0$) }{\footnotesize\par}

{\footnotesize{}as well as scenarios}{\footnotesize\par}

{\footnotesize{}$p_{max}\in\left\{ 1.2,1.5,2.0\right\} $} & {\footnotesize{}No leverage allowed ($p_{max}=1.0$)}\tabularnewline
\hline 
{\footnotesize{}Shortable asset to fund leveraged position (if applicable)} & {\footnotesize{}-} & {\footnotesize{}T30} & {\footnotesize{}-}\tabularnewline
\hline 
{\footnotesize{}Borrowing premiums: Scenarios for premium $b$ over
T30 return on leveraged positions (if applicable)} & {\footnotesize{}N/a} & {\footnotesize{}$b=0$ or $b=0.03$} & {\footnotesize{}N/a}\tabularnewline
\hline 
\end{tabular}
\end{table}

The underlying data sets for the training and testing of the neural
network giving the IR-optimal investment strategies using stationary
block bootstrap resampling of empirical market data (see Section \ref{sec:Numerical-solutions}
and Appendix \ref{sec:Appendix - Source-data}) instead of calibrated
process dynamics. In particular, we use all available inflation-adjusted
market data over the time period January 1926 to December 2023, together
with an expected block size of 3 months, to obtain 500,000 jointly
bootstrapped asset return paths (see \cite{PvsForsythLi2023_NN,LiForsyth2019}
and Appendix \ref{sec:Appendix - Source-data} for more information).
As in \cite{NiLiForsyth2023_LFNN}, we use a shallow NN (2 hidden
layers) with only the minimal input features $\left(t,\boldsymbol{X}\left(t\right)\right)=\left(t,W\left(t\right),\hat{W}\left(t\right)\right)$,
since that has been found sufficient to obtain a stable and accurate
IR-optimal investment strategy in a setting where no additional market
signals are used as inputs.

For illustrative purposes, we also present the investment results
obtained from investing according to the IR-optimal investment strategy
on selected historical data paths. This is discussed in more detail
in Remark \ref{rem: Performance on historical data path}. 

\begin{brem}

\label{rem: Performance on historical data path}(Performance on single
historical data paths) Since the future evolution of asset returns
are not expected to replicate the past evolution of returns \textit{precisely},
we consider illustrative investment results based on bootstrapped
data sets as significantly more informative than using a single historical
data path of asset returns to illustrate performance. 

However, for purposes of concreteness and intuition, we do show the
evolution of the LETF and VETF investor wealth obtained by implementing
the corresponding IR-optimal portfolios and the benchmark on four
historical data paths each spanning a period equal to the investment
time horizon of 10 years:
\begin{enumerate}
\item January 2000 until December 2009, which illustrates the impact on
the portfolio wealth of both the DotCom bubble crash as well as the
GFC period.
\item January 2005 until December 2014, which focuses on the GFC and the
subsequent period of relatively slow market recovery.
\item January 2010 until December 2019, which illustrates the performance
during the bull market of the 2010s, a period of very low interest
rates and therefore cheap leverage.
\item January 2014 until December 2023, which combines an initial period
of strong growth and low interest rates with the Covid-19 period and
subsequent recovery, only to be followed by the bear market for stocks
lasting from January to October 2022 and higher interest rates.
\end{enumerate}
Note that while the historical path of returns enter the training
data of the NN indirectly via bootstrap resampling, the probability
that the actual historical data path itself appearing in the resulting
bootstrapped data sets is vanishingly small (see \cite{NiLiForsyth2020}
for a proof), so that the historical data paths can themselves be
considered as effectively ``out-of-sample'' for testing purposes.
However, we emphasize that in this section the main focus remains
on the investment results based on the much richer bootstrapped data
set of returns data, which ensures a meaningful discussion of the
implications for wealth \textit{distributions}, for example, rather
than individual wealth values from a single historical path. \qed

\end{brem}

\subsection{Comparison of investment results\label{subsec: Comparison of investment results}}

{\PVSedit{Figure \ref{fig: NumSolns Sc1 CDFs} illustrates the distributions
of the IR-optimal terminal investor wealth $W_{k}^{\ast}\left(T\right),k\in\left\{ v,\ell\right\} $
for different portfolios formed under the leverage scenarios
as per Table \ref{tab: Ind Inv Results - Leverage and borrowing premium scenarios},
as well as the distribution of the benchmark terminal wealth $\hat{W}\left(T\right)$.
With regards to Figure \ref{fig: NumSolns Sc1 CDFs}(a), we see that
the IR-optimal strategy using the LETF (and bonds) achieves \textit{partial
stochastic dominance} (\cite{Atkinson1987,PvSDangForsyth2019_Distributions,NiLiForsyth2023_LFNN}) over the IR-optimal
strategies using the VETF (and bonds), even if the VETF investment
can be leveraged\footnote{Note that IR-optimal strategies incorporating either the VETF or the LETF achieve partial
stochastic dominance over the benchmark, which is to be expected considering
the results of \cite{PvSForsythLi2022_stochbm}.}. The main driver of the difference in performance of the LETF-based strategy relative to that of the VETF-based strategy in this setting is a combination of (i) the call-like payoff of the LETF as underlying asset over a relatively short time horizon (e.g. 1 quarter) and (ii) the contrarian nature of the discretely-rebalanced IR-optimal investment strategy locking in the gains from rebalancing. The observation that holding a
LETF position for a quarter amounts to holding a ``continuously rebalanced''
position in the equity index and bonds which results in a power law-type
payoff (see Section \ref{subsec:Intuition:-lump-sum investment scenario}), also holds for the historical data underlying these results (see Appendix \ref{subsec: Appendix Bootstrapped quarterly returns} for additional analyses).}}

\noindent 
\begin{figure}[!tbh]
\noindent \begin{centering}
\subfloat[CDFs of $W^{\ast}\left(T\right),$$W_{k}^{\ast}\left(T\right),k\in\left\{ v,\ell\right\} $]{\includegraphics[scale=0.75]{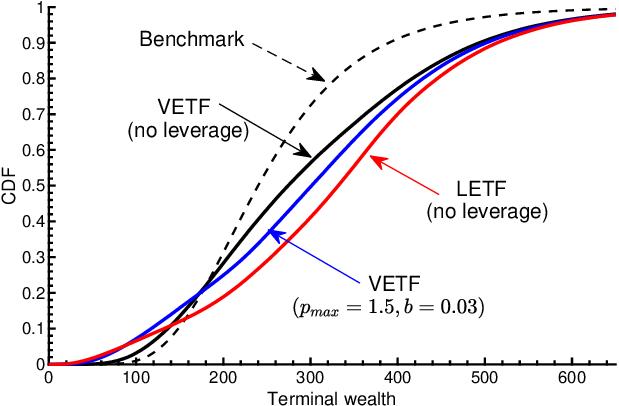}

}$\quad$$\quad$\subfloat[Extreme left tail: CDFs of $\hat{W}\left(T\right)$, $W_{k}^{\ast}\left(T\right),k\in\left\{ v,\ell\right\} $]{\includegraphics[scale=0.75]{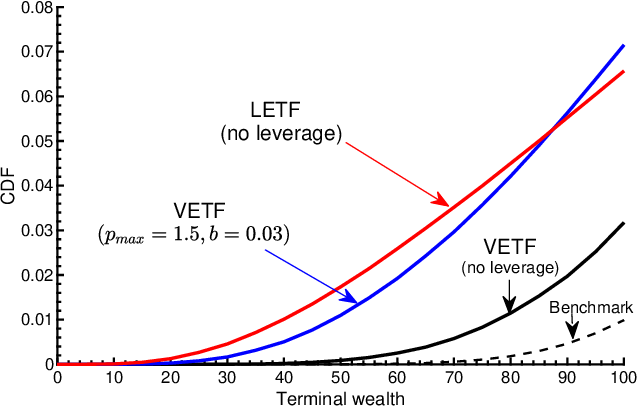}

}
\par\end{centering}
\caption{{\PVSedit{CDFs of IR-optimal terminal wealth $W_{k}^{\ast}\left(T\right),k\in\left\{ v,\ell\right\} $
and the benchmark terminal wealth $\hat{W}\left(T\right)$.}}\label{fig: NumSolns Sc1 CDFs}}
\end{figure}

However, as Figure \ref{fig: NumSolns Sc1 CDFs}(b) illustrates, there
is no free lunch with regards to leverage, similar to what we observed
in Figures \ref{fig: Intuition - no jumps}, \ref{fig: Intuition - with jumps}
and \ref{fig: Single period discrete rebalancing}. In more detail,
when considering the extreme left tails of the IR-optimal terminal
wealth CDFs, whether leveraging an investment implicitly (via the
LETF) or explicitly (via a leveraged VETF investment), the downside
wealth outcomes are worse than using the VETF with no leverage (or
simply the benchmark). Note that this is based on the distribution
of empirical market data together with the implementation of IR-optimal
investment strategies, and is consistent with the observations
regarding the downside protection offered by the LETFs in truly extreme
cases (illustrated in Figures \ref{fig: Single period discrete rebalancing}). 

{\PVSedit{
While Figure \ref{fig: NumSolns Sc1 CDFs} considers a maximum leverage
ratio of $p_{max}=1.5$ with a borrowing premium for short-selling
of $b=3\%$ (applicable to the leveraged VETF position), the results
of Figure \ref{fig: Appendix NumSolns Sc1 - W_T CDFs - Lev and borrowing costs}
in Appendix \ref{sec:Appendix - Additional numerical results} show
that qualitatively similar results are obtained even if leverage is
allowed to increase to $p_{max}=2.0$. However, Appendix \ref{sec:Appendix - Additional numerical results}
shows that in the unrealistic case where the borrowing premium on
short-selling is reduced to zero (i.e. if the investor pays interest
on short positions at exactly the T-bill rate), then a leveraged VETF-based
strategy with $p_{max}=2.0$ generates results that are comparable to (though \textit{slightly} better than) the LETF-based strategy that uses no additional leverage. These
results are to be expected given the difference in VETF and LETF expense ratios, since in the case of continuous rebalancing,
no investment constraints and zero costs, Proposition \ref{prop: Special case ZERO costs}
shows that the IR-optimal investment results are identical, regardless
of whether the investor uses a LETF with multiplier $\beta$ or a
VETF leveraged $\beta$ times. While the underlying assumptions of
Proposition \ref{prop: Special case ZERO costs} are clearly violated in the
setting of this section and the results in Appendix \ref{sec:Appendix - Additional numerical results},
this observation underscores the practical relevance of the closed-form
solutions in the stylized setting of Section \ref{sec:Closed-form-solutions}.
}}

Figure \ref{fig: NumSolns Sc1 BM outperformance} focuses on different
measures of benchmark outperformance rather than investor wealth,
with Figure \ref{fig: NumSolns Sc1 BM outperformance}(a) illustrating
the CDF of the terminal 
pathwise wealth ratio $W_{k}^{\ast}\left(T\right)/\hat{W}\left(T\right),k\in\left\{ v,\ell\right\} $
and Figure \ref{fig: NumSolns Sc1 BM outperformance}(b) illustrating
the probability of benchmark outperformance over time. It is clear
that IR-optimal portfolios formed using the LETF and no further leverage
significantly improves the benchmark outperformance characteristics
of the resulting strategy. {\PVSedit{Note that the results of Appendix \ref{sec:Appendix - Additional numerical results}
(Figure \ref{fig: Appendix NumSolns Sc1 - Ratio CDFs - Lev and borrowing costs}
and Figure \ref{fig: Appendix NumSolns Sc1 BM outperformance - LEVERAGE and COSTS})
show that the conclusions of Figure \ref{fig: NumSolns Sc1 BM outperformance}
remain qualitatively applicable for different leverage scenarios for the VETF strategy, with the LETF strategy being slightly outperformed by a VETF-based strategy only in the specific and unrealistic case of a zero borrowing premium
and maximum leverage of $p_{max}=2.0$.}}

\noindent 
\begin{figure}[!tbh]
\noindent \begin{centering}
\subfloat[CDF of ratio $W_{k}^{\ast}\left(T\right)/\hat{W}\left(T\right),k\in\left\{ v,\ell\right\} $]{\includegraphics[scale=0.75]{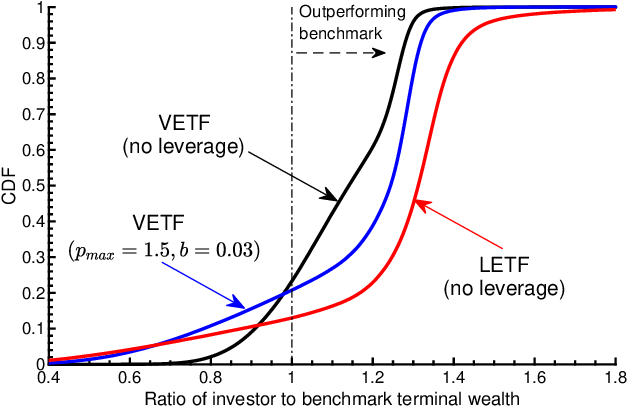}

}$\quad$$\quad$\subfloat[Probability $W_{k}^{\ast}\left(t\right)>\hat{W}\left(t\right),k\in\left\{ v,\ell\right\} $
as a function of $t$]{\includegraphics[scale=0.75]{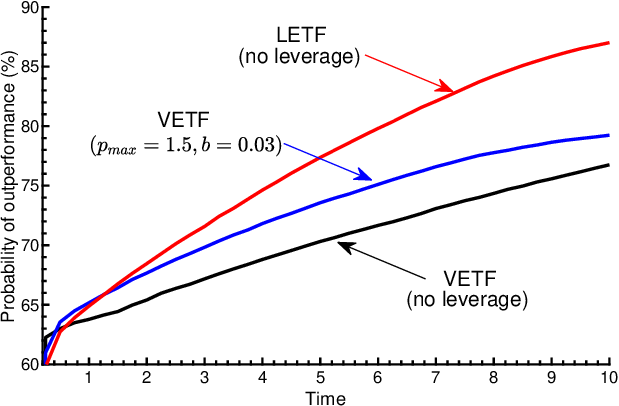}

}
\par\end{centering}
\caption{{\PVSedit{CDFs of IR-optimal terminal wealth ratios $W_{k}^{\ast}\left(T\right)/\hat{W}\left(T\right),k\in\left\{ v,\ell\right\} $,
and probability $W_{k}^{\ast}\left(t\right)>\hat{W}\left(t\right),k\in\left\{ v,\ell\right\} $
of benchmark outperformance as a function of time $t$.}} \label{fig: NumSolns Sc1 BM outperformance}}
\end{figure}

Figure \ref{fig: NumSolns Sc1 Pctiles in ETFs} compares selected
percentiles of the IR-optimal proportion of wealth in the LETF (no
leverage) or leveraged VETF $\left(p_{max}=1.5,b=0.03\right)$ over
time, using the same scale in Figure \ref{fig: NumSolns Sc1 Pctiles in ETFs}(a)
and Figure \ref{fig: NumSolns Sc1 Pctiles in ETFs}(b) for illustrative
purposes. Figure \ref{fig: NumSolns Sc1 Pctiles in ETFs}(a) shows
that the LETF investor initially allocates around 70\% of wealth allocated
to LETF, which falls to around 40\% or less for both the 20th and 50th percentiles.
In the case of the portfolio with a leveraged position in the VETF
$\left(p_{max}=1.5\right)$, Figure \ref{fig: NumSolns Sc1 Pctiles in ETFs}(b)
shows that the median allocation to the VETF exceeds 100\% of wealth
for more than half the investment time horizon. 

\noindent 
\begin{figure}[!tbh]
\noindent \begin{centering}
\subfloat[Proportion of wealth in LETF (no leverage)]{\includegraphics[scale=0.75]{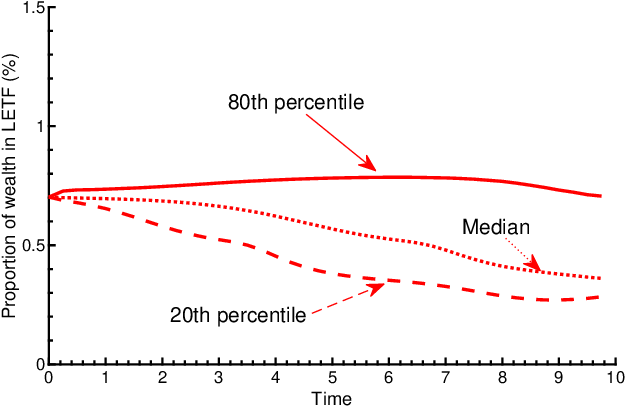}

}$\quad$$\quad$\subfloat[Proportion of wealth in VETF $\left(p_{max}=1.5,b=0.03\right)$]{\includegraphics[scale=0.75]{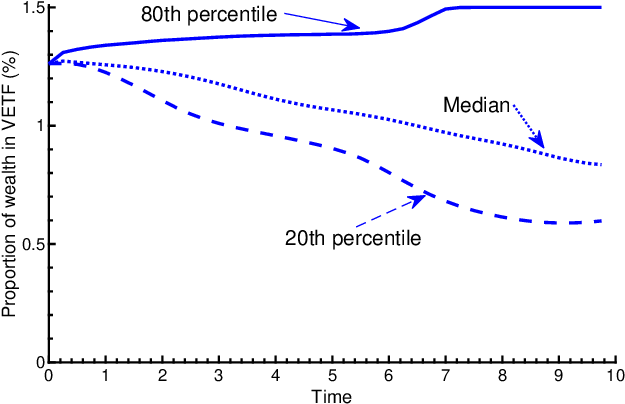}

}
\par\end{centering}
\caption{{\PVSedit{Selected percentiles of the IR-optimal proportion of wealth in the
LETF (no leverage) compared to the corresponding percentiles invested in the leveraged VETF $\left(p_{max}=1.5,b=0.03\right)$
over time. Note the same scale on the y-axis has been used to facilitate
comparison.}} \label{fig: NumSolns Sc1 Pctiles in ETFs}}
\end{figure}

Figure \ref{fig: NumSolns Sc1 Pctiles in ETFs} shows that executing
the IR-optimal strategy using a LETF offers the investor more flexibility:
with lower levels of wealth tied up in the LETF compared to the allocation
to a VETF, together with the higher volatility of LETF-based returns,
the investor can ``lock in'' periods of good past returns by implementing
a systematic de-risking of the portfolio to a lower allocation to
LETF over time, increasing the allocation to bonds. While the leveraged
VETF-based strategy essentially follows the same contrarian pattern,
the 80th percentile in Figure \ref{fig: NumSolns Sc1 Pctiles in ETFs}(b)
shows that it significantly harder for the leveraged VETF strategy
to recover from periods of poor past returns in a setting of maximum
leverage restrictions, borrowing costs and no trading in the event
of bankruptcy. 

{\PVSnew{In addition to the preceding results which are based on the bootstrap resampling of historical data,
for illustrative purposes  we consider the investment performance on
 selected single historical paths (see Remark \ref{rem: Performance on historical data path}) illustrated in Figure \ref{fig: NumSolns Sc1 Historical Path}.}}
Figures \ref{fig: NumSolns Sc1 Historical Path}(a) and (b) show that
both the LETF and VETF investors (regardless of VETF leverage) underperform
the benchmark during the lowest points of the DotCom and GFC crashes,
with the LETF investor experiences larger peak-to-trough declines
but also faster post-crash recovery. Figure \ref{fig: NumSolns Sc1 Historical Path}(c)
illustrate that the LETF-based strategy remains only slightly ahead
of the VETF-based strategies during periods of strong equity market
performance and low interest rates, while Figure \ref{fig: NumSolns Sc1 Historical Path}(d)
shows that the LETF investor stays ahead despite the significant impact
on portfolio wealth of the Covid-19 period and subsequent bear market
of 2022.

\noindent 
\begin{figure}[!tbh]
\noindent \begin{centering}
\subfloat[Investing Jan 2000 - Dec 2009]{\includegraphics[scale=0.75]{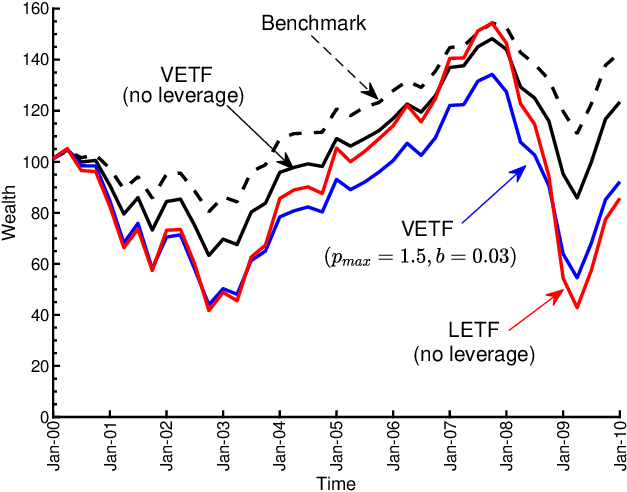}

}$\quad$$\quad$\subfloat[Investing Jan 2005 - Dec 2014]{\includegraphics[scale=0.75]{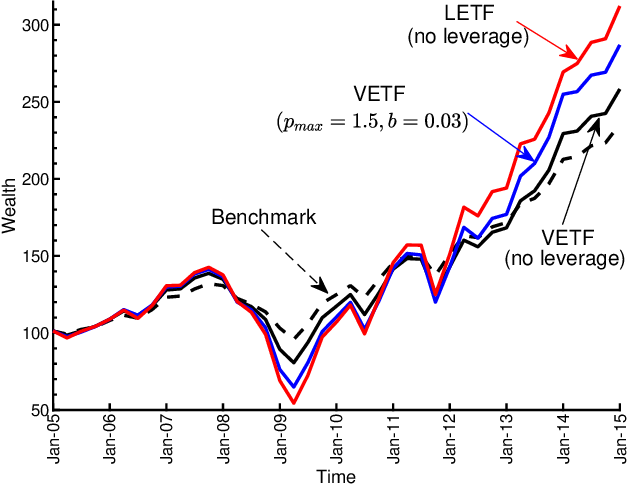}

}
\par\end{centering}
\noindent \begin{centering}
\subfloat[Investing Jan 2010 - Dec 2019]{\includegraphics[scale=0.75]{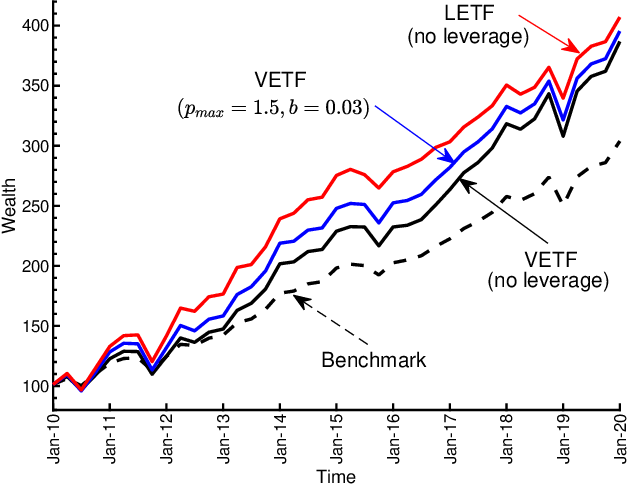}

}$\quad$$\quad$\subfloat[Investing Jan 2014 - Dec 2023]{\includegraphics[scale=0.75]{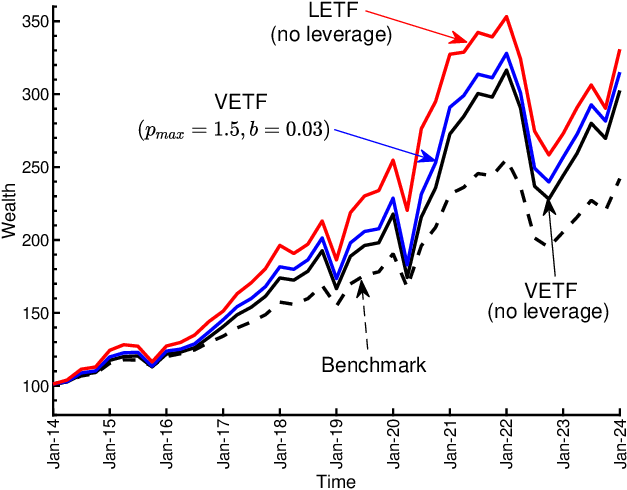}

}
\par\end{centering}
\caption{{\PVSedit{Evolution of portfolio wealth over time when investing according to
the corresponding IR-optimal investment strategies on historical paths
selected. See Remark \ref{rem: Performance on historical data path} for a discussion regarding the paths selected.}}
\label{fig: NumSolns Sc1 Historical Path}}
\end{figure}

\section{Conclusion\label{sec:Conclusion}}

In this paper, we investigated the potential of including a broad
stock market index-based leveraged ETF (LETF) in long-term, dynamically-optimal investment
strategies designed to maximize the outperformance over standard performance
benchmarks in terms of the information ratio (IR). 

Using both closed-form and numerical solutions, we showed that an
investor can exploit the observation that LETFs offer call-like payoffs,
and therefore could be a convenient way to add inexpensive leverage
to the portfolio while providing extreme downside protection. 

Under stylized assumptions including continuous rebalancing and no investment constraints, we derived the closed-form IR-optimal investment strategy for the LETF investor, which provided valuable intuition as to the contrarian nature of the strategy. 

{\PVSedit{To allow for more general and practical conclusions, we use a data-driven neural network approach based on stationary block bootstrap resampling of historical data (including proxy LETF returns) since 1926. This ensures that the investment strategies and performance analysis incorporate all empirical aspects of actual returns, including potentially sophisticated volatility dynamics. We derive IR-optimal strategies that allow for quarterly trading, leverage restrictions, no trading in the event of insolvency and the presence of margin costs on borrowing. Our findings show that unleveraged IR-optimal strategies with a broad stock market LETF not only outperform the benchmark more often than possibly leveraged IR-optimal strategies derived using a VETF, but can achieve partial stochastic dominance over the benchmark and (leveraged or unleveraged) VETF-based strategies in terms of terminal wealth.}}

Two important caveats are to be kept in mind regarding our results
demonstrating the potential of LETFs: (i) The results and conclusions
are associated with dynamic IR-optimal investment strategies, which
are most emphatically \textit{not} naive strategies like the buy-and-hold
strategies over long time horizons often considered in the literature
(see the Introduction for a discussion). In particular, critical to
the investment outcomes are the rebalancing of the portfolio within
the context of a contrarian investment strategy. (ii) The results
emphasize that there is no free lunch with regards to leverage. Specifically,
the extreme left tails of the IR-optimal terminal wealth CDFs confirm
that whether leveraging an investment implicitly (via the LETF) or
explicitly (via a leveraged VETF investment), the downside wealth
outcomes are worse than using the VETF without any leverage, and therefore
the upside outcomes of leverage is not without significant risks. Nevertheless, bootstrap resampling tests indicate
that use of an optimal strategy using LETFs outperforms the benchmark
$>95\%$ of the time, which may make the extreme tail risk acceptable.

Despite the controversy surrounding the uses of LETFs for investment
purposes in the literature, our results  help to explain
the empirical appeal of LETFs to institutional and retail investors
alike, and encourage a reconsideration of the role of broad stock
market LETFs within the context of more sophisticated investment strategies.

\section{Declaration}
Forsyth's work was supported by the Natural Sciences and Engineering Research Council of Canada (NSERC)
grant RGPIN-2017-03760.  Li's work was supported by NSERC grant RGPIN-2020-04331.

The authors have no conflicts of interest to report.

\noindent \setlength{\bibsep}{1pt plus 0.3ex} 
\small

\bibliographystyle{mynatbib}
\bibliography{References_v29}

\begin{thebibliography}{75}
\expandafter\ifx\csname natexlab\endcsname\relax\def\natexlab#1{#1}\fi
\expandafter\ifx\csname url\endcsname\relax
  \def\url#1{\texttt{#1}}\fi
\expandafter\ifx\csname urlprefix\endcsname\relax\def\urlprefix{URL }\fi

\bibitem[{Ahn et~al.(2015)Ahn, Haugh, and Jain}]{Ahn_2015}
Ahn, A., M.~Haugh, and A.~Jain (2015).
\newblock Consistent pricing of options on leveraged {ETFs}.
\newblock \emph{SIAM Journal on Financial Mathematics} 6, 559--593.

\bibitem[{Alekseev and Sokolov(2016)}]{AlekseevSokolov2016}
Alekseev, A.~G. and M.~V. Sokolov (2016).
\newblock Benchmark-based evaluation of portfolio performance: a
  characterization.
\newblock \emph{Annals of Finance} 12, 409--440.

\bibitem[{Anarkulova et~al.(2022)Anarkulova, Cederburg, and
  O'Doherty}]{Cederburg_2022}
Anarkulova, A., S.~Cederburg, and M.~S. O'Doherty (2022).
\newblock Stocks for the long run? {Evidence} from a broad sample of developed
  markets.
\newblock \emph{Journal of Financial Economics} 143:1, 409--433.

\bibitem[{Applebaum(2004)}]{ApplebaumBook}
Applebaum, D. (2004).
\newblock \emph{L{\'e}vy processes and stochastic calculus}.
\newblock Cambridge University Press.

\bibitem[{Atkinson(1987)}]{Atkinson1987}
Atkinson, A. (1987).
\newblock On the measurement of poverty.
\newblock \emph{Econometrica} 55(4), 749--764.

\bibitem[{Avellaneda and Zhang(2010)}]{AvellanedaZhang2010}
Avellaneda, M. and S.~Zhang (2010).
\newblock Path-dependence of leveraged {ETF} returns.
\newblock \emph{SIAM Journal on Financial Mathematics} 1(1), 586--603.

\bibitem[{{Bajeux-Besnainou} et~al.(2013){Bajeux-Besnainou}, Portait, and
  Tergny}]{BajeuxEtAl2013}
{Bajeux-Besnainou}, I., R.~Portait, and G.~Tergny (2013).
\newblock Optimal portfolio allocations with tracking error volatility and
  stochastic hedging constraints.
\newblock \emph{Quantitative Finance} 13(10), 1599--1612.

\bibitem[{Bansal and Marshall(2015)}]{BansalMarshall2015}
Bansal, V.~K. and J.~F. Marshall (2015).
\newblock A tracking error approach to leveraged {ETFs}: Are they really that
  bad?
\newblock \emph{Global Finance Journal} 26, 47--63.

\bibitem[{Bednarek and Patel(2022)}]{BednarekPatel2022}
Bednarek, Z. and P.~Patel (2022).
\newblock Just say no to leveraged {ETFs}.
\newblock \emph{Journal of Investment Management} 20(3), 53--69.

\bibitem[{Bensoussan et~al.(2014)Bensoussan, Wong, Yam, and
  Yung}]{BensoussanEtAl2014}
Bensoussan, A., K.~C. Wong, S.~C.~P. Yam, and S.~P. Yung (2014).
\newblock Time-consistent portfolio selection under short-selling prohibition:
  From discrete to continuous setting.
\newblock \emph{SIAM Journal on Financial Mathematics} 5, 153--190.

\bibitem[{Bertrand and Prigent(2022)}]{Bertrand2022}
Bertrand, P. and J.-L. Prigent (2022).
\newblock On the diversification and rebalancing returns: performance
  comparison of constant mix versus buy and hold strategies SSRN 4153690.

\bibitem[{Bolshakov and Chincarini(2020)}]{BolshakovChincarini2020}
Bolshakov, A. and L.~B. Chincarini (2020).
\newblock Manager skill and portfolio size with respect to a benchmark.
\newblock \emph{European Financial Management} 26(1), 176--197.

\bibitem[{Buehler et~al.(2019)Buehler, Gonon, Teichmann, and
  Wood}]{BuehlerGononEtAl2018}
Buehler, H., L.~Gonon, J.~Teichmann, and B.~Wood (2019).
\newblock Deep hedging.
\newblock \emph{Quantitative Finance} 19(8), 1271--1291.

\bibitem[{{Canadian Pension Plan}(2022)}]{cpp_site}
{Canadian Pension Plan} (2022).
\newblock Our investment strategy.
\newblock
  \url{www.cppinvestments.com/the-fund/how-we-invest/our-investment-strategy}.

\bibitem[{Carver(2009)}]{Carver2009}
Carver, A.~B. (2009).
\newblock Do leveraged and inverse {ETFs} converge to zero?
\newblock \emph{ETFs and Indexing} 1, 144--149.

\bibitem[{Cavaglia et~al.(2022)Cavaglia, Scott, Blay, and Hixon}]{Scott_2022}
Cavaglia, S., L.~Scott, K.~Blay, and S.~Hixon (2022).
\newblock Multi-asset class factor premia: A strategic asset allocation
  perspective.
\newblock \emph{The Journal of Portfolio Management} 48:9, 14--32.

\bibitem[{Charupat and Miu(2011)}]{CharupatMiu2011}
Charupat, N. and P.~Miu (2011).
\newblock The pricing and performance of leveraged exchange-traded funds.
\newblock \emph{Journal of Banking and Finance} 35(4), 966--977.

\bibitem[{Cheng and Madhavan(2009)}]{ChenMadhavan2009}
Cheng, M. and A.~Madhavan (2009).
\newblock The dynamics of leveraged and inverse exchange-traded funds.
\newblock \emph{Journal of Investment Management} 7, 43--62.

\bibitem[{Cogneau and Zakalmouline(2013)}]{CogneauZakalmouline2013}
Cogneau, P. and V.~Zakalmouline (2013).
\newblock Block bootstrap methods and the choice of stocks for the long run.
\newblock \emph{Quantitative Finance} 13, 1443--1457.

\bibitem[{Dai et~al.(2023)Dai, Kou, Soner, and Yang}]{DaiEtAl2023}
Dai, M., S.~Kou, H.~M. Soner, and C.~Yang (2023).
\newblock Leveraged exchange-traded funds with market closure and frictions.
\newblock \emph{Management Science} 69(4), 2517--2535.

\bibitem[{Dang and Forsyth(2016)}]{DangForsyth2016}
Dang, D. and P.~Forsyth (2016).
\newblock Better than pre-commitment mean-variance portfolio allocation
  strategies: A semi-self-financing {H}amilton--{J}acobi--{B}ellman equation
  approach.
\newblock \emph{European Journal of Operational Research} 250, 827--841.

\bibitem[{{DeVault} et~al.(2021){DeVault}, Turtle, and Wang}]{DeVaultEtAl2021}
{DeVault}, L., H.~Turtle, and K.~Wang (2021).
\newblock Blessing or curse? institutional investment in leveraged {ETFs}.
\newblock \emph{Journal of Banking and Finance} 129(106169).

\bibitem[{Dichtl et~al.(2016)Dichtl, Drobetz, and Wambach}]{dichtl2016}
Dichtl, H., W.~Drobetz, and M.~Wambach (2016).
\newblock Testing rebalancing strategies for stock-bond portfolos across
  different asset allocations.
\newblock \emph{Applied Economics} 48, 772--788.

\bibitem[{Forsyth(2020)}]{Forsyth2019CVaR}
Forsyth, P. (2020).
\newblock Multiperiod mean conditional value at risk asset allocation: Is it
  advantageous to be time consistent?
\newblock \emph{SIAM Journal on Financial Mathematics} 11(2), 358--384.

\bibitem[{Forsyth and Vetzal(2017)}]{ForsythVetzal2016}
Forsyth, P. and K.~Vetzal (2017).
\newblock Dynamic mean variance asset allocation: Tests for robustness.
\newblock \emph{International Journal of Financial Engineering} 4:2.
\newblock 1750021 (electronic).

\bibitem[{Forsyth et~al.(2019)Forsyth, Vetzal, and
  Westmacott}]{ForsythVetzalWestmacott2019}
Forsyth, P., K.~Vetzal, and G.~Westmacott (2019).
\newblock Management of portfolio depletion risk through optimal life cycle
  asset allocation.
\newblock \emph{North American Actuarial Journal} 23(3), 447--468.

\bibitem[{Giese(2010)}]{Giese2010}
Giese, G. (2010).
\newblock On the risk-return profile of leveraged and inverse {ETFs}.
\newblock \emph{Journal of Asset Management} 11, 219--228.

\bibitem[{Goetzmann et~al.(2002)Goetzmann, Ingersoll, Spiegel, and
  Welch}]{GoetzmannEtAl2002}
Goetzmann, W., J.~Ingersoll, M.~Spiegel, and I.~Welch (2002).
\newblock Sharpening {Sharpe} ratios.
\newblock \emph{NBER Working Paper 9116} .

\bibitem[{Goetzmann et~al.(2007)Goetzmann, Ingersoll, Spiegel, and
  Welch}]{GoetzmannEtAl2007}
Goetzmann, W., J.~Ingersoll, M.~Spiegel, and I.~Welch (2007).
\newblock Manipulation and manipulation-proof performance measures.
\newblock \emph{Review of Financial Economics} 20, 1503--1546.

\bibitem[{{Government Pension Fund Global}(2022)}]{Norway_plan}
{Government Pension Fund Global} (2022).
\newblock Investment strategy.
\newblock \url{www.nbim.no/en/the-fund/how-we-invest/investment-strategy/}.

\bibitem[{Guasoni and Mayerhofer(2023)}]{GuasoniMayerhofer2023}
Guasoni, P. and E.~Mayerhofer (2023).
\newblock Leveraged funds: Robust replication and performance evaluation.
\newblock \emph{Quantitative Finance} 23(7-8), 1155--1176.

\bibitem[{Han and Weinan(2016)}]{HanWeinan2016}
Han, J. and E.~Weinan (2016).
\newblock Deep learning approximation for stochastic control problems.
\newblock \emph{NIPS Deep Reinforcement Learning Workshop} .

\bibitem[{Hassine and Roncalli(2013)}]{HassineRoncalli2013}
Hassine, M. and T.~Roncalli (2013).
\newblock Measuring performance of exchange-traded funds.
\newblock \emph{The Journal of Index Investing} 4(3), 57--85.

\bibitem[{Hill and Foster(2009)}]{HillFoster2009}
Hill, J. and G.~Foster (2009).
\newblock Understanding returns of leveraged and inverse funds.
\newblock \emph{Journal of Indexes} September/October, 40--58.

\bibitem[{Hu and Lauri{\`e}re(2023)}]{HuLauriere2023}
Hu, R. and M.~Lauri{\`e}re (2023).
\newblock Recent developments in machine learning methods for stochastic
  control and games.
\newblock \emph{Working paper} .

\bibitem[{Israelsen and Cogswell(2007)}]{IsraelsenCogswell2006}
Israelsen, C.~K. and G.~F. Cogswell (2007).
\newblock The error of tracking error.
\newblock \emph{Journal of Asset Management} 7(6), 419--424.

\bibitem[{Jarrow(2010)}]{Jarrow2010}
Jarrow, R.~A. (2010).
\newblock Understanding the risk of leveraged {ETFs}.
\newblock \emph{Finance Research Letters} 7(3), 135--139.

\bibitem[{Kashyap et~al.(2021)Kashyap, Kovrijnykh, Li, and
  Pavlova}]{KashyapEtAl2021}
Kashyap, A.~K., N.~Kovrijnykh, J.~Li, and A.~Pavlova (2021).
\newblock The benchmark inclusion subsidy.
\newblock \emph{Journal of Financial Economics} Forthcoming.

\bibitem[{Knapp(2023)}]{Knapp2023}
Knapp, C. (2023).
\newblock Leveraged {ETFs}: The data says we're missing out.
\newblock \emph{LinkedIn} https://linkedin.com/
  pulse/leveraged-etfs-data-says-were-missing-out-christian-knapp.

\bibitem[{Korn and Lindberg(2014)}]{KornLindberg2014}
Korn, R. and C.~Lindberg (2014).
\newblock Portfolio optimization for an investor with a benchmark.
\newblock \emph{Decisions in Economics and Finance} 37, 373--384.

\bibitem[{Kou(2002)}]{KouOriginal}
Kou, S.~G. (2002).
\newblock A jump-diffusion model for option pricing.
\newblock \emph{Management Science} 48(8), 1086--1101.

\bibitem[{Lehalle and Simon(2021)}]{LehalleSimon2021}
Lehalle, C.-A. and G.~Simon (2021).
\newblock Portfolio selection with active strategies: how long only constraints
  shape convictions.
\newblock \emph{Journal of Asset Management} .
\newline\urlprefix\url{https://doi.org/10.1057/s41260-021-00219-z}

\bibitem[{Leung et~al.(2017)Leung, Lorig, and Pascucci}]{LeungEtAl2016}
Leung, T., M.~Lorig, and A.~Pascucci (2017).
\newblock Leveraged {ETF} implied volatilities from {ETF} dynamics.
\newblock \emph{Mathematical Finance} 27(4), 1035--1068.

\bibitem[{Leung and Santoli(2012)}]{LeungSantoli2012}
Leung, T. and M.~Santoli (2012).
\newblock Leveraged etfs: admissible leverage and risk horizon.
\newblock \emph{Journal of Investment Strategies} 2(1), 39--61.

\bibitem[{Leung and Sircar(2015)}]{LeungSircar2015}
Leung, T. and R.~Sircar (2015).
\newblock Implied volatility of leveraged {ETF} options.
\newblock \emph{Applied Mathematical Finance} 22(2), 162--188.

\bibitem[{Li and Ng(2000)}]{LiNg2000}
Li, D. and W.-L. Ng (2000).
\newblock Optimal dynamic portfolio selection: multi period mean variance
  formulation.
\newblock \emph{Mathematical Finance} 10, 387--406.

\bibitem[{Li and Forsyth(2019)}]{LiForsyth2019}
Li, Y. and P.~Forsyth (2019).
\newblock A data-driven neural network approach to optimal asset allocation for
  target based defined contribution pension plans.
\newblock \emph{Insurance: Mathematics and Economics} 86, 189--204.

\bibitem[{Mackintosh(2008)}]{Mackintosh2008}
Mackintosh, P. (2008).
\newblock Double trouble.
\newblock \emph{The Journal of Beta Investment Strategies} 1, 25--31.

\bibitem[{M\"{a}kinen and Toivanen(2024)}]{Jari_2024}
M\"{a}kinen, R. A.~E. and J.~Toivanen (2024).
\newblock Short communication: Monte carlo expected wealth and risk measure
  trade-off portfolio optimization.
\newblock \emph{SIAM Journal on Financial Mathematics} 15(2), SC41--SC53.

\bibitem[{Menoncin and Vigna(2017)}]{MenoncinVigna2013}
Menoncin, F. and E.~Vigna (2017).
\newblock Mean-variance target-based optimisation for defined contribution
  schemes in a stochastic framework.
\newblock \emph{Insurance: Mathematics and Economics} 76, 172--184.

\bibitem[{Merton(1976)}]{MertonJumps1976}
Merton, R. (1976).
\newblock Option pricing when underlying stock returns are discontinuous.
\newblock \emph{Journal of Financial Economics} 3, 125--144.

\bibitem[{Ni et~al.(2022)Ni, Li, Forsyth, and Carroll}]{NiLiForsyth2020}
Ni, C., Y.~Li, P.~Forsyth, and R.~Carroll (2022).
\newblock Optimal asset allocation for outperforming a stochastic benchmark
  target.
\newblock \emph{Quantitative Finance} 22:9, 1595--1626.

\bibitem[{Ni et~al.(2024)Ni, Li, and Forsyth}]{NiLiForsyth2023_LFNN}
Ni, C., Y.~Li, and P.~A. Forsyth (2024).
\newblock Neural network approach to portfolio optimization with leverage
  constraints: a case study on high inflation investment.
\newblock \emph{Quantitative Finance} 24(6), 753--777.

\bibitem[{Oksendal and Sulem(2019)}]{OksendalSulemBook}
Oksendal, B. and A.~Sulem (2019).
\newblock \emph{Applied Stochastic Control of Jump Diffusions}.
\newblock Springer, 3rd edition.

\bibitem[{Politis and Romano(1994)}]{politis1994}
Politis, D. and J.~Romano (1994).
\newblock The stationary bootstrap.
\newblock \emph{Journal of the American Statistical Association} 89,
  1303--1313.

\bibitem[{Politis et~al.(1999)Politis, Romano, and Wolf}]{Politis_1999_b}
Politis, D., J.~Romano, and M.~Wolf (1999).
\newblock \emph{Subsampling}.
\newblock Springer-Verlag.
\newblock Chapter 4: Subsampling for nonstationary time series.

\bibitem[{Politis(2003)}]{Politis_2003_a}
Politis, D.~N. (2003).
\newblock The impact of bootstrap methods on time series analysis.
\newblock \emph{Statistical Science} 18:2, 219--230.

\bibitem[{Reppen and Soner(2023)}]{ReppenSoner2023}
Reppen, A.~M. and H.~M. Soner (2023).
\newblock Deep empirical risk minimization in finance: looking into the future.
\newblock \emph{Mathematical Finance} 33(1), 116--145.

\bibitem[{Reppen et~al.(2023)Reppen, Soner, and
  Tissot-Daguette}]{ReppenSonerEtAl2023}
Reppen, A.~M., H.~M. Soner, and V.~Tissot-Daguette (2023).
\newblock Deep stochastic optimization in finance.
\newblock \emph{Digital Finance} 5, 91--111.

\bibitem[{Scott and Cavaglia(2017)}]{Scott_2017}
Scott, L. and S.~Cavaglia (2017).
\newblock A wealth management perspective on factor premia and the value of
  downside protection.
\newblock \emph{The Journal of Portfolio Management} 43:3, 1--9.

\bibitem[{Simonian and Martirosyan(2022)}]{Simonian_2022}
Simonian, J. and A.~Martirosyan (2022).
\newblock Sharpe parity redux.
\newblock \emph{The Journal of Portfolio Management} 48:9, 183--193.

\bibitem[{Strub et~al.(2019)Strub, Li, and Cui}]{Strub_2019}
Strub, M., D.~Li, and X.~Cui (2019).
\newblock An enhanced mean-variance framework for robo-advising applications
  SSRN 302111.

\bibitem[{Sullivan(2009)}]{Sullivan2009}
Sullivan, R.~N. (2009).
\newblock The pitfalls of leveraged and inverse {ETFs}.
\newblock \emph{CFA Magazine} 20(3).

\bibitem[{Trainor et~al.(2018)Trainor, Chhachi, and Brown}]{TrainorEtAl2018}
Trainor, W., I.~Chhachi, and C.~Brown (2018).
\newblock A portfolio of leveraged exchange traded funds.
\newblock \emph{Working paper} .

\bibitem[{{Van Staden} et~al.(2018){Van Staden}, Dang, and
  Forsyth}]{PvSDangForsyth2018_TCMV}
{Van Staden}, P.~M., D.~Dang, and P.~Forsyth (2018).
\newblock Time-consistent mean-variance portfolio optimization: a numerical
  impulse control approach.
\newblock \emph{Insurance: Mathematics and Economics} 83(C), 9--28.

\bibitem[{{Van Staden} et~al.(2021){Van Staden}, Dang, and
  Forsyth}]{PvSDangForsyth2019_Distributions}
{Van Staden}, P.~M., D.~Dang, and P.~Forsyth (2021).
\newblock On the distribution of terminal wealth under dynamic mean-variance
  optimal investment strategies.
\newblock \emph{SIAM Journal on Financial Mathematics} 12(2), 566--603.

\bibitem[{{Van Staden} et~al.(2023){Van Staden}, Forsyth, and
  Li}]{PvSForsythLi2022_stochbm}
{Van Staden}, P.~M., P.~A. Forsyth, and Y.~Li (2023).
\newblock Beating a benchmark: dynamic programming may not be the right
  numerical approach.
\newblock \emph{SIAM Journal on Financial Mathematics} 14(2).

\bibitem[{{Van Staden} et~al.(2024){Van Staden}, Forsyth, and
  Li}]{PvsForsythLi2023_NN}
{Van Staden}, P.~M., P.~A. Forsyth, and Y.~Li (2024).
\newblock A global-in-time neural network approach to dynamic portfolio
  optimization.
\newblock \emph{Applied Mathematical Finance} (online), 1--33.
\newline\urlprefix\url{https://doi.org/10.1080/1350486X.2024.2410200}

\bibitem[{Vigna(2014)}]{Vigna_efficiency2014}
Vigna, E. (2014).
\newblock On efficiency of mean-variance based portfolio selection in defined
  contribution pension schemes.
\newblock \emph{Quantitative Finance} 14(2), 237--258.

\bibitem[{Vigna(2020)}]{Vigna2016TC}
Vigna, E. (2020).
\newblock On time consistency for mean-variance portfolio selection.
\newblock \emph{International Journal of Theoretical and Applied Finance}
  23(6).

\bibitem[{Vigna(2022)}]{Vigna2017TailOptimality}
Vigna, E. (2022).
\newblock Tail optimality and preferences consistency for intertemporal
  optimization problems.
\newblock \emph{SIAM Journal on Financial Mathematics} 13(1).

\bibitem[{Wagalath(2014)}]{Wagalath2014}
Wagalath, L. (2014).
\newblock Modelling the rebalancing slippage of leveraged exchange-traded
  funds.
\newblock \emph{Quantitative Finance} 14(9), 1503--1511.

\bibitem[{Yoon et~al.(2019)Yoon, Jarrett, and {Van der
  Schaar}}]{YoonTimeGAN2019}
Yoon, J., D.~Jarrett, and M.~{Van der Schaar} (2019).
\newblock Time-series generative adversarial networks.
\newblock \emph{33rd Conference on Neural Information Processing Systems
  (NeurIPS 2019)}
  \url{https://proceedings.neurips.cc/paper/2019/file/c9efe5f26cd17ba6216bbe2a7d26d490-Paper.pdf}.

\bibitem[{Zhao(2007)}]{Zhao2007}
Zhao, Y. (2007).
\newblock A dynamic model of active portfolio management with benchmark
  orientation.
\newblock \emph{Journal of Banking and Finance} 31, 3336--3356.

\bibitem[{Zhou and Li(2000)}]{ZhouLi2000}
Zhou, X. and D.~Li (2000).
\newblock Continuous time mean variance portfolio selection: a stochastic {LQ}
  framework.
\newblock \emph{Applied Mathematics and Optimization} 42, 19--33.

\end{thebibliography}

\noindent \begin{appendices}

\normalsize

\section{Proofs of main results\label{sec:Appendix Proofs of key results}}

The proofs of the key results of Section \ref{sec:Closed-form-solutions}
are presented in this appendix.

\subsection{Proof of Theorem \ref{thm: Verification theorem}\label{subsec: Appendix proof of Verification thm}}

For a given deterministic benchmark strategy $\hat{\varrho}_{s}\left(t\right)$,
consider an arbitrary admissible investor strategy $\varrho_{\ell}\left(t\right)\coloneqq\varrho_{\ell}\left(t,\boldsymbol{X}_{\ell}\left(t\right)\right)\in\mathcal{A}_{0}$,
where we omit the dependence of $\varrho_{\ell}$ on $\boldsymbol{X}_{\ell}\left(t\right)=\left(W_{\ell}\left(t\right),\hat{W}\left(t\right),\hat{\varrho}_{s}\left(t\right)\right)$
for notational simplicity. Considering the objective functional of
the IR problem (\ref{eq: IR problem for analytical}) at a given point
$\left(t,w,\hat{w}\right)\in\left[t_{0},T\right]\times\mathbb{R}^{2}$
for a given and fixed value of $\gamma>0$, define 
\begin{eqnarray}
J\left(t,w,\hat{w};\varrho_{\ell}\right) & = & E_{\varrho_{\ell}}^{t,w,\hat{w}}\left[\left(W_{\ell}\left(T\right)-\left[\hat{W}\left(T\right)+\gamma\right]\right)^{2}\right].\label{eq: Appendix Objective functional}
\end{eqnarray}

If we proceed informally and assume that $J$ sufficiently smooth,
then the application of Itô's lemma for jump processes (\cite{OksendalSulemBook})
gives 
\begin{eqnarray}
 &  & E_{\varrho_{\ell}}^{t,w,\hat{w}}\left[\int_{t}^{t+h}dJ\left(u,W_{\ell}\left(u\right),\hat{W}\left(u\right);\varrho_{\ell}\right)\right]\nonumber \\
 & = & E_{\varrho_{\ell}}^{t,w,\hat{w}}\left[\int_{t}^{t+h}\left(\frac{\partial J}{\partial t}+\frac{\partial J}{\partial w}\cdot\left\{ W_{\ell}\left(u\right)\cdot\left[r+\varrho_{\ell}\left(u\right)\left\{ \beta\left(\mu-\lambda\kappa_{1}^{s}-r\right)-c_{\ell}\right\} \right]+q\right\} \right)\cdot du\right]\nonumber \\
 &  & +E_{\varrho_{\ell}}^{t,w,\hat{w}}\left[\int_{t}^{t+h}\frac{\partial J}{\partial\hat{w}}\cdot\left\{ \hat{W}\left(u\right)\cdot\left[r+\left(\mu-\lambda\kappa_{1}^{s}-r\right)\hat{\varrho}_{s}\left(u\right)\right]+q\right\} \cdot du\right]\nonumber \\
 &  & +E_{\varrho_{\ell}}^{t,w,\hat{w}}\left[\int_{t}^{t+h}\frac{1}{2}\left(\frac{\partial^{2}J}{\partial\hat{w}^{2}}\cdot\left[\hat{\varrho}_{s}\left(u\right)\hat{W}\left(u\right)\sigma\right]^{2}+\frac{\partial^{2}J}{\partial w^{2}}\cdot\left[\varrho_{\ell}\left(u\right)W_{\ell}\left(u\right)\beta\sigma\right]^{2}\right)\cdot du\right]\nonumber \\
 &  & +E_{\varrho_{\ell}}^{t,w,\hat{w}}\left[\int_{t}^{t+h}\frac{\partial^{2}J}{\partial w\partial\hat{w}}\cdot\left[\varrho_{\ell}\left(u\right)W_{\ell}\left(u\right)\beta\sigma\right]\left[\hat{\varrho}_{s}\left(u\right)\hat{W}\left(u\right)\sigma\right]\cdot du\right]\nonumber \\
 &  & +E_{\varrho_{\ell}}^{t,w,\hat{w}}\left[\lambda\int_{t}^{t+h}\left[\int_{0}^{\infty}\phi\left(u,W_{\ell}\left(u^{-}\right),\hat{W}\left(u^{-}\right),\xi^{s};\varrho_{\ell}\right)G\left(\xi^{s}\right)d\xi^{s}-J\left(W_{\ell}\left(u^{-}\right),\hat{W}\left(u^{-}\right),u;\varrho_{\ell}\right)\right]du\right],\label{eq: Appendix dJ2}
\end{eqnarray}
where all partial derivatives are evaluated at $\left(u,W_{\ell}\left(u\right),\hat{W}\left(u\right);\varrho_{\ell}\right)$,
and
\begin{eqnarray}
 &  & \phi\left(u,W_{\ell}\left(u^{-}\right),\hat{W}\left(u^{-}\right),\xi^{s};\varrho_{\ell}\right)\nonumber \\
 & = & J\left(W_{\ell}\left(u^{-}\right)+\varrho_{\ell}\left(u\right)W_{\ell}\left(u^{-}\right)\beta\left(\xi^{\ell}-1\right),\hat{W}\left(u^{-}\right)+\hat{\varrho}_{s}\left(u\right)\hat{W}\left(u^{-}\right)\left(\xi^{s}-1\right),u;\varrho_{\ell}\right).  \\ \nonumber 
   \label{eq: Appendix phi for dJ2}
\end{eqnarray}
Recall that the LETF jump multiplier $\xi^{\ell}$ is a function (\ref{eq: LETF jump multiplier})
of the underlying index $S$ jump multiplier $\xi^{s}$, so $\phi$
in (\ref{eq: Appendix phi for dJ2}) can be interpreted as a function
of $\xi^{s}$ if all other values are held fixed. 

Continuing to proceed informally, dividing (\ref{eq: Appendix dJ2})
by $h>0$, taking limits as $h\downarrow0$ and assuming the limits
and expectations could be interchanged, an application of the dynamic
programming principle results in the PIDE (\ref{eq: IR pide})-(\ref{eq: IR pide terminal condition}).

While providing the necessary intuition, the preceding arguments are
merely informal. However, since similar arguments (see \cite{OksendalSulemBook,ApplebaumBook})
can be applied to a suitably smooth test function instead of the objective
functional in order to formally prove (\ref{eq: IR pide})-(\ref{eq: IR pide terminal condition}),
the details are therefore omitted.

\subsection{Proof of Proposition \ref{prop: Closed-form Optimal control for LETF}\label{subsec: Appendix proof of Optimal control for LETF}}

The quadratic terminal condition (\ref{eq: IR pide terminal condition})
suggests an ansatz for the value function $V$ of the form 
\begin{eqnarray}
V\left(t,w,\hat{w},\hat{\varrho}_{s}\right) & = & A\left(t\right)w^{2}+\hat{A}\left(t\right)\hat{w}^{2}+D\left(t\right)w\hat{w}+F\left(t\right)w+\hat{F}\left(t\right)\hat{w}+C\left(t\right),\label{eq: V ansatz}
\end{eqnarray}
where $A,\hat{A},D,F,\hat{F}$ and $C$ are deterministic but unknown
functions of time. Since (\ref{eq: V ansatz}) implies partial derivatives
of the form 
\begin{equation}
\frac{\partial V}{\partial w}=2A\left(t\right)w+F\left(t\right)+D\left(t\right)\hat{w},\qquad\frac{\partial^{2}V}{\partial w^{2}}=2A\left(t\right),\qquad\textrm{and }\qquad\frac{\partial^{2}V}{\partial w\partial\hat{w}}=D\left(t\right),\label{eq: partial derivatives in terms of unknowns}
\end{equation}
substituting (\ref{eq: V ansatz})-(\ref{eq: partial derivatives in terms of unknowns})
into the HJB PIDE (\ref{eq: IR pide}) results in the pointwise supremum
$\varrho_{\ell}^{\ast}=\varrho_{\ell}^{\ast}\left(t,w,\hat{w},\hat{\varrho}_{s}\right)$
obtained from the first-order condition that satisfies the relationship
\begin{eqnarray}
\varrho_{\ell}^{\ast}\cdot w & = & -\left(\frac{\beta\left[\mu+\lambda\left(\kappa_{1}^{\ell}-\kappa_{1}^{s}\right)-r\right]-c_{\ell}}{\beta^{2}\left(\sigma^{2}+\lambda\kappa_{2}^{\ell}\right)}\right)\left[w+\frac{F\left(t\right)}{2A\left(t\right)}+\frac{D\left(t\right)}{2A\left(t\right)}\cdot\hat{w}\right]\nonumber \\
 &  & -\left[\frac{\sigma^{2}+\lambda\kappa_{\chi}^{\ell,s}}{\beta\left(\sigma^{2}+\lambda\kappa_{2}^{\ell}\right)}\right]\frac{D\left(t\right)}{2A\left(t\right)}\cdot\hat{\varrho}_{s}\hat{w}.\label{eq: p* implied by ansatz}
\end{eqnarray}
Substituting (\ref{eq: p* implied by ansatz}) into (\ref{eq: Ham_IR})
to obtain $\mathcal{H}\left(\varrho_{\ell}^{\ast};t,w,\hat{w},\hat{\varrho}_{s}\right)$,
the PIDE (\ref{eq: IR pide})-(\ref{eq: IR pide terminal condition})
implies the following set of ordinary differential equations (ODEs)
for the unknown functions $A,D$ and $F$ on $t\in\left[t_{0},T\right]$,
\begin{eqnarray}
\frac{d}{dt}A\left(t\right) & = & -\left(2r-\frac{\left(\beta\left[\mu+\lambda\left(\kappa_{1}^{\ell}-\kappa_{1}^{s}\right)-r\right]-c_{\ell}\right)^{2}}{\beta^{2}\left(\sigma^{2}+\lambda\kappa_{2}^{\ell}\right)}\right)A\left(t\right),\qquad A\left(T\right)=1,\label{eq: ODE for A}\\
\frac{d}{dt}D\left(t\right) & = & -\left(2r-\frac{\left(\beta\left[\mu+\lambda\left(\kappa_{1}^{\ell}-\kappa_{1}^{s}\right)-r\right]-c_{\ell}\right)^{2}}{\beta^{2}\left(\sigma^{2}+\lambda\kappa_{2}^{\ell}\right)}+K_{\beta}^{\ell,s}\cdot\hat{\varrho}_{s}\left(t\right)\right)D\left(t\right),\qquad D\left(T\right)=-2,\label{eq: ODE for D}\\
\frac{d}{dt}F\left(t\right) & = & -2qA\left(t\right)-\left(r-\frac{\left(\beta\left[\mu+\lambda\left(\kappa_{1}^{\ell}-\kappa_{1}^{s}\right)-r\right]-c_{\ell}\right)^{2}}{\beta^{2}\left(\sigma^{2}+\lambda\kappa_{2}^{\ell}\right)}\right)F\left(t\right)-qD\left(t\right),\qquad F\left(T\right)=-2\gamma,\label{eq: ODE for F}
\end{eqnarray}
where the constant $K_{\beta}^{\ell,s}$ is given by (\ref{eq: Constant K beta for LETF})

Note that the derivation of (\ref{eq: ODE for D}) as an ODE requires
the benchmark strategy $\hat{\varrho}_{s}$ to be deterministic (in
the case of closed-form solutions) as per Assumption \ref{assu: Extra stylized-assumptions-for continuous rebalancing}.
Solving ODEs (\ref{eq: ODE for A})-(\ref{eq: ODE for F}), we obtain
\begin{eqnarray}
A\left(t\right) & = & \exp\left\{ \left(2r-\frac{\left(\beta\left[\mu+\lambda\left(\kappa_{1}^{\ell}-\kappa_{1}^{s}\right)-r\right]-c_{\ell}\right)^{2}}{\beta^{2}\left(\sigma^{2}+\lambda\kappa_{2}^{\ell}\right)}\right)\left(T-t\right)\right\} ,\label{eq: Function A}\\
D\left(t\right) & = & -2\cdot\exp\left\{ \left(2r-\frac{\left(\beta\left[\mu+\lambda\left(\kappa_{1}^{\ell}-\kappa_{1}^{s}\right)-r\right]-c_{\ell}\right)^{2}}{\beta^{2}\left(\sigma^{2}+\lambda\kappa_{2}^{\ell}\right)}\right)\left(T-t\right)+K_{\beta}^{\ell,s}\cdot\int_{t}^{T}\hat{\varrho}_{s}\left(u\right)du\right\} ,\label{eq: Function D}\\
F\left(t\right) & = & 2\exp\left\{ \left(r-\frac{\left(\beta\left[\mu+\lambda\left(\kappa_{1}^{\ell}-\kappa_{1}^{s}\right)-r\right]-c_{\ell}\right)^{2}}{\beta^{2}\left(\sigma^{2}+\lambda\kappa_{2}^{\ell}\right)}\right)\left(T-t\right)\right\} \times\nonumber \\
 &  & \left[-\gamma+\frac{q}{r}\left(e^{r\left(T-t\right)}-1\right)-q\cdot\int_{t}^{T}\exp\left\{ r\left(T-v\right)+K_{\beta}^{\ell,s}\cdot\int_{v}^{T}\hat{\varrho}_{s}\left(u\right)du\right\} dv\right].\label{eq: Function F}
\end{eqnarray}
Substituting (\ref{eq: Function A})-(\ref{eq: Function F}) into
(\ref{eq: p* implied by ansatz}) and simplifying, we obtain the optimal
fraction of wealth $\varrho_{\ell}^{\ast}$ to invest in the LETF
(\ref{eq: LETF optimal strategy}) as per Proposition \ref{prop: Closed-form Optimal control for LETF}.

\subsection{Expressions for $\kappa_{1}^{\ell}$, $\kappa_{2}^{\ell}$ and $\kappa_{\chi}^{\ell,s}$\label{subsec: Appendix expressions for kappas}}

For the purposes of illustrating the closed-form solutions of Section
\ref{sec:Closed-form-solutions}, the broad equity market index $S$
is assumed to have dynamics (\ref{eq: S dynamics}) with jumps as
modelled in \cite{KouOriginal}. As a result, with $p_{up}$ denoting
the probability of an upward jump given that a jump occurs, $y=\log\xi^{s}$
is assumed in \cite{KouOriginal} to follow an asymmetric double-exponential
distribution with PDF $g\left(y\right)$, 
\begin{align}
g\left(y\right)= & p_{up}\eta_{1}e^{-\eta_{1}y}\mathbb{I}_{\left\{ y\geq0\right\} }+\left(1-p_{up}\right)\eta_{2}e^{\eta_{2}y}\mathbb{I}_{\left\{ y<0\right\} },\label{eq: log jump pdf - KOU}
\end{align}
where $p_{up}\in\left[0,1\right]\textrm{ and }\eta_{1}>1,\eta_{2}>0$.
Equivalently, the PDF of $\xi^{s}$ is given by
\begin{eqnarray}
G\left(\xi^{s}\right) & = & p_{up}\eta_{1}\left(\xi^{s}\right)^{-\eta_{1}-1}\mathbb{I}_{\left[\xi^{s}\geq1\right]}\left(\xi^{s}\right)+\left(1-p_{up}\right)\eta_{2}\left(\xi^{s}\right)^{\eta_{2}-1}\mathbb{I}_{\left[0\leq\xi^{s}<1\right]}\left(\xi^{s}\right).\label{eq: Jump pdf - KOU}
\end{eqnarray}

Recall that we have defined $\kappa_{1}^{s}$ and $\kappa_{1}^{s}$
in (\ref{eq: Kappa and Kappa2 for S dynamics}), repeated here for
convenience,
\begin{equation}
\kappa_{1}^{s}=\mathbb{E}\left[\xi^{s}-1\right],\qquad\kappa_{2}^{s}=\mathbb{E}\left[\left(\xi^{s}-1\right)^{2}\right].\label{eq: Appendix Kappa and Kappa2 for S dynamics}
\end{equation}
From the results in \cite{KouOriginal}, we can obtain (\ref{eq: Appendix Kappa and Kappa2 for S dynamics})
for the distribution (\ref{eq: Jump pdf - KOU}) using the results
\begin{equation}
\mathbb{E}\left[\xi^{s}\right]=\frac{p_{up}\eta_{1}}{\eta_{1}-1}+\frac{\left(1-p_{up}\right)\eta_{2}}{\eta_{2}+1},\qquad\mathbb{E}\left[\left(\xi^{s}\right)^{2}\right]=\frac{p_{up}\eta_{1}}{\eta_{1}-2}+\frac{\left(1-p_{up}\right)\eta_{2}}{\eta_{2}+2}.\label{eq: Appendix expectation for S kappas}
\end{equation}
However, since the LETF experiences slightly different jumps as per
(\ref{eq: LETF jump multiplier}), which we repeat here for convenience,
\begin{eqnarray} \xi^{\ell} & = & \begin{dcases} \xi^{s} & \textrm{if }\xi^{s}>\left(\beta-1\right)/\beta,\\ \frac{\left(\beta-1\right)}{\beta} & \textrm{if }\xi^{s}\leq\left(\beta-1\right)/\beta, \end{dcases}\label{eq: Appendix LETF jump multiplier} \end{eqnarray}we
cannot use the results (\ref{eq: Appendix Kappa and Kappa2 for S dynamics})
directly. Instead, expressions for $\kappa_{1}^{\ell}$, $\kappa_{2}^{\ell}$
and $\kappa_{\chi}^{\ell,s}$ are obtained using the results the following
lemma.
\begin{lem}
\label{lem: Deriving LETF kappas for Kou model}($\kappa_{1}^{\ell}$,
$\kappa_{2}^{\ell}$ and $\kappa_{\chi}^{\ell,s}$ in the \cite{KouOriginal}
model) Suppose the jump multiplier $\xi^{s}$ in the $S$-dynamics
(\ref{eq: S dynamics}) has PDF $G\left(\xi^{s}\right)$ given by
(\ref{eq: Jump pdf - KOU}), and a LETF with returns multiplier $\beta>1$
and dynamics (\ref{eq: LETF dynamics}) has jump multiplier $\xi^{\ell}$,
which is defined in terms of $\xi^{s}$ as per (\ref{eq: Appendix LETF jump multiplier}).
Then the quantities 
\begin{equation}
\kappa_{1}^{\ell}=\mathbb{E}\left[\xi^{\ell}-1\right],\qquad\kappa_{2}^{\ell}=\mathbb{E}\left[\left(\xi^{\ell}-1\right)^{2}\right],\qquad\kappa_{\chi}^{\ell,s}=\mathbb{E}\left[\left(\xi^{\ell}-1\right)\left(\xi^{s}-1\right)\right],\label{eq: Appendix Kappas for LETF dynamics}
\end{equation}
required by the IR-optimal investment strategy in Proposition \ref{prop: Closed-form Optimal control for LETF}
can be obtained using the following expressions:
\begin{eqnarray}
\mathbb{E}\left[\xi^{\ell}\right] & = & \frac{p_{up}\eta_{1}}{\eta_{1}-1}+\left(1-p_{up}\right)\eta_{2}\cdot\left[\frac{\vartheta^{\eta_{2}+1}}{\eta_{2}}+\left(\frac{1-\vartheta^{\eta_{2}+1}}{\eta_{2}+1}\right)\right],\label{eq: Appendix E_ell}\\
\mathbb{E}\left[\left(\xi^{\ell}\right)^{2}\right] & = & \frac{p_{up}\eta_{1}}{\eta_{1}-2}+\left(1-p_{up}\right)\eta_{2}\cdot\left[\frac{\vartheta^{\eta_{2}+2}}{\eta_{2}}+\left(\frac{1-\vartheta^{\eta_{2}+2}}{\eta_{2}+2}\right)\right],\label{eq: Appendix E_ell_2}\\
\mathbb{E}\left[\xi^{\ell}\xi^{s}\right] & = & \frac{p_{up}\eta_{1}}{\eta_{1}-2}+\left(1-p_{up}\right)\eta_{2}\cdot\left[\frac{\vartheta^{\eta_{2}+2}}{\eta_{2}+1}+\left(\frac{1-\vartheta^{\eta_{2}+2}}{\eta_{2}+2}\right)\right],\label{eq: Appendix E_ell_E_s}
\end{eqnarray}
where $\vartheta=\left(\beta-1\right)/\beta$.
\end{lem}

\begin{proof}
Consider (\ref{eq: Appendix E_ell_E_s}). Since $\beta>1$ and $\vartheta=\left(\beta-1\right)/\beta$,
we have $0<\vartheta<1$. Therefore, using the definition of the LETF
jump multiplier (\ref{eq: Appendix LETF jump multiplier}) and the
PDF $G\left(\xi^{s}\right)$, we have 
\begin{eqnarray}
\mathbb{E}\left[\xi^{\ell}\xi^{s}\right] & = & \int_{0}^{\infty}\xi^{\ell}\xi^{s}\cdot G\left(\xi^{s}\right)d\xi^{s}\nonumber \\
 & = & \vartheta\cdot\int_{0}^{\vartheta}\xi^{s}\cdot G\left(\xi^{s}\right)d\xi^{s}+\int_{\vartheta}^{1}\left(\xi^{s}\right)^{2}\cdot G\left(\xi^{s}\right)d\xi^{s}+\int_{1}^{\infty}\left(\xi^{s}\right)^{2}\cdot G\left(\xi^{s}\right)d\xi^{s}.\label{eq: Appendix deriving E_ell_E_s step 1}
\end{eqnarray}
Standard results (see (\ref{eq: Appendix expectation for S kappas})
and \cite{KouOriginal}) for the Kou model gives 
\begin{eqnarray}
\int_{1}^{\infty}\left(\xi^{s}\right)^{2}\cdot G\left(\xi^{s}\right)d\xi^{s} & = & \frac{p_{up}\eta_{1}}{\eta_{1}-2}.\label{eq: Appendix deriving E_ell_E_s step 2}
\end{eqnarray}
Using (\ref{eq: Appendix deriving E_ell_E_s step 2}) and writing
the first two terms of (\ref{eq: Appendix deriving E_ell_E_s step 1})
in terms of the log jump multiplier $y=\log\xi^{s}$ with PDF $g\left(y\right)$
as per (\ref{eq: log jump pdf - KOU}), we have 
\begin{eqnarray}
\mathbb{E}\left[\xi^{\ell}\xi^{s}\right] & = & \vartheta\cdot\int_{-\infty}^{\log\vartheta}e^{y}g\left(y\right)dy+\int_{\log\vartheta}^{0}e^{2y}g\left(y\right)dy+\frac{p_{up}\eta_{1}}{\eta_{1}-2}\nonumber \\
 & = & \left(1-p_{up}\right)\eta_{2}\vartheta\cdot\int_{-\infty}^{\log\vartheta}e^{\left(\eta_{2}+1\right)y}dy+\left(1-p_{up}\right)\eta_{2}\int_{\log\vartheta}^{0}e^{\left(\eta_{2}+2\right)y}dy+\frac{p_{up}\eta_{1}}{\eta_{1}-2}.\label{eq: Appendix deriving E_ell_E_s step 3}
\end{eqnarray}
Simplifying (\ref{eq: Appendix deriving E_ell_E_s step 3}) gives
(\ref{eq: Appendix E_ell_E_s}). Since (\ref{eq: Appendix E_ell})
and (\ref{eq: Appendix E_ell_2}) can be obtained using similar arguments,
the details are omitted.
\end{proof}

\subsection{Proof of Corollary \ref{cor: Closed-form Optimal control for VETF}\label{subsec: Appendix proof of Optimal control for VETF}}

For purposes of intuition, we first give informal arguments as to
how the results of Corollary \ref{cor: Closed-form Optimal control for VETF}
relate to the results of Proposition \ref{prop: Closed-form Optimal control for LETF}.
Recall that the VETF has returns multiplier $\beta=1$ (i.e. the VETF
simply aims to replicate the returns of $S$ before costs) and expense
ratio $c_{v}>0$. Note that if we let $\beta\downarrow1$ in (\ref{eq: LETF jump multiplier}),
we have 
\begin{eqnarray}
\lim_{\beta\downarrow1}\xi^{\ell} & = & \xi^{s}\qquad\textrm{a.s.},\label{eq: Appendix VETF opt proof multipliers converge}
\end{eqnarray}
from which it follows that 
\begin{equation}
\lim_{\beta\downarrow1}\kappa_{1}^{\ell}=\kappa_{1}^{s},\quad\lim_{\beta\downarrow1}\kappa_{2}^{\ell}=\kappa_{2}^{s},\quad\lim_{\beta\downarrow1}\kappa_{\chi}^{\ell,s}=\kappa_{2}^{s}.\label{eq: Appendix VETF opt proof kappas converge}
\end{equation}

Therefore, comparing the VETF and LETF investor wealth dynamics (\ref{eq: SDE W VETF})-(\ref{eq: SDE W LETF})
in the case of identical expense ratios (i.e. $c_{\ell}=c_{v}$),
identical but not necessarily optimal investment strategies ($\varrho_{\ell}=\varrho_{v}$)
and the identical initial wealth, we have 
\begin{eqnarray}
\lim_{\beta\downarrow1}W_{\ell}\left(t\right) & = & W_{v}\left(t\right)\quad\textrm{a.s. }\forall t\in\left[t_{0},T\right],\quad\textrm{if }W_{\ell}\left(t_{0}\right)=W_{v}\left(t_{0}\right),c_{\ell}=c_{v},\quad\textrm{and }\varrho_{\ell}=\varrho_{v}.\label{eq: Appendix VETF opt proof Wealth converges}
\end{eqnarray}

In other words, if we let $\beta\downarrow1$ in the LETF investor
wealth dynamics (\ref{eq: SDE W LETF}), we recover the VETF investor
wealth dynamics (\ref{eq: SDE W VETF}). Continuing to proceed informally,
the results of Corollary \ref{cor: Closed-form Optimal control for VETF}
can therefore be obtained by letting $\beta\downarrow1$ in the results
of Proposition \ref{prop: Closed-form Optimal control for LETF},
provided we use the VETF expense ratio $c_{v}$ in both (\ref{eq: SDE W LETF})
and (\ref{eq: SDE W VETF}). Note that the definition (\ref{eq: Constant K beta for LETF})
of $K_{\beta}^{\ell,s}$, results (\ref{eq: Appendix VETF opt proof kappas converge})
and setting $c_{\ell}=c_{v}$ imply that 
\begin{eqnarray}
\lim_{\beta\downarrow1}K_{\beta}^{\ell,s} & = & \lim_{\beta\downarrow1}\left[\mu-r-\frac{\left(\beta\left[\mu+\lambda\left(\kappa_{1}^{\ell}-\kappa_{1}^{s}\right)-r\right]-c_{v}\right)\left(\sigma^{2}+\lambda\kappa_{\chi}^{\ell,s}\right)}{\beta\left(\sigma^{2}+\lambda\kappa_{2}^{\ell}\right)}\right]\nonumber \\
 & = & c_{v},\label{eq: Appendix VETF opt proof constant K converges to c_v}
\end{eqnarray}
which confirms that the functions $g_{v}$ and $h_{v}$ ((\ref{eq: function g for VETF})-(\ref{eq: function h for VETF}))
can be obtained from the functions the functions $g_{\ell}$ and $h_{\ell}$
((\ref{eq: function g for LETF})-(\ref{eq: function h for LETF}))
if identical expense ratios are used.

The preceding discussions were merely informal. More formally, the
proof of Corollary \ref{cor: Closed-form Optimal control for VETF}
proceeds along the same lines as the proof of Proposition \ref{prop: Closed-form Optimal control for LETF},
except that VETF investor wealth dynamics (\ref{eq: SDE W VETF})
is used instead of (\ref{eq: SDE W LETF}), and details are therefore
omitted.

\subsection{Proof of Proposition \ref{prop: Special case ZERO costs}\label{subsec: Appendix proof of optimal strategies ZERO COSTS}}

Suppose we have zero expense ratios, i.e. $c_{v}=c_{\ell}=0$, and
there are no jumps in the underlying $S$-dynamics (i.e. $\lambda=0$
in (\ref{eq: S dynamics})). Substituting these values in the deterministic
functions $g_{\ell}$ and $h_{\ell}$ ((\ref{eq: function g for LETF})
and (\ref{eq: function h for LETF})) in the case of a LETF and the
deterministic functions $g_{v}$ and $h_{v}$ ((\ref{eq: function g for VETF})
and (\ref{eq: function h for VETF})) in the case of the VETF, we
have
\begin{equation}
g_{\ell}\left(t\right)=g_{v}\left(t\right)=1,\label{eq: Appendix g1 g_beta for zero costs}
\end{equation}
and 
\begin{equation}
h_{\ell}\left(t\right)=h_{v}\left(t\right)=0.\label{eq: Appendix h1 h_beta for zero costs}
\end{equation}

In the case of the LETF investor, the optimal control (\ref{eq: LETF optimal strategy})
now satisfies
\begin{eqnarray}
\beta\cdot W_{\ell}^{\ast}\left(t\right)\cdot\varrho_{\ell}^{\ast}\left(t,\boldsymbol{X}_{\ell}^{\ast}\left(t\right)\right) & = & \left(\frac{\mu-r}{\sigma^{2}}\right)\cdot\left[\gamma e^{-r\left(T-t\right)}-\left(W_{\ell}^{\ast}\left(t\right)-\hat{W}\left(t\right)\right)\right]+\hat{\varrho}_{s}\left(t\right)\hat{W}\left(t\right),\label{eq: Appendix LETF optimal control zero costs}
\end{eqnarray}
whereas in the case of the VETF investor, the optimal control (\ref{eq: LETF optimal strategy})
now becomes 
\begin{eqnarray}
W_{v}^{\ast}\left(t\right)\cdot\varrho_{v}^{\ast}\left(t,\boldsymbol{X}_{v}^{\ast}\left(t\right)\right) & = & \left(\frac{\mu-r}{\sigma^{2}}\right)\cdot\left[\gamma e^{-r\left(T-t\right)}-\left(W_{v}^{\ast}\left(t\right)-\hat{W}\left(t\right)\right)\right]+\hat{\varrho}_{s}\left(t\right)\hat{W}\left(t\right).\label{eq: Appendix SETF optimal control zero costs}
\end{eqnarray}

Using (\ref{eq: Appendix LETF optimal control zero costs}) and (\ref{eq: Appendix SETF optimal control zero costs}),
define the auxiliary process $Q\left(t\right)$ as
\begin{eqnarray}
Q\left(t\right) & = & e^{-rt}\cdot\left[W_{\ell}^{\ast}\left(t\right)-W_{v}^{\ast}\left(t\right)\right],\qquad t\in\left[t_{0}=0,T\right],\label{eq: Appendix Q as discounted wealth difference}
\end{eqnarray}
with $Q\left(t_{0}\right)=e^{-rt_{0}}\left[W_{\ell}^{\ast}\left(t_{0}\right)-W_{v}^{\ast}\left(t_{0}\right)\right]=w_{0}-w_{0}=0.$

Substituting the optimal controls in this special case ((\ref{eq: Appendix LETF optimal control zero costs})
and (\ref{eq: Appendix SETF optimal control zero costs})) into the
wealth dynamics (\ref{eq: SDE W VETF})-(\ref{eq: SDE W LETF}) and
recalling that there are no jumps, we obtain the dynamics
\begin{align}
\frac{dQ\left(t\right)}{Q\left(t\right)}= & \left(\frac{\mu-r}{\sigma}\right)^{2}\cdot dt-\left(\frac{\mu-r}{\sigma}\right)\cdot dZ\left(t\right).\label{eq: Appendix Q dynamics}
\end{align}
Since $Q\left(t_{0}\right)=0$, dynamics (\ref{eq: Appendix Q dynamics})
therefore imply that $Q\left(t\right)=0,\forall t\geq t_{0}$, so
that in the special case of zero costs, we have 
\begin{eqnarray}
W_{\ell}^{\ast}\left(t\right) & = & W_{v}^{\ast}\left(t\right),\qquad\forall t\in\left[t_{0},T\right],\label{eq:Appendix wealth equal for zero costs}
\end{eqnarray}
which confirms (\ref{eq: prop zero cost - W equal}).

Using (\ref{eq:Appendix wealth equal for zero costs}) to write $W^{\ast}\left(t\right)\coloneqq W_{\ell}^{\ast}\left(t\right)=W_{v}^{\ast}\left(t\right)$
in this special case, the difference in controls (\ref{eq: Appendix LETF optimal control zero costs})
and (\ref{eq: Appendix SETF optimal control zero costs}) satisfy
\begin{align}
\left[\beta\cdot\varrho_{\ell}^{\ast}\left(t,\boldsymbol{X}_{\ell}^{\ast}\left(t\right)\right)-\varrho_{v}^{\ast}\left(t,\boldsymbol{X}_{v}^{\ast}\left(t\right)\right)\right]\cdot W^{\ast}\left(t\right)= & \beta\cdot W_{\ell}^{\ast}\left(t\right)\cdot\varrho_{\ell}^{\ast}\left(t,\boldsymbol{X}_{\ell}^{\ast}\left(t\right)\right)-W_{v}^{\ast}\left(t\right)\cdot\varrho_{v}^{\ast}\left(t,\boldsymbol{X}_{v}^{\ast}\left(t\right)\right)\nonumber \\
= & -\left(\frac{\mu-r}{\sigma^{2}}\right)\left[W_{\ell}^{\ast}\left(t\right)-W_{v}^{\ast}\left(t\right)\right]\nonumber \\
= & 0,\qquad\forall t\geq t_{0},\label{eq: Appendix control difference zero costs}
\end{align}
thereby verifying (\ref{eq: prop zero cost - LETF control a beta multiple}).

Finally, given the form of the optimal control (\ref{eq: Appendix SETF optimal control zero costs})
for the VETF in this special case, together with Assumption \ref{assu: Stylized-assumptions-for closed-form}
and wealth dynamics (\ref{eq: SDE W VETF}), imply that we can obtain
the optimal Information Ratio in the case of the VETF as given by
(\ref{eq: Prop zero cost - IR optimal equal}) (see \cite{PvSForsythLi2022_stochbm}).
However, since (\ref{eq:Appendix wealth equal for zero costs}), this
is also the optimal IR using the LET in this special case, thereby
confirming (\ref{eq: Prop zero cost - IR optimal equal}) and completing
the proof of Proposition \ref{prop: Special case ZERO costs}.

\section{Source data and parameters\label{sec:Appendix - Source-data}}

In this appendix, we provide details regarding the source data used
to obtain the indicative investment results presented in Section \ref{sec:Closed-form-solutions}
and Section \ref{sec:Indicative-investment-results}.

Returns data for US Treasury bills and bonds, as well as the broad
equity market index, were obtained from the CRSP\footnote{Calculations were based on data from the Historical Indexes 2024©,
Center for Research in Security Prices (CRSP), The University of Chicago
Booth School of Business. Wharton Research Data Services was used
in preparing this article. This service and the data available thereon
constitute valuable intellectual property and trade secrets of WRDS
and/or its third party suppliers.}. In more detail, the historical time series are as follows:
\begin{enumerate}
\item T30 (30-day Treasury bill): CRSP, monthly returns for 30-day Treasury
bill.
\item B10 (10-year Treasury bond): CRSP, monthly returns for 10-year Treasury
bond.
\item Market (broad equity market index): CRSP, monthly and daily returns,
including dividends and distributions, for a capitalization-weighted
index consisting of all domestic stocks trading on major US exchanges
(the VWD index).
\end{enumerate}
CRSP data was obtained for the historical time period 1926:01 to 2023:12.
All time series were inflation-adjusted using inflation data from
the US Bureau of Labor Statistics\footnote{The annual average CPI-U index, which is based on inflation data for
urban consumers, were used - see \texttt{http://www.bls.gov.cpi}}.

\subsection{Constructing VETF and LETF returns time series\label{subsec: Appendix Constructing-SETF-and LETF returns time series}}

LETFs were only introduced in 2006 (\cite{BansalMarshall2015}), whereas
the first VETFs were listed in the US in the 1990s. In order to obtain
longer time series of returns for indicative investment results of
Section \ref{sec:Indicative-investment-results}, we construct a proxy
returns time series for a VETF and LETF referencing a broad equity
market index as follows:
\begin{enumerate}
\item Obtain daily returns for the underlying broad equity market index
referenced by the VETF and LETF. For this purpose, we used daily returns
for the CRSP capitalization-weighted index consisting of all domestic
stocks trading on major US exchanges (the VWD index - see above),
with historical data that is available since January 1926. We prefer
to use a time series that is as long as possible, since this would
include additional periods of exceptional market volatility such as
1929-1933.
\item Multiply each daily return by the returns multiplier $\beta$, where
we used $\beta=2$ for the LETF and $\beta=1$ for the VETF, and construct
a time series of monthly returns.
\item Adjust the time series of VWD returns using (\ref{eq: Approx LETF dynamics ito S dynamics})
to reflect the ETF expense ratios $c_{k},k\in\left\{ v,\ell\right\} $
and the observed T-bill rate $r$. As per Table \ref{tab: Closed-form solns - Candidate assets and benchmark},
we assumed expense ratios of $c_{\ell}=0.89\%$ p.a. for the LETF
and $c_{v}=0.06\%$ p.a. for the VETF to reflect typical values observed
in the market.
\item Inflation-adjust the time series using inflation data from the US
Bureau of Labor Statistics.
\end{enumerate}

Although we can construct a synthetic LETF that achieves the desired leveraged returns by explicitly assuming borrowing at the risk-free rate plus a spread, in practice, LETFs are managed by large institutional fund managers who typically achieve leverage more efficiently through total return swaps. Due to their significant scale and high creditworthiness, these institutions can negotiate swap agreements priced at levels very close to the risk-free (T-bill) rate. Consequently, the effective borrowing cost embedded within the LETF fee remains minimal (e.g. captured by the T-bill rate and the expense ratio $c_{\ell}=0.89\%$), which is significantly lower than the direct borrowing spreads (e.g. $b = 3\%$) typically faced by other investors. This institutional advantage therefore justifies modelling the synthetic historical LETF returns as if borrowing effectively occurs at the T-bill rate with a modest LETF fee.

Note that a proxy time series of LETF returns is similarly constructed
in \cite{BansalMarshall2015}, although a number of details (such
as inflation adjustment of returns and choice of underlying index)
differ. As noted in \cite{BansalMarshall2015}, the construction of
such a proxy returns time series assumes that the ETF managers achieve
a negligible tracking error with respect to the underlying index.
We observe that these assumptions are often made out of necessity
in the literature concerning LETFs (e.g. \cite{BansalMarshall2015},
\cite{LeungSircar2015}). Furthermore, given improvements in designing
replication strategies for LETFs that remain robust even during periods
of market volatility (see for example \cite{GuasoniMayerhofer2023}),
this appears to be a reasonable assumption for ETFs written on major
stock market indices as considered in this paper.

We emphasize that the proxy time series for VETF and LETF returns
are only used for bootstrapping the data sets for the numerical solutions
implementing the data-driven neural network approach (Section \ref{sec:Numerical-solutions}
and Section \ref{sec:Indicative-investment-results}), and \textit{not}
for the closed-form solutions of Section \ref{sec:Closed-form-solutions}.
This follows since closed-form solutions in Section \ref{sec:Closed-form-solutions}
assume parametric dynamics for the underlying assets including the
broad equity market index, from which the LETF and VETF dynamics can
be constructed using (\ref{eq: LETF dynamics}) and (\ref{eq: VETF dynamics}),
respectively. 

\subsection{Calibrated parameters for closed-form solutions\label{sec: Appendix Calibrated-parameters-for closed-form} }

For the closed-form solutions of Section \ref{sec:Closed-form-solutions},
using the CRSP data for 30-day T-bills and the broad equity market
index (VWD index) for the period 1926:01 to 2023:12 as outlined above,
the filtering technique as per \cite{DangForsyth2016,ForsythVetzal2016}
for calibrating inflation-adjusted \cite{KouOriginal} jump-diffusion
processes resulted in the calibrated process parameters as presented
in Table \ref{tab:Calibrated asset parameters for analytical solutions}.
Given the specified dynamics (\ref{eq: B dynamics})-(\ref{eq: S dynamics})
of the risk-free asset $B$ and equity market index $S$ (with parameters
as in Table \ref{tab:Calibrated asset parameters for analytical solutions}),
we can obtain the dynamics of the LETF (\ref{eq: LETF dynamics})
and VETF (\ref{eq: VETF dynamics}).

\noindent 
\begin{table}[!hbt]
\caption{\label{tab:Calibrated asset parameters for analytical solutions}Closed-form
solutions: Calibrated, inflation-adjusted parameters for asset dynamics
(\ref{eq: B dynamics}) and (\ref{eq: S dynamics}), assuming the
\cite{KouOriginal} jump-diffusion model with $G\left(\xi^{s}\right)$
given by (\ref{eq: Jump pdf - KOU}). The calibration methodology
of \cite{DangForsyth2016,ForsythVetzal2016} is used with a jump threshold
parameter value of 3.}

\noindent \centering{}{\footnotesize{}}%
\begin{tabular}{|>{\raggedright}m{4cm}||>{\centering}p{0.9cm}|c|c|c|c|c|c|}
\hline 
\multirow{2}{4cm}{{\footnotesize{}Assumption for $S$-dynamics}} & \multicolumn{7}{c|}{{\footnotesize{}Calibrated parameters}}\tabularnewline
\cline{2-8} \cline{3-8} \cline{4-8} \cline{5-8} \cline{6-8} \cline{7-8} \cline{8-8} 
 & {\footnotesize{}$r$} & {\footnotesize{}$\mu$} & {\footnotesize{}$\sigma$} & {\footnotesize{}$\lambda$} & {\footnotesize{}$p_{up}$} & {\footnotesize{}$\eta_{1}$} & {\footnotesize{}$\eta_{2}$}\tabularnewline
\hline 
{\footnotesize{}Jump-diffusion (Kou model)} & {\footnotesize{}0.0031} & {\footnotesize{}0.0873} & {\footnotesize{}0.1477} & {\footnotesize{}0.3163} & {\footnotesize{}0.2258} & {\footnotesize{}4.3591} & {\footnotesize{}5.5337}\tabularnewline
\hline 
{\footnotesize{}GBM dynamics (no jumps)} & {\footnotesize{}0.0031} & {\footnotesize{}0.0819} & {\footnotesize{}0.1850} & {\footnotesize{}-} & {\footnotesize{}-} & {\footnotesize{}-} & {\footnotesize{}-}\tabularnewline
\hline 
\end{tabular}{\footnotesize\par}
\end{table}

Note that the values of the remaining parameters for the parametric
dynamics can be calculated by substituting the values of Table \ref{tab:Calibrated asset parameters for analytical solutions}
into  the results of Appendix \ref{subsec: Appendix expressions for kappas}.
This gives $\kappa_{1}^{s}=-0.0513$, $\kappa_{2}^{s}=0.0884$, $\kappa_{1}^{\ell}=-0.0500$,
$\kappa_{2}^{\ell}=0.0870$ and $\kappa_{\chi}^{\ell,s}=0.0876$.

{\PVSedit{
\section{Additional numerical results\label{sec:Appendix - Additional numerical results}}

Additional leverage and borrowing cost scenarios are analyzed as a
supplement to the results of Section \ref{sec:Indicative-investment-results}. The results of this appendix confirm that the conclusions of Section \ref{sec:Indicative-investment-results} do not appear to be sensitive in a qualitative sense to the specific maximum leverage or borrowing costs parameters used in the numerical
analysis, provided these values are within a reasonable range.

\subsection{Scenarios: Leverage and borrowing costs\label{subsec: Appendix scenarios leverage and borrowing costs}}
 We consider scenarios where the maximum leverage allowed decreases from
$p_{max}=1.5$ to $p_{max}=1.2$, or increases to $p_{max}=2.0$,
and where zero borrowing costs might be applicable (as opposed to
borrowing costs of $b=0.03$ throughout Section \ref{sec:Indicative-investment-results}).

As discussed in Section \ref{sec:Indicative-investment-results},
we observe that the IR-optimal portfolio of the LETF investor achieves
partial stochastic dominance over the corresponding IR-optimal portfolio
of the VETF investor with the same outperformance target $\gamma$,
even if the VETF investment can be leveraged 2x, provided borrowing premiums are applicable - see Figure \ref{fig: Appendix NumSolns Sc1 - W_T CDFs - Lev and borrowing costs}(b). However, when borrowing premiums drop to zero (i.e. the unrealistic scenario where the investor can borrow at the T-bill rate), a leveraged IR-optimal VETF-based strategy with $p_{max}=2.0$ performs similarly, though \textit{slightly} better, than the LETF-based strategy with no leverage (see Figure \ref{fig: Appendix NumSolns Sc1 - W_T CDFs - Lev and borrowing costs}(a)). Figure \ref{fig: Appendix NumSolns Sc1 - Ratio CDFs - Lev and borrowing costs} demonstrates that similar observations also hold if we consider benchmark outperformance instead of portfolio wealth. For a further explanation of the results of Figure \ref{fig: Appendix NumSolns Sc1 - W_T CDFs - Lev and borrowing costs} and Figure \ref{fig: Appendix NumSolns Sc1 - Ratio CDFs - Lev and borrowing costs}, see Appendix \ref{subsec: Appendix Bootstrapped quarterly returns}.

\noindent 
\begin{figure}[!tbh]
\noindent \begin{centering}
\subfloat[CDF of $W\left(T\right)$ - Zero borrowing costs]{\includegraphics[scale=0.75]{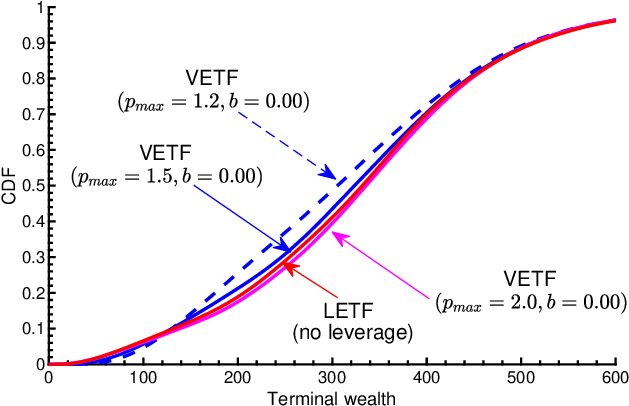}

}$\quad$$\quad$\subfloat[CDF of $W\left(T\right)$ - With borrowing costs]{\includegraphics[scale=0.75]{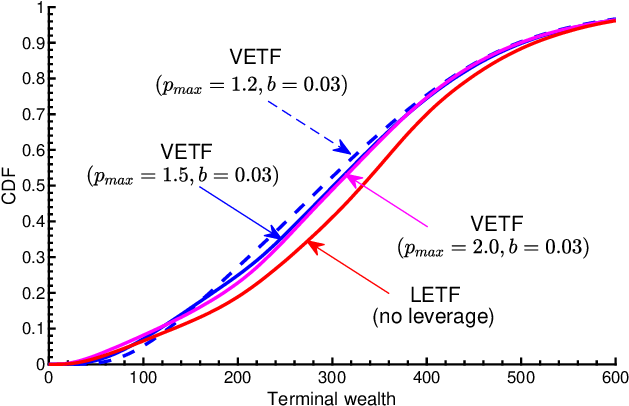}

}
\par\end{centering}
\caption{Effect of leverage and borrowing cost assumptions on CDFs of IR-optimal
terminal wealth $W_{k}^{\ast}\left(T\right),k\in\left\{ v,\ell\right\}$, with $T = 10$ years.\label{fig: Appendix NumSolns Sc1 - W_T CDFs - Lev and borrowing costs}}
\end{figure}

\noindent 
\begin{figure}[!tbh]
\noindent \begin{centering}
\subfloat[CDF of ratio $W^{\ast}\left(T\right)/\hat{W}\left(T\right)$ - Zero
borrowing costs]{\includegraphics[scale=0.75]{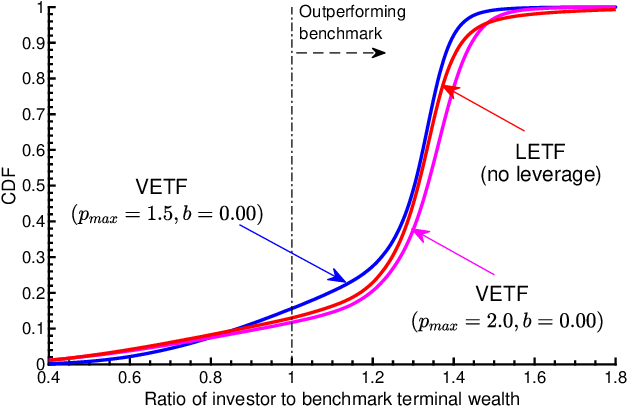}

}$\quad$$\quad$\subfloat[CDF of ratio $W\left(T\right)/\hat{W}\left(T\right)$ - With borrowing
costs]{\includegraphics[scale=0.75]{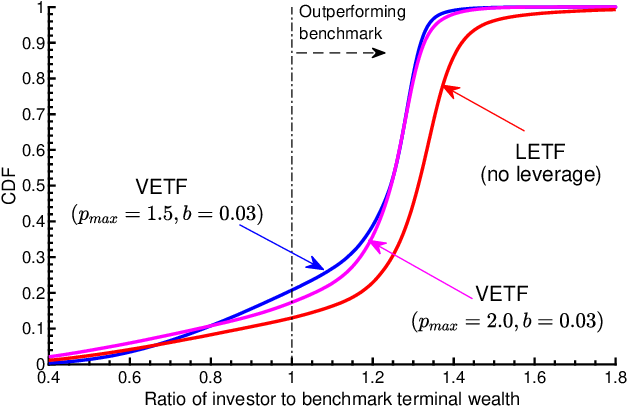}

}
\par\end{centering}
\caption{Effect of leverage and borrowing cost assumptions on the CDFs of IR-optimal
terminal wealth ratios $W_{k}^{\ast}\left(T\right)/\hat{W}\left(T\right),k\in\left\{ v,\ell\right\}$, with $T = 10$ years.
\label{fig: Appendix NumSolns Sc1 - Ratio CDFs - Lev and borrowing costs}}
\end{figure}
\noindent 
\begin{figure}[!tbh]
\noindent \begin{centering}
\subfloat[Zero costs]{\includegraphics[scale=0.75]{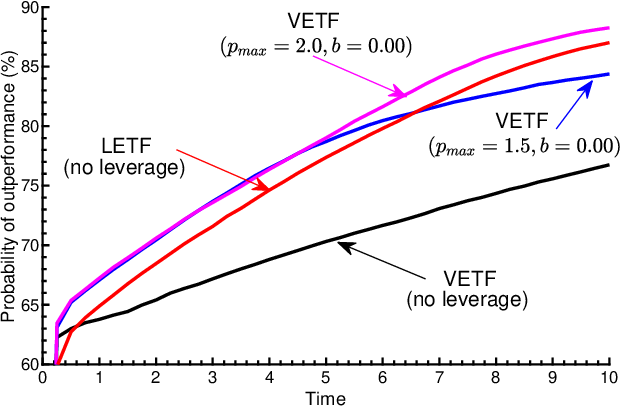}

}$\quad$$\quad$\subfloat[With costs]{\includegraphics[scale=0.75]{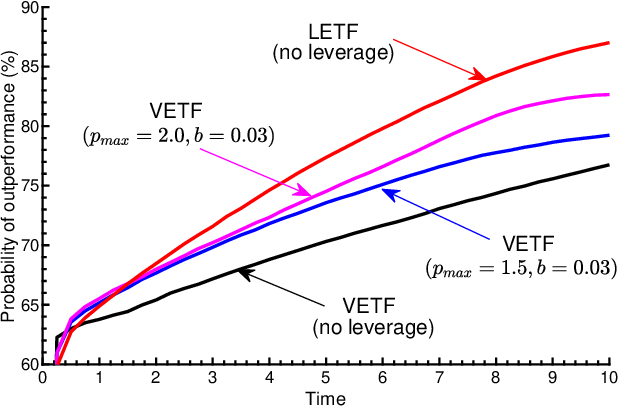}

}
\par\end{centering}
\caption{Effect of leverage and borrowing cost assumptions on the probability
$W_{k}^{\ast}\left(t\right)>\hat{W}\left(t\right),k\in\left\{ v,\ell\right\} $
of benchmark outperformance as a function of time $t$, over the time horizon of $T = 10$ years. \label{fig: Appendix NumSolns Sc1 BM outperformance - LEVERAGE and COSTS}}
\end{figure}

\subsection{Historical data: Bootstrapped quarterly returns\label{subsec: Appendix Bootstrapped quarterly returns} }

We demonstrate that the observations from Section \ref{subsec:Intuition:-lump-sum investment scenario}
(obtained using calibrated parametric models for the underlying assets) also
hold in the case of the historical returns used to obtain the results
of Section \ref{sec:Indicative-investment-results} and Appendix \ref{subsec: Appendix scenarios leverage and borrowing costs}.
Recall that in Section \ref{subsec:Intuition:-lump-sum investment scenario}
we observed that holding a LETF position for one quarter amounts to
holding a ``continuously rebalanced'' position in the equity index
and bonds, which results in a power law-type payoff from holding the
LETF.

Figure \ref{fig: NumSolns bootstrapped pathwise return differences_b0}
and Figure \ref{fig: NumSolns bootstrapped pathwise return differences_b3}
compare the pathwise quarterly returns of two simple strategies using
the bootstrapped historical data. The strategies consist of (i) investing
all wealth in the LETF at the start of a quarter, and (ii) investing
200\% of wealth in the VETF at the start of a quarter funded by borrowing
100\% of wealth at the T-bill rate plus a borrowing premium (where applicable). Outcomes are compared at the
end of the quarter with no intermediate trading. The only difference
between Figure \ref{fig: NumSolns bootstrapped pathwise return differences_b0}
and Figure \ref{fig: NumSolns bootstrapped pathwise return differences_b3}
is that Figure \ref{fig: NumSolns bootstrapped pathwise return differences_b0}
applies a zero borrowing premium ($b=0.00$) to fund short positions
in the VETF, whereas Figure \ref{fig: NumSolns bootstrapped pathwise return differences_b3}
applies a borrowing premium of $b=0.03$ to fund short positions in
the VETF.

Figure \ref{fig: NumSolns bootstrapped pathwise return differences_b0}(a)
and Figure \ref{fig: NumSolns bootstrapped pathwise return differences_b3}(a)
illustrate the quarterly returns of the simple strategies (y-axis)
for a given level of equity index quarterly return (x-axis) over the
quarter. Figure \ref{fig: NumSolns bootstrapped pathwise return differences_b0}(b)
and and Figure \ref{fig: NumSolns bootstrapped pathwise return differences_b3}(b)
illustrate the distribution of pathwise quarterly return differences
where, for each value of the x-axis in the corresponding Figure \ref{fig: NumSolns bootstrapped pathwise return differences_b0}(a)
and Figure \ref{fig: NumSolns bootstrapped pathwise return differences_b3}(a),
and therefore for a particular given path of (joint) asset returns,
we calculate the vertical difference between the returns of the two
strategies.

\noindent 
\begin{figure}[!tbh]
\noindent \begin{centering}
\subfloat[Pathwise payoffs of simple strategies]{\includegraphics[scale=0.75]{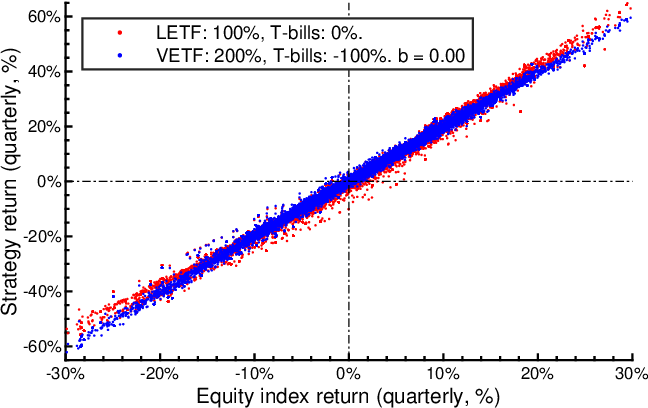}

}$\quad$$\quad$\subfloat[Distribution - pathwise return differences]{\includegraphics[scale=0.75]{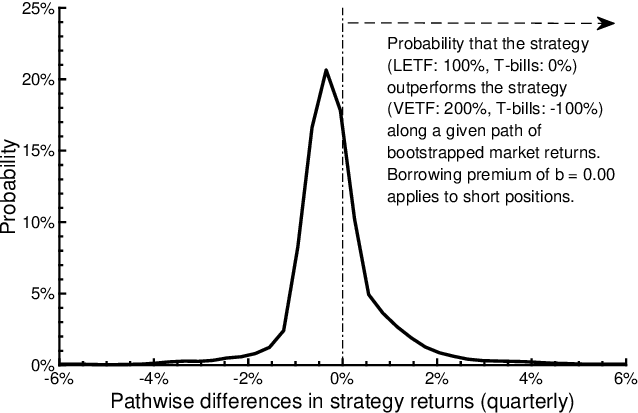}

}
\par\end{centering}
\caption{Zero borrowing premium ($b = 0.00$): Pathwise comparison of the bootstrapped historical quarterly inflation-adjusted returns of two simple strategies.See Appendix \ref{subsec: Appendix Bootstrapped quarterly returns} for more detail.  
\label{fig: NumSolns bootstrapped pathwise return differences_b0}}
\end{figure}
\noindent 
\begin{figure}[!tbh]
\noindent \begin{centering}
\subfloat[Pathwise payoffs of simple strategies]{\includegraphics[scale=0.75]{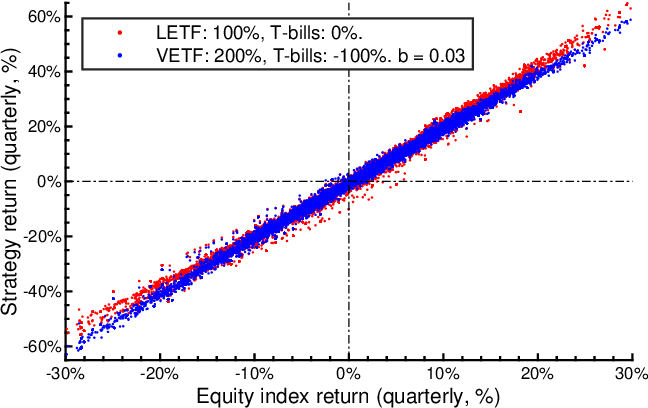}

}$\quad$$\quad$\subfloat[Distribution - pathwise return differences]{\includegraphics[scale=0.75]{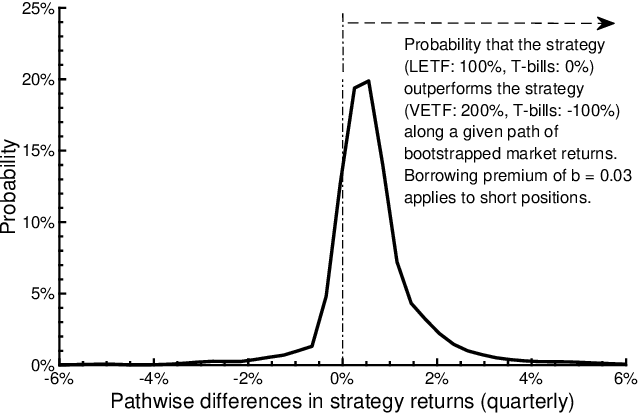}

}
\par\end{centering}
\caption{Positive borrowing premium ($b = 0.03$): Pathwise comparison of the bootstrapped historical quarterly inflation-adjusted returns of two simple strategies. See Appendix \ref{subsec: Appendix Bootstrapped quarterly returns} for more detail.
\label{fig: NumSolns bootstrapped pathwise return differences_b3}}
\end{figure}
Figure \ref{fig: NumSolns bootstrapped pathwise return differences_b0}(a)
and Figure \ref{fig: NumSolns bootstrapped pathwise return differences_b3}(a)
confirm that the call-like payoff of the LETF also holds in the bootstrapped
historical data. However, this is translated into a slight advantage for the LETF
relative to the 2x leveraged VETF strategy only when borrowing premiums
are positive (compare Figure \ref{fig: NumSolns bootstrapped pathwise return differences_b0}(b)
and Figure \ref{fig: NumSolns bootstrapped pathwise return differences_b3}(b)),
i.e. when the investor cannot fund short positions at only the T-bill
rate. 

Given these payoff structures of the LETF,
the IR-optimal LETF strategy responds to gains by reducing exposure
to the LETF, thus locking in the results of prior quarters of good
performance while reducing exposure to future possible losses by having
lower exposure to the LETF (see Section \ref{sec:Closed-form-solutions} and Section \ref{sec:Indicative-investment-results}). The compounding effect of applying the
contrarian IR-optimal investment strategy quarter after quarter given
returns as per Figure \ref{fig: NumSolns bootstrapped pathwise return differences_b0}(a)
and Figure \ref{fig: NumSolns bootstrapped pathwise return differences_b3}(b)
ultimately results in the terminal wealth results reported in Figures
\ref{fig: Appendix NumSolns Sc1 - W_T CDFs - Lev and borrowing costs},
\ref{fig: Appendix NumSolns Sc1 - Ratio CDFs - Lev and borrowing costs}
and \ref{fig: Appendix NumSolns Sc1 BM outperformance - LEVERAGE and COSTS}.
}}

\end{appendices}

\end{document}